\documentclass[12pt,]{article}
\usepackage{lmodern}
\usepackage{amssymb,amsmath}
\usepackage{ifxetex,ifluatex}
\usepackage{adjustbox} % adjust the size of table
\usepackage{mathtools}
\usepackage{threeparttable}
\usepackage{tabularx}
\usepackage{rotating}
\usepackage{fixltx2e} % provides \textsubscript
\ifnum 0\ifxetex 1\fi\ifluatex 1\fi=0 % if pdftex
  \usepackage[T1]{fontenc}
  \usepackage[utf8]{inputenc}
\else % if luatex or xelatex
  \ifxetex
    \usepackage{mathspec}
  \else
    \usepackage{fontspec}
  \fi
  \defaultfontfeatures{Ligatures=TeX,Scale=MatchLowercase}
\fi
\IfFileExists{upquote.sty}{\usepackage{upquote}}{}
\IfFileExists{microtype.sty}{%
\usepackage[]{microtype}
\UseMicrotypeSet[protrusion]{basicmath} % disable protrusion for tt fonts
}{}
\usepackage[colorlinks = true,
            linkcolor = blue,
            urlcolor  = blue,
            citecolor = blue,
            anchorcolor = blue]{hyperref}
\PassOptionsToPackage{hyphens}{url} % url is loaded by hyperref
\hypersetup{
            pdftitle={Generalizing Trimming Bounds},
            pdfauthor={Cyrus Samii; Ye Wang; Junlong Aaron Zhou},
            pdfborder={0 0 0},
            breaklinks=true}
\usepackage[margin=1in]{geometry}
\usepackage{graphicx,grffile} 
\makeatletter
\def\maxwidth{\ifdim\Gin@nat@width>\linewidth\linewidth\else\Gin@nat@width\fi}
\def\maxheight{\ifdim\Gin@nat@height>\textheight\textheight\else\Gin@nat@height\fi}
\makeatother
\setkeys{Gin}{width=\maxwidth,height=\maxheight,keepaspectratio}
\ifx\paragraph\undefined\else
\let\oldparagraph\paragraph
\renewcommand{\paragraph}[1]{\oldparagraph{#1}\mbox{}}
\fi
\ifx\subparagraph\undefined\else
\let\oldsubparagraph\subparagraph
\renewcommand{\subparagraph}[1]{\oldsubparagraph{#1}\mbox{}}
\fi
\newcommand\independent{\protect\mathpalette{\protect\independenT}{\perp}}
\def\independenT#1#2{\mathrel{\rlap{$#1#2$}\mkern2mu{#1#2}}}

\makeatletter
\def\fps@figure{htbp}
\makeatother

\usepackage{setspace}
\usepackage{ragged2e,subcaption,caption}
\usepackage[linesnumbered,ruled,vlined]{algorithm2e}
\usepackage{longtable,booktabs}

\usepackage{tikz}
\usetikzlibrary{positioning,shapes.geometric,chains,calc,shapes}

\newtheorem{assn}{Assumption}

\usepackage{natbib}

\usepackage{comment} 

\title{Generalizing Trimming Bounds for Endogenously Missing Outcome Data Using Random Forests\thanks{Cyrus Samii, Department of Politics, New York University (Email: \url{cds2083@nyu.edu}); Ye Wang, Department of Political Science, University of North Carolina at Chapel Hill (Email: \url{yewang@unc.edu}); Junlong Aaron Zhou, Tencent America (Email: \url{jlzhou@nyu.edu})}}
\author{Cyrus Samii \and Ye Wang \and Junlong Aaron Zhou}
\providecommand{\institute}[1]{}
\date{\today}
\begin{document}
\maketitle
\begin{abstract}
In many experimental or quasi-experimental studies, outcomes of interest are only observed for subjects who select (or are selected) to engage in the activity generating the outcome.
Outcome data is thus endogenously missing for units who do not engage, in which case random or conditionally random treatment assignment prior to such choices is insufficient to point identify treatment effects.
Non-parametric partial identification bounds are a way to address endogenous missingness without having to make disputable parametric assumptions. 
Basic bounding approaches often yield bounds that are very wide and therefore minimally informative. 
We present methods for narrowing non-parametric bounds on treatment effects by adjusting for potentially large numbers of covariates, working with generalized random forests. 
Our approach allows for agnosticism about the data-generating process and honest inference.
We use a simulation study and two replication exercises to demonstrate the benefits of our approach.

\end{abstract}

\newpage

\section{Introduction}\label{introduction}
%Frame in terms of selection problems: treatment assignment or missing data.
%Examples of experimental attrition, missing outcome data, intensive margin, endogenous selection of control condition.
Experiments and quasi-experiments often face endogenously missing outcome data, in which case random or conditionally random treatment assignment is insufficient to point identify causal effects.
Consider the experiment in \citet{santoro2022promise}, in which subjects were asked to have an online conversation about their ``perfect day'' with someone from a different political party (an outpartisan). 
In one group, subjects were informed that their conversation partner was an outpartisan, while in the other group, no such information about their partner was given. 
Examining only the complete data, the authors found that knowingly talking to an outpartisan was associated with more warmth toward outparty voters in the post-treatment survey. 
Nevertheless, among the 986 subjects who entered the treatment assignment stage, only 478 completed the conversation and 469 went on to complete the post-treatment survey.
The completion rate was higher among those informed about their partner's outpartisan status. 
For people with a strong interest in politics, a conversation with a ``random person'' about their perfect day may seem uninteresting while a conversation with an outpartisan might be intriguing. 
At the same time, some people may have a sense of obligation such that they do not drop out if they find the conversation uninteresting.
Missingness based on such differences in subjects' types would leave us with an imbalanced comparison across treatment groups.

Political scientists encounter endogenously missing outcome data all the time. In numerous applications, a subject's outcome may only be revealed if that subject makes certain choices. Furthermore, the treatment typically creates asymmetric choice conditions for the subjects, leading to what \citet{slough2022phantom} describes as the ``phantom counterfactual'' problem. In the aforementioned example, knowing the partner's identity may generate two distinct effects: it can reduce the likelihood of completing a conversation while, among those who do complete a conversation, enhancing the subject's affability towards outpartisans.  These are known as the ``extensive margin effect'' and the ``intensive margin effect'' respectively in the literature \citep{staub2014causal, kim2019effects, gulzar2020does, paulsen2023foundations}. If the conversation was not completed, the outcome for the subject would remain undefined (a ``phantom''). Among subjects who completed the conversation, those informed of their partner's partisanship may have a higher degree of political interest as compared to the group completing the conversation in the control group. 
Consequently, the variation in the average outcome across the two groups is indicative of not only the treatment's impact but also differences attributable to heterogeneity in levels of political interest. 
As elaborated by \citet{montgomery2018conditioning}, limiting the analysis to observations with non-missing outcomes implies conditioning on a post-treatment variable (the response indicator), introducing collider biases. If the degree of political interest is unobservable to researchers, common tools for imputation or adjustment are not justified \citep{honaker2011amelia, li2013weighting, blackwell2017unified, liu2021latent}.

One solution is to forgo point identification and rather construct {\it bounds} on treatment effects \citep{molinari2020microeconometrics}. 
The bounds enclose a set of effect values that are consistent with patterns of missingness (``identified set''). 
We focus on ``trimming bounds'' methods based on the seminal paper by \citet{lee2009training}, which leverages a monotonicity assumption to bound the average causal effect for those who would always have their outcomes observed whether under treatment or control (``always-responders''). 
In the example of \citet{santoro2022promise}, it would characterize effects among those who are always willing to engage in an online conversation. 
From a policy perspective, this subgroup defines those for which the intervention raises the fewest red flags in terms of forcing people into an exercise that they find aversive. 
It is also the subgroup for which the intensive margin effect is most clearly defined \citep{staub2014causal, slough2022phantom}.

The logic of trimming bounds is very simple. 
Suppose those who respond in the control group would always also respond if they had been assigned to the treatment group (monotonic selection). 
Then the control group consists only of ``always-responders.''
The treatment group is a mixture of always-responders and those who respond only if treated (``compliers''). 
The difference in response rates under treatment versus control measures the share of compliers. We can bound the effect for always-responders by taking the means of trimmed versions of the treated outcome distribution: a lower bound comes from trimming off the share of compliers from the top of the distribution, and an upper bound comes from trimming this share from the bottom.

Unfortunately, these basic trimming bounds can be very wide. We develop an approach to tighten them by extracting information from pre-treatment covariates with the ``generalized random forest'' (known as \textit{grf}), a machine learning algorithm developed in \citet{wager2018estimation} and \citet{athey2019generalized}. The algorithm approximates the local moment condition for any quantity with a collection of \textit{regression tree}s \citep{montgomery2018tree}. 
We use the random forest for precise estimation of ``nuisance parameters,'' such as the conditional trimming probability.
We then construct {\it covariate-tightened} trimming bounds that are assured to be no wider (in expectation), and can often be substantially narrower, than the basic ones. 
To avoid the regularization bias from the machine learning step \citep{chernozhukov2018double}, we tailor the algorithm by incorporating both Neyman orthogonalization and cross-fitting \citep{chernozhukov2017double, ratkovic2021relaxing}. 
We show that the bounds estimators are consistent and approximately Normal in large samples, and we provide approaches for constructing valid confidence intervals.

In both simulations and applications, we show that our proposed method works significantly better than the basic trimming bounds and other similar proposals \citep{olma2020nonparametric, semenova2020better}. 
It controls the covariates flexibly even when they are high-dimensional, without imposing strong restrictions on model specification. 
%Hence, the estimated bounds can be quite informative about the possible magnitude of the treatment effects. In the simulation, we observe pronounced shrinkage in the width of the identified set. The asymptotic bias of the estimates are found to be small and the coverage rate reaches the nominal level under moderate sample sizes. In a replication of \citet{santoro2022promise}, the estimated lower bound is significantly larger than zero, confirming the original finding of the study even when the attrition rate is higher than 50\%. 
The method can be generalized unit-level missingness, monotonicity holding only conditionally, binary outcome variables, or the probability of being treated needing to be estimated. 
We demonstrate performance in both the experimental study of \citet{santoro2022promise} and an observational study by \citet{blattman2010consequences}.

With endogenously missing outcome data, trimming bounds are an alternative for those uncomfortable with relying on either parametric assumptions or untestable exclusion restrictions, as needed in selection modeling methods based on \citet{heckman1979sample}.
Our work is in line with other work in political methodology seeking to relax disputable modeling assumptions by using machine learning \citep{blackwell2022reducing, ratkovic2021relaxing} and partial identification \citep{knox2020administrative, duarte2021automated}. It enables empirical researchers to examine the potential range of the causal effects when the analysis has to be conditioned on variables that might be affected by the treatment \citep{montgomery2018conditioning, blackwell2023priming}.

The rest of the paper is organized as follows: Section \ref{set-up} formally describes the problem we try to address and introduces the estimands. Sections \ref{tb} and \ref{ctb} illustrate the basic idea of the trimming bounds method and how they can be tightened by using information from covariates. Section \ref{estimation} presents our estimation strategy, including the algorithm and its statistical properties. Section \ref{extension} discusses several extensions. Sections \ref{simulation} and \ref{application} offer evidence from simulation and applications, respectively. Section \ref{conclusion} concludes.

\section{Setup}\label{set-up}
We consider a randomized or natural experiment with subjects indexed by \(i=1,...,N\), for which we also have measured \(P\) pre-treatment covariates collected in the vector \(\mathbf{X}_i = (X_{i1}, X_{i2}, \dots, X_{iP})\). For each subject, we always observe values of the
covariates\footnote{We discuss unit-level missingness in Section \ref{missing-covariates}.}, the treatment status \(D_i \in\{0,1\}\), and the response indicator \(S_i \in\{0,1\}\). The realized outcome \(Y_i\) is observed only when \(S_i = 1\). For $D_i = 1$ and $D_i=0$, respectively, potential outcomes are given by $(Y_i(1), Y_i(0))$ and potential response indicators are given by $(S_i(1), S_i(0))$.  The realized outcome is $Y_i = D_iY_i(1) + (1-D_i)Y_i(0)$ and the realized response is $S_i = D_iS_i(1) + (1-D_i)S_i(0)$. Define $U_i$ to be unobserved factors that affect both $(S_i(1),S_i(0))$ and $(Y_i(1),Y_i(0))$.

We make the following assumption on treatment assignment:
\begin{assn}
Strong ignorability:  
\begin{center}
    $(Y_i(1), Y_i(0), S_i(1), S_i(0)) \independent D_i | \mathbf{X}_i$, \\
    $\varepsilon < P(D_i = 1 | \mathbf{X}_i) < 1 - \varepsilon, \text{ with } \varepsilon > 0.$
\end{center}
 \label{ass:ind}
\end{assn}

The first part states that the treatment is conditionally independent to both the potential outcomes and the potential responses. The second part requires that the propensity score, $p(\mathbf{X}_i) = P(D_i = 1 | \mathbf{X}_i)$, is strictly bounded between 0 and 1.\footnote{We assume that $p(\mathbf{X}_i)$ is known to the researcher so far and discuss estimated propensity scores in Section \ref{est-pscore}.} 
The assumption holds in experiments when treatment is randomly or conditionally randomly assigned.
In observational studies, the assumption implies that the covariate vector $\mathbf{X}_i$ includes all confounders.

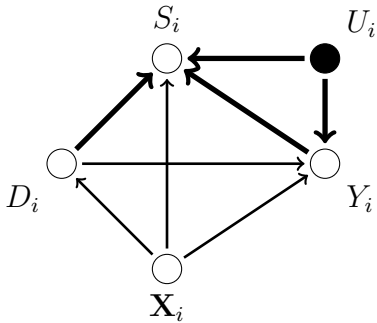
\begin{figure}[!t]
\caption{Endogenous Missingness\label{fg.dag}}
\begin{center}
\begin{tikzpicture}
[scale=0.7,dot/.style={fill,draw,circle,minimum width=1pt},
arrow style/.style={->,line width=1pt, shorten <=2pt,shorten >=2pt },
arrow2 style/.style={->,line width=2pt, shorten <=2pt,shorten >=2pt },]
\node [fill=white, dot, label=above: $S_{i}$] (s) at (3,3) {}; 
\node [fill=white, dot, label=below right: $Y_{i}$] (y) at (6,1) {} edge [arrow2 style] (s);
\node [fill=white, dot, label=below left: $D_{i}$] (d) at (1,1) {} edge [arrow style] (y) edge [arrow2 style] (s);
\node [fill=black, dot, label=above right: $U_{i}$] (u) at (6,3) {} edge [arrow2 style] (y) edge [arrow2 style] (s);
\node [fill=white, dot, label=below: $\mathbf{X}_{i}$] (x) at (3,-1) {} edge [arrow style] (y) edge [arrow style] (s) edge [arrow style] (d);
\end{tikzpicture}
\end{center}
\par
\footnotesize\textbf{Note:} The above figure shows a directed acyclic graph (DAG) representing an endogenous missingness mechanism. White nodes represent observable variables and black nodes represent unobservable ones. An arrow from one node to the other (a ``path'') marks the causal relationship from the former to the latter.
The heavier arrows indicate the relationships that generate the bias from conditioning on $S_i$.
\end{figure}

The directed acyclic graph (DAG) in Figure~\ref{fg.dag} shows relationships that our setting admits. For the \citet{santoro2022promise} study, $U_i$ could represent the degree of political interest. 
Our approach also admits the possibility that $Y_i$ directly affects missingness.
If we only use data from subjects whose outcome is observed ($S_i = 1$), 
this amounts to conditioning on $S_i$, which is a ``collider'' between the treatment $D_i$ and then either the unobserved $U_i$ or the outcome $Y_i$ \citep{elwert2014endogenous, montgomery2018conditioning}.
For example, suppose that for some units, $D_i$ and $U_i$ negatively affect response. 
For these units, high values of $D_i$ (i.e. being treated) would tend to require low values of $U_i$ for $S_i = 1$ to hold. 
The lower $U_i$ values would affect the distribution of $Y_i$ values.
Correlations induced by conditioning on $S_i = 1$ confound our ability to estimate the causal effect of $D_i$ on $Y_i$.   
Such induced correlations result in bias.
% For the \citet{santoro2022promise} study, $U_i$ could represent aversion to remaining in the study after interacting with an outpartisan.
% For the hypothetical taxation study, $U_i$ could represent factors determining the amount that a firm would pay if registered, which affects whether the firm registers. 

To further understand the source of the bias, we classify subjects into four different types based on their responses to the treatment: always-responders ($S_i(0) = S_i(1) = 1$), never-responders ($S_i(0) = S_i(1) = 0$), compliers ($S_i(0) = 0, S_i(1) = 1$), and defiers ($S_i(0) = 1, S_i(1) = 0$).  
These subgroups are examples of what \citet{frangakis_rubin2002} call ``principal strata,'' similar to those in the setting with instrumental variables \citep{angrist_etal96_late}. 
From Table~\ref{tab_ps} Panel 1, we see that among units with observed outcomes ($S_i=1$), the control group may consist of both defiers and always-responders, while the treatment group is a mixture of compliers and always-responders. Therefore, the group-mean difference in the observed outcome does not capture an average of treatment effects for a fixed sub-population.

\begin{table}[!t]
\centering
\caption{Principal strata} \label{tab_ps}
\begin{center}
    Panel 1. Full set of principal strata
\end{center}
\begin{tabular}[]{@{}lll@{}}
\hline \\[-1.8ex]  
Variable & \(S_i = 0\) & \(S_i = 1\)\tabularnewline
\hline \\[-1.8ex]  
\(D_i = 0\) & Never-responders, Compliers  & Defiers, Always-responders \tabularnewline
\(D_i = 1\) & Never-responders, Defiers  & Compliers, Always-responders \tabularnewline
\hline \\[-1.8ex]  
\end{tabular}
\begin{center}
    Panel 2. Under monotonic selection
\end{center}
\begin{tabular}[]{@{}lll@{}}
\hline \\[-1.8ex]  
Variable & \(S_i = 0\) & \(S_i = 1\)\tabularnewline
\hline \\[-1.8ex]  
\(D_i = 0\) & Never-responders, Compliers  & Always-responders \tabularnewline
\(D_i = 1\) & Never-responders & Compliers, Always-responders \tabularnewline
\hline \\[-1.8ex]  
\end{tabular}
\end{table}

From the decomposition in Table~\ref{tab_ps} Panel 1, we cannot identify the share of units falling into each principal stratum from data. Following \citet{lee2009training}, we work with the following assumption on the selection process:
\begin{assn}
Monotonic selection: 
\begin{center}
    $S_i(1) \ge S_i(0)$
\end{center}
for any $i$. 
\label{ass:mono}
\end{assn}
The assumption requires that under treatment, a unit is at least as likely to respond as under control, thus it excludes the existence of defiers. 
This is shown in Table~\ref{tab_ps} Panel 2.
The assumption is similar to the ``no defiers'' assumption for instrumental variables.
It should be motivated by arguments about the choice behavior of agents determining response \citep{slough2022phantom}.
We could alternatively assume the opposite: $S_i(1) \le S_i(0)$ for any $i$. Then, we have defiers but not compliers.
Either form of monotonicity allows one to form trimming bounds.
In the \citet{santoro2022promise} example, subjects who would engage when uninformed of their partners' partisanship would be assumed to engage when informed.
In actuality, our re-analysis below suggests that monotonicity may not hold unconditionally in this example.
For the moment, we will maintain the monotonicity assumption, and then in Section \ref{cond-mono} we will relax it to allow for monotonicity that is conditional on covariates \citep{semenova2020better}.
That section also discusses how to test Assumption \ref{ass:mono}. 
Parametric and semi-parametric selection models implicitly impose a monotonicity assumption as part of the first-stage regression specification; the manner in which we state monotonicity here is a non-parametric generalization.\footnote{See \citet{vytlacil2002independence} for a discussion in the context of instrumental variables.} 

In Table \ref{tab:exp}, we demonstrate how to fit empirical examples from different fields of political science (American politics, comparative politics, public administration, international relations, political economy, and methodology) into our setup, including the implication of Assumption \ref{ass:mono} and the sub-population captured by always-responders in each context.
We acknowledge that in some cases monotonicity may only make sense conditional on certain background characteristics; again, we hold off on discussing methods to allow this until  Section \ref{cond-mono}.

\begin{table}
\centering
  \scriptsize
  \caption{Empirical Examples} 
 \label{tab:exp}  
   \begin{threeparttable}
\begin{tabular}{@{\extracolsep{5pt}} llll} 
\\[-1.8ex]\hline 
\hline \\[-1.8ex]  
\\[-1.8ex] 
Application & \citet{santoro2022promise} & \citet{blattman2010consequences} & \citet{knox2020administrative} \\
 & & &  \\
$Y_i$ & Warmth toward & Education/distress level & Police use of force \\
 & outpartisan voters &  & in encounter  \\
 & &  &  \\
 & & &  \\
$D_i = 1$ & Knowing partisanship of the & Being kidnapped & Citizen is minority  \\
 & partner in conversation & in conflict &  \\
 &  & & \\
 & & & \\
$S_i = 1$ & Conversation happened and & Survived and found & Citizen is under arrest \\
 &  outcome data obtained & by researchers & \\
 &  & & \\
 & & & \\
Mono. selection & Knowing partners' partisanship  & Avoiding kidnapping does  & No encounters in which \\
& does not cause subjects  &  not cause death or & minority citizens less likely \\
& to drop out & disappearance &  to be arrested \\
& & & \\
Always-responders & Subjects who will complete & Subjects who will survive & Encounters in which citizen \\
& conversation regardless of & and be found even if & will be arrested regardless \\
& knowledge about partner &  kidnapped &  of race of citizen \\

&  & & \\
\hline \\[-1.8ex]  
\\[-1.8ex] 
Application & \citet{spilker2018trade} & \citet{hall2019wealth} & \citet{cheema2023canvassing} \\  %\cite{blackwell2023priming} \\ 
 & & &  \\
$Y_i$ & Number of products & Having descendants who & Proportion of women  \\ % Outcome for subjects \\ 
 & traded by a firm & fought for the South & who voted in a \\ % with $S_i = 1$ when $S_i$ is  \\ 
 & & in the Civil War & household  \\ % measured after treatment \\ 
 & & &  \\
$D_i = 1$ & Firm is under a & Won the Georgia & Household was canvassed \\ % Subject received the \\ 
 & trade agreement & land lottery and & before election \\
 &  & own more slaves &  \\
 & & & \\
$S_i = 1$ & Firm is an exporter & Having descendants & Household can be tracked \\ % The moderator $S_i$ takes \\
 & &  &  in the post-election survey \\ % value $1$ \\ % 
 & & & \\
Mono. selection & Trade agreements do not & Winning lottery does not & Being canvassed makes \\ % $S_i$ is more (less) likely \\ % 
& cause firm not to export & cause anyone to have  & household more (less) likely  \\ % to be $1$ under treatment \\ % 
& & fewer descendants & to be tracked  \\ %   \\
& & &  \\
Always-responders & Firms that export even & Households that will have & Households that can always  \\ % Subjects for whom $S_i = 1$ \\ 
& without trade agreement & descendants regardless of  & be tracked \\ % regardless of treatment \\ % 
& &  lottery & \\ % status \\% 
& & & \\
 \hline \\[-1.8ex] 
\end{tabular}%
    \begin{tablenotes}
      \small
      \item \textit{Note:} This table summarizes how examples from different fields of political science fit into the setup developed here.
    \end{tablenotes}
\end{threeparttable}
\end{table}

\section{Trimming bounds}\label{tb}
\citet{lee2009training} shows that we can use ``trimming'' to bound the average treatment effect for always-responders, \(\tau(1,1) \coloneqq E[Y_i(1)-Y_i(0)|S_i(0)=S_i(1)=1]\), under Assumptions \ref{ass:ind} and \ref{ass:mono}. 
The focus on always-responders has substantive motivation.
It is the subset of the population that would not withdraw (or be withdrawn) from the intervention under prevailing circumstances.
It is also the subset of the population for which intensive margin effects are well defined (as in Lee's original application).\footnote{Table \ref{tab:exp} lists what the always-responders refer to in applications.}

To lighten the notation in the discussion that follows, suppose that we have unconditionally random assignment. 
For the more general analysis under weak ignorability (Assumption~\ref{ass:ind}), one would simply switch the expectations to include iterated expectations that condition on and then marginalize over $\mathbf{X}_i$.  
Given this notational convenience, define the average treatment effect for always-responders as follows:
\begin{align*}
\tau(1,1) & := E[Y_{i}(1) | S_{i}(1) = S_{i}(0) = 1] - E[Y_{i}(0) | S_{i}(1) = S_{i}(0) = 1] \\
& = E[Y_{i} | D_i = 1, S_{i}(1) = S_{i}(0) = 1] - E[Y_{i} | D_i = 0, S_i = 1].
\end{align*}
Given random assignment and monotonicity, the mean of observed outcomes in the observed control identifies the always-responders mean under control.
Those who respond in the treatment group can belong to either always-responders or compliers, and so the observed mean does not identify the always-responders mean under treatment.

Recall that Assumption~\ref{ass:ind} implies that the shares of units in each principal stratum are balanced (in expectation) across treatment and control.
Under Assumption~\ref{ass:mono}, the rate at which outcomes are observed in the control group, $\text{Pr}[S_i = 1 \mid D_i = 0]$, identifies the share of units in the overall sample that are always-responders.
The rate at which outcomes are observed in the treatment group, $\text{Pr}[S_i = 1 \mid D_i = 1]$, identifies the share of units in the overall sample that are either compliers or always-responders.
If we let $q$ denote the share of treated units {\it with observed outcomes} that are always-responders, then we have
$$
q = \frac{\text{Pr}[S_i = 1 \mid D_i = 0]}{\text{Pr}[S_i = 1 \mid D_i = 1]}.
$$
The mean of observed outcomes in the treatment group can be written as a mixture of the always-responders mean (with weight $q$) and the compliers mean (with weight $1-q$):
\begin{align*}
E[Y_{i} | D_i = 1, S_i = 1] = & q E[Y_{i} | D_i = 1, S_{i}(1) = S_{i}(0) = 1] \\
& + (1-q)E[Y_{i} | D_i = 1, S_{i}(1) = 0, S_{i}(0) = 1].
\end{align*}

The quantity $E[Y_{i} | D_i = 1, S_{i}(1) = S_{i}(0) = 1]$ is not identified under random assignment and monotonicity. 
But we can construct sharp upper and lower bounds that surround it by trimming the lower and upper tails of the observed treated outcome distribution by the share of compliers, thus retaining the portion with mass $q$ \citep[Prop. 1a]{lee2009training}. As shown in Figure \ref{fig_leebounds}, in the worst case, always-responders take the bottom (left-most) $q$ share of the treated outcome's distribution. Hence, the average treated outcome of always-responders is bounded from below by $Y(1)$'s expectation over the shaded area in the left plot. Similarly, the average treated outcome for always-responders is bounded from above by $Y(1)$'s expectation over the shaded area in the right plot. Mathematically, we have 

\begin{figure}[!t]
\caption{Basic trimming bounds in \citet{lee2009training}}
\label{fig_leebounds}
 \begin{center}
\includegraphics[width=.49\linewidth, height=.5\linewidth]{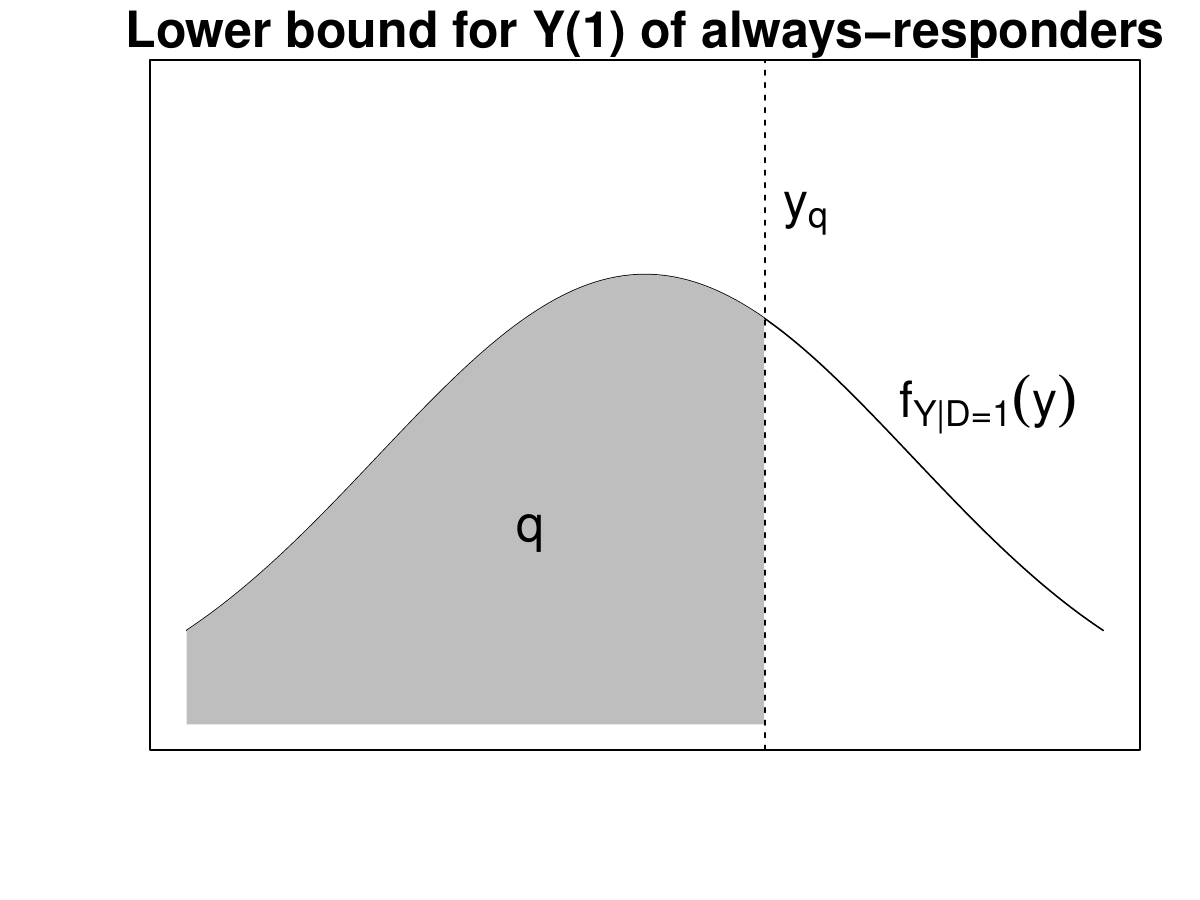}
\includegraphics[width=.49\linewidth, height=.5\linewidth]{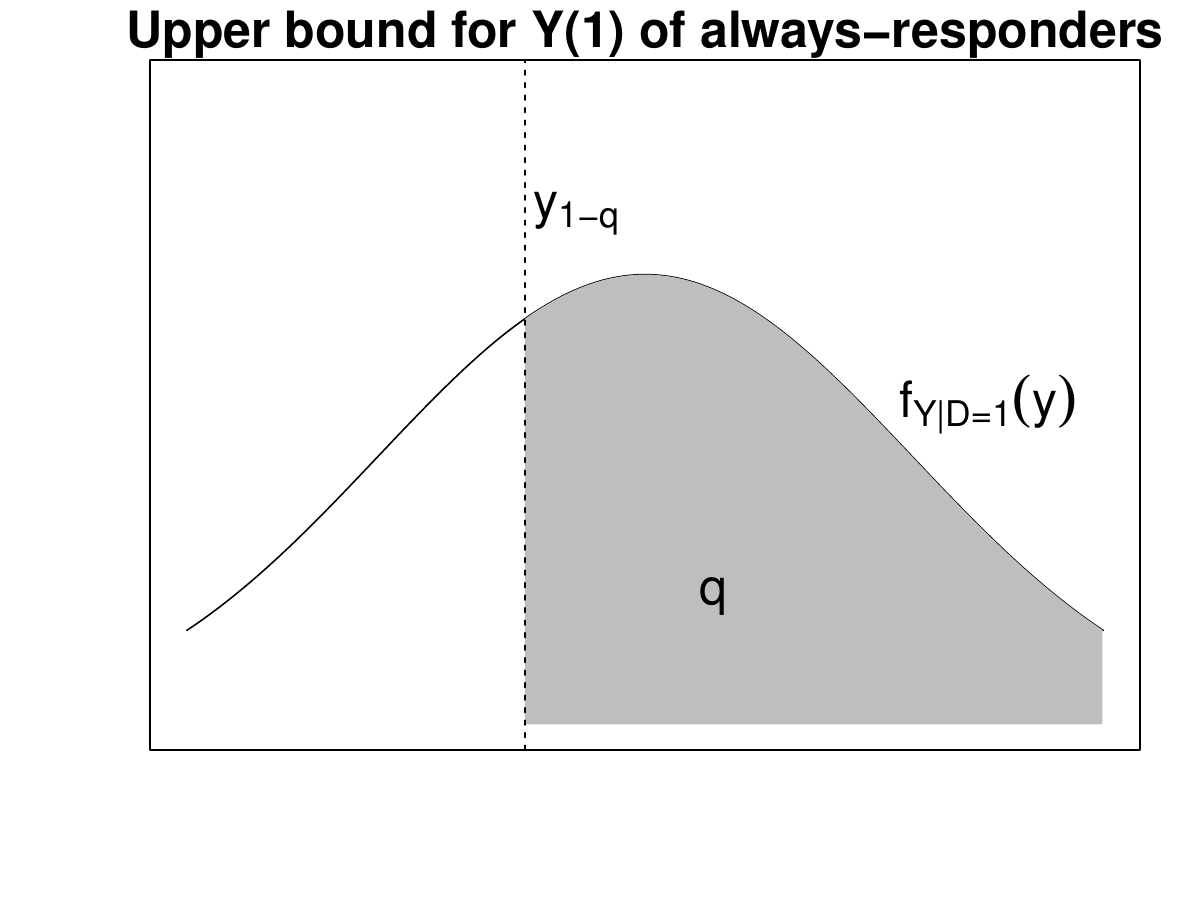} 
 \end{center}
\end{figure}

\[
\begin{aligned}
E[Y_{i} | D_i = 1, S_i = 1, Y_i \leq y_{q}] & \leq E[Y_{i} | D_i = 1, S_{i}(1) = S_{i}(0) = 1] \\
& \leq E[Y_{i} | D_i = 1, S_i = 1, Y_i \geq y_{1-q}],
\end{aligned}
\] where \(y_{q}\) is the $q$th-quantile of the observed outcome's conditional distribution in the treatment group that satisfies the condition
\(\int_{-\infty}^{y_{q}} dF_{Y|D=1, S=1}(y) = q\). Here $F_{Y|D=1, S=1}(\cdot)$ is the distribution function for observed treatment group outcomes.
The quantile \(y_{1-q}\) is similarly defined. We denote $E[Y_{i} | D_i = 1, S_i = 1, Y_i \leq y_{q}] - E[Y_{i} | D_i = 0, S_i = 1]$ as $\tau^L_{TB}(1, 1)$ and $E[Y_{i} | D_i = 1, S_i = 1, Y_i \geq y_{q}] - E[Y_{i} | D_i = 0, S_i = 1]$ as $\tau^U_{TB}(1, 1)$, where TB is short for trimming bounds. Clearly, 
$$
\tau^L_{TB}(1, 1) \leq \tau(1, 1) \leq \tau^U_{TB}(1, 1).
$$
The interval $[\tau^L_{TB}(1, 1), \tau^U_{TB}(1, 1)]$ is the identified set for the always-responder treatment effect.

\section{Covariate-tightened trimming bounds}\label{ctb}

Basic trimming bounds that do not incorporate covariate information can be very wide. \citet[Prop. 1b]{lee2009training} shows that under Assumptions \ref{ass:ind} and \ref{ass:mono}, we can use covariates to tighten the bounds for more precise inference. For \(\mathbf{X}_i = \mathbf{x}\), consider the always-responders' conditional average treatment effect,
\(\tau_{\mathbf{x}}(1,1) \coloneqq E[Y_{i}(1) - Y_{i}(0) | S_{i}(1) = S_{i}(0) = 1, \mathbf{X}_i = \mathbf{x}]\). We can derive the following bounds using the same logic as above:
\[
\begin{aligned}
E[Y_{i} | D_i = 1, S_i = 1, Y_i \leq y_{q(\mathbf{x})}(\mathbf{x}), \mathbf{X}_i = \mathbf{x}] & \leq E[Y_{i} | D_i = 1, S_{i}(1) = S_{i}(0) = 1, \mathbf{X}_i = \mathbf{x}] \\
& \leq E[Y_{i} | D_i = 1, S_i = 1, Y_i \geq y_{1-q(\mathbf{x})}(\mathbf{x}), \mathbf{X}_i = \mathbf{x}]
\end{aligned}
\]
where $q(\mathbf{x})$ is the conditional proportion of always-responders among units with observed outcomes in the treatment group; \(y_{q(\mathbf{x})}(\mathbf{x})\) and \(y_{1-q(\mathbf{x})}(\mathbf{x})\) are the conditional $q(\mathbf{x})$-quantile and $(1-q(\mathbf{x}))$-quantile of the conditional treated outcome distribution.

We denote the conditional control mean for always-responders, $E[Y_{i} | D_i = 0, S_i = 1, \mathbf{X}_i = \mathbf{x}]$, as $\theta_{0}(\mathbf{x})$, \(E[Y_{i} | D_i = 1, S_i = 1, \mathbf{X}_i = \mathbf{x}, Y_i \leq y_{q(\mathbf{x})}(\mathbf{x})]\) as
$\theta_{1}^{L}(\mathbf{x})$, and \(E[Y_{i} | D_i = 1, S_i = 1, \mathbf{X}_i = \mathbf{x}, Y_i \geq y_{1-q(\mathbf{x})}(\mathbf{x})]\) as
$\theta_{1}^{U}(\mathbf{x})$, respectively. Then, 
$$
\tau^{L}_{CTB, \mathbf{x}}(1,1) \coloneqq \theta_{1}^{L}(\mathbf{x}) - \theta_{0}(\mathbf{x}) \leq \tau_{\mathbf{x}}(1,1) \leq \tau^{U}_{CTB, \mathbf{x}}(1,1) \coloneqq \theta_{1}^{U}(\mathbf{x}) - \theta_{0}(\mathbf{x}),
$$
where CTB is short for covariate-tightened trimming bounds. The conditional bounds allow us to see how the bounds vary across different values of the covariates. In practice, we can either estimate the conditional bounds for all the observations or over a series of evaluation points that are representative of the whole sample. For instance, we can generate the evaluation points by going through the quantiles of one covariate while fixing the values of the other covariates at their sample mode or average.

We can further integrate ($\tau_{CTB, \mathbf{x}}^L(1,1)$, $\tau_{CTB, \mathbf{x}}^U(1,1)$) over the distribution of $\mathbf{X}_i$ among always-responders to recover the unconditional covariate-tightened bounds on $\tau(1, 1)$. Define
\begin{align*}
& \tau^{L}_{CTB}(1,1) \coloneqq \frac{\int \tau^{L}_{CTB, \mathbf{x}}(1,1) q_0(\mathbf{x}) dF(\mathbf{x})}{\int q_0(\mathbf{x}) dF(\mathbf{x})}, \\
& \tau^{U}_{CTB}(1,1) \coloneqq \frac{\int \tau^{U}_{CTB, \mathbf{x}}(1,1) q_0(\mathbf{x}) dF(\mathbf{x})}{\int q_0(\mathbf{x}) dF(\mathbf{x})},
\end{align*}
where 
\begin{align*}
q_0(\mathbf{x}) := \Pr[S_i(1) = S_i(0) = 1 \mid \mathbf{X}_i = \mathbf{x}] = \Pr[S_i = 1 \mid D_i = 0, \mathbf{X}_i = \mathbf{x}]
\end{align*}
is the expected conditional proportion of always-takers in the overall sample and $F(\mathbf{x})$ is the distribution of the covariates. The expression indicates that $q_0(\mathbf{x})$ can be identified from the control group.
In expectation, covariate-tightened trimming bounds will not be wider than the basic trimming bounds:
$$
\tau^L_{TB}(1, 1) \leq \tau^L_{CTB}(1, 1) \leq \tau(1, 1) \leq \tau^U_{CTB}(1, 1) \leq \tau^U_{TB}(1, 1).
$$
\citet[Prop. 1b]{lee2002trimming} proves this when $\mathbf{X}_i$ is independent of $D_i$, and our Appendix shows this under the more general Assumption~\ref{ass:ind}.
The improvement depends on how well the covariates predict the outcome or response rate \citep{semenova2020better}. These bounds are also sharp---we cannot further reduce the identified set's width without additional information.

\section{Estimation}\label{estimation}

Estimating the bounds involves multiple steps, all of which need to be accounted for in the inference.
As in \citet{lee2009training}, we approach this as a generalized method of moments problem, with suitable additions to account for our use of a generalized random forest ({\it grf}) to estimate nuisance parameters. (See the Appendix for details.)
The attraction of \textit{grf} is its accuracy in approximating local moment conditions even when the covariates are high-dimensional and conditional expectation or quantile functions are non-linear. 
As pointed out by \citet{belloni2017program} and \citet{chernozhukov2018double}, if the same sample is used to estimate both the nuisance and the target parameters, the estimation of the latter is likely to be biased with poor convergence properties. 
Following the suggestions of \citet{chernozhukov2018double}, our \textit{grf} algorithm incorporates Neyman orthogonalization and cross-fitting to avoid such asymptotic biases while making maximal use of the information in our sample data.

The steps to implement our estimation strategy are summarized in Algorithm \ref{alg:honest} below. We first randomly split the sample into $K$ folds for cross-fitting. Next, we employ the ``probability forest'' and ``quantile forest''---both are variants of \textit{grf}---to estimate $(q_0(\mathbf{x}), q_1(\mathbf{x}))$ and $(y_{q(\mathbf{x})}(\mathbf{x}), y_{1-q(\mathbf{x})}(\mathbf{x}))$, respectively, using $(K-1)$ folds from the sample.\footnote{Note that $q(\mathbf{x})$ is estimated on the untreated subjects while the conditional quantiles are estimated on the treated ones. Hence, there is no need to make a separate split to estimate $q(\mathbf{x})$.} Then, we rely on the remaining fold to estimate either the conditional or the aggregated bounds, with the previously estimated $(q_0(\mathbf{x}), q_1(\mathbf{x}))$ and $(y_{q(\mathbf{x})}(\mathbf{x}), y_{1-q(\mathbf{x})}(\mathbf{x}))$ terms plugged in. The aggregated bounds can be directly estimated by the sample analogues of their orthogonalized moment conditions. For the conditional ones, we first estimate $\theta(\mathbf{x}) = (\theta_1^L(\mathbf{x}), \theta_1^U(\mathbf{x}), \theta_0(\mathbf{x}))$ by applying the ``regression forest,'' another variant of \textit{grf}, to their orthogonalized moment conditions, and then calculate the bounds according to their definition. 
%This method's properties have been extensively studied by previous research, and the theory allows us to derive the statistical properties of the resulting trimming bounds estimator. 
Finally, we take the average over results from the $K$ folds as the estimate to avoid any loss in efficiency. 

\begin{figure}
\begin{algorithm}[H]\label{alg:honest}
Define the set of covariates $\mathbf{X} = (X_1, \dots, X_P)$ and evaluation points $\mathcal{X}$ (all sample values for aggregate bounds or specified evaluation points for conditional bounds).

Randomly split the sample into $K$ sets $(\mathcal{I}_1, \mathcal{I}_2, \dots, \mathcal{I}_K)$ with approximately equal size.

\For {each $k \in \{1,2,\dots,K\}$}{
Fit the $q(\mathbf{X})$ model using the probability forests on $\mathcal{I}_{-k} = \cup_{j \neq k} \mathcal{I}_{j}$.

For any $\mathbf{x} \in \mathcal{X}$, treat $q(\mathbf{x})$ as fixed and fit the models ($y_{q(\mathbf{x})}(\mathbf{X})$, $y_{1-q(\mathbf{x})}(\mathbf{X})$) using the quantile forests on $\mathcal{I}_{-k}$.

Treat $q(\mathbf{x})$, $y_{q(\mathbf{x})}(\mathbf{x})$ and $y_{1-q(\mathbf{x})}(\mathbf{x})$ as pre-fixed, construct orthogonalized moment conditions, and fit them on $\mathcal{I}_{k}$.

Calculate estimates for the conditional bounds ($\hat{\tau}_{CTB, \mathbf{x}}^{U}(1,1), \hat{\tau}_{CTB, \mathbf{x}}^{L}(1,1)$) or aggregated bounds ($\hat{\tau}_{CTB}^{U}(1,1), \hat{\tau}_{CTB}^{L}(1,1)$) using the orthogonalized moment conditions.
}
Take the average over the $K$ estimates.

% Aggregate $\hat{\tau}^{U}_{\mathbf{x}}(1,1)$ and $\hat{\tau}^{L}_{\mathbf{x}}(1,1)$ to form estimated bounds of $\tau(1,1)$.
 \caption{Honest inference for the covariate-tightened trimming bounds}
\end{algorithm}
\end{figure}

In the Appendix, we show formally that following Algorithm \ref{alg:honest} yields estimates of either the conditional bounds, ($\hat{\tau}_{CTB, \mathbf{x}}^{U}(1,1), \hat{\tau}_{CTB, \mathbf{x}}^{L}(1,1)$), or the aggregated bounds, ($\hat{\tau}_{CTB}^{U}(1,1), \hat{\tau}_{CTB}^{L}(1,1)$), that are consistent and converge to Normal distributions as $N \rightarrow \infty$. Inference based on this algorithm is honest by the definition of \citet{athey2016recursive}, because it accounts for the uncertainties caused by estimating nuisance parameters properly. For the conditional bounds, the standard error estimates are available from the regression forest output. 
For the aggregated bounds, we can calculate their standard errors from the variance of the orthogonalized moment conditions. We can then construct valid confidence intervals for either the conditional or the aggregated bounds. Note that they are not confidence intervals for the CATE or the ATE {\it per se}, but rather for the bounds, and hence tend to be conservative for coverage of the CATE or ATE. \citet{imbens2004confidence} and \citet{stoye2009more} propose ways to narrow the joint intervals to reflect the fact that the CATE or the ATE will not reach both bounds simultaneously, although we simply report the confidence intervals on the bounds themselves.
\section{Extensions}\label{extension}

In this section, we discuss extensions to cover situations that may arise in practice, including unit-level missing data, conditional monotonicity, binary outcomes, and estimated propensity scores (e.g., for non-experimental data).

\paragraph{Unit-level missing data}\label{missing-covariates}
In practice, we sometimes do not even have covariate information on subjects who do not respond. This will happen, for instance, if we collect the information on $(Y_i, D_i, \mathbf{X}_i)$ only in a post-treatment survey. If so, it will be no longer feasible to estimate $q(\mathbf{x})$ directly, since we do not know the value of $\mathbf{X}_i$ when $S_i = 0$.

However, under Assumption \ref{ass:ind}, we can show that
$$
\begin{aligned}
q(\mathbf{x}) = 1- \frac{1 - p(\mathbf{x})-P(D=0|\mathbf{X} = \mathbf{x},S=1)}{(1-P(D=0|\mathbf{X} = \mathbf{x},S=1))(1-p(\mathbf{x}))}
\end{aligned}
$$
using Bayes' rule. Therefore, we can infer the value of $q(\mathbf{x})$ from the estimate of $P(D_i=0|\mathbf{X}_i = \mathbf{x},S_i=1)$, which only involves observations with $S_i = 1$ and can be estimated with the probability forest.

\paragraph{Conditional monotonicity}\label{cond-mono}
We can further relax Assumption \ref{ass:mono} and allow the direction of monotonicity to change across the values of covariates, as suggested by \citet{semenova2020better}. We replace Assumption \ref{ass:mono} with the following:
\begin{assn}
Conditionally monotonic selection: $S_i(1) \le S_i(0) | \mathbf{X}_i = \mathbf{x}$ or $S_i(1) \ge S_i(0) | \mathbf{X}_i = \mathbf{x}$. \label{ass:cond-mono}
\end{assn}

Under Assumption \ref{ass:cond-mono}, the covariates space $X$ can be divided into two parts, $\mathcal{X}^{+}$ and $\mathcal{X}^{-}$, that satisfy
$$
\begin{cases}
S_i(1) \le S_i(0), & \mathbf{X}_i \in \mathcal{X}^{-} \\
S_i(1) \ge S_i(0), & \mathbf{X}_i \in \mathcal{X}^+ 
\end{cases}
\\ 
\text{where } \mathcal{X}^{+} \cup \mathcal{X}^{-} = \mathcal{X}  
$$
Conditional on each part, we can estimate either the conditional or the aggregated bounds using Algorithm \ref{alg:honest}. To acquire the aggregated bounds on the whole sample, we just need to take the average over estimates from $\mathcal{X}^{+}$ and $\mathcal{X}^{-}$, weighted by the proportion of each part. To identify $\mathcal{X}^+$ and $\mathcal{X}^-$, we first estimate the conditional trimming probability, $q(\mathbf{x})$, and approximate $\mathcal{X}^+$ and $\mathcal{X}^-$ with $\hat{\mathcal{X}}^{+} \coloneqq \{\mathbf{X}_i: \hat{q}(\mathbf{x}) < 1\}$ and $\hat{\mathcal{X}}^{-} \coloneqq \{\mathbf{X}_i: \hat{q}(\mathbf{x}) > 1\}$, respectively. \citet{semenova2020better} shows that the classification is accurate in large samples under regularity conditions. In a finite sample, if the set $\hat{\mathcal{X}}^{-}$ contains many units, it would raise doubts about the validity of Assumption \ref{ass:mono} and could indicate that one should use Assumption \ref{ass:cond-mono}.

\paragraph{Binary outcome}\label{binary-outcome}
In some applications, the outcome variable is binary rather than continuous. For example, political scientists are often interested in whether voters turn out in an election or whether they vote for a particular candidate \citep{jacobson2015campaigns, kalla2018minimal}. Now, the conditional quantile is no longer continuous in $\mathbf{X}$. Nevertheless, we can still estimate the bounds based on similar ideas as for continuous outcomes.

\begin{figure}
\caption{Trimming Bounds with Binary Outcome}
\label{fig_bo}
 \begin{center}
 \includegraphics[width=.45\linewidth, height=.45\linewidth]{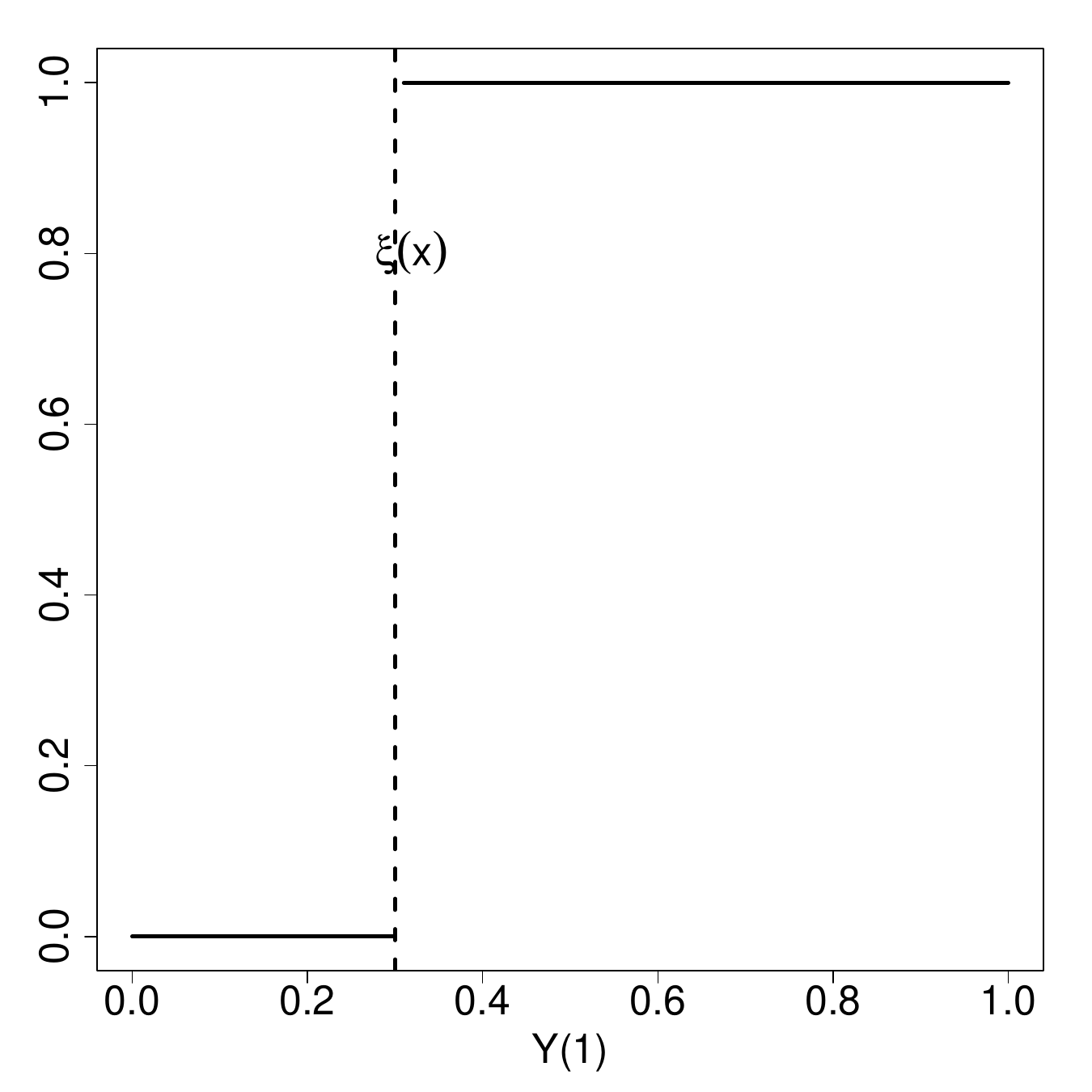}
\includegraphics[width=.45\linewidth, height=.45\linewidth]{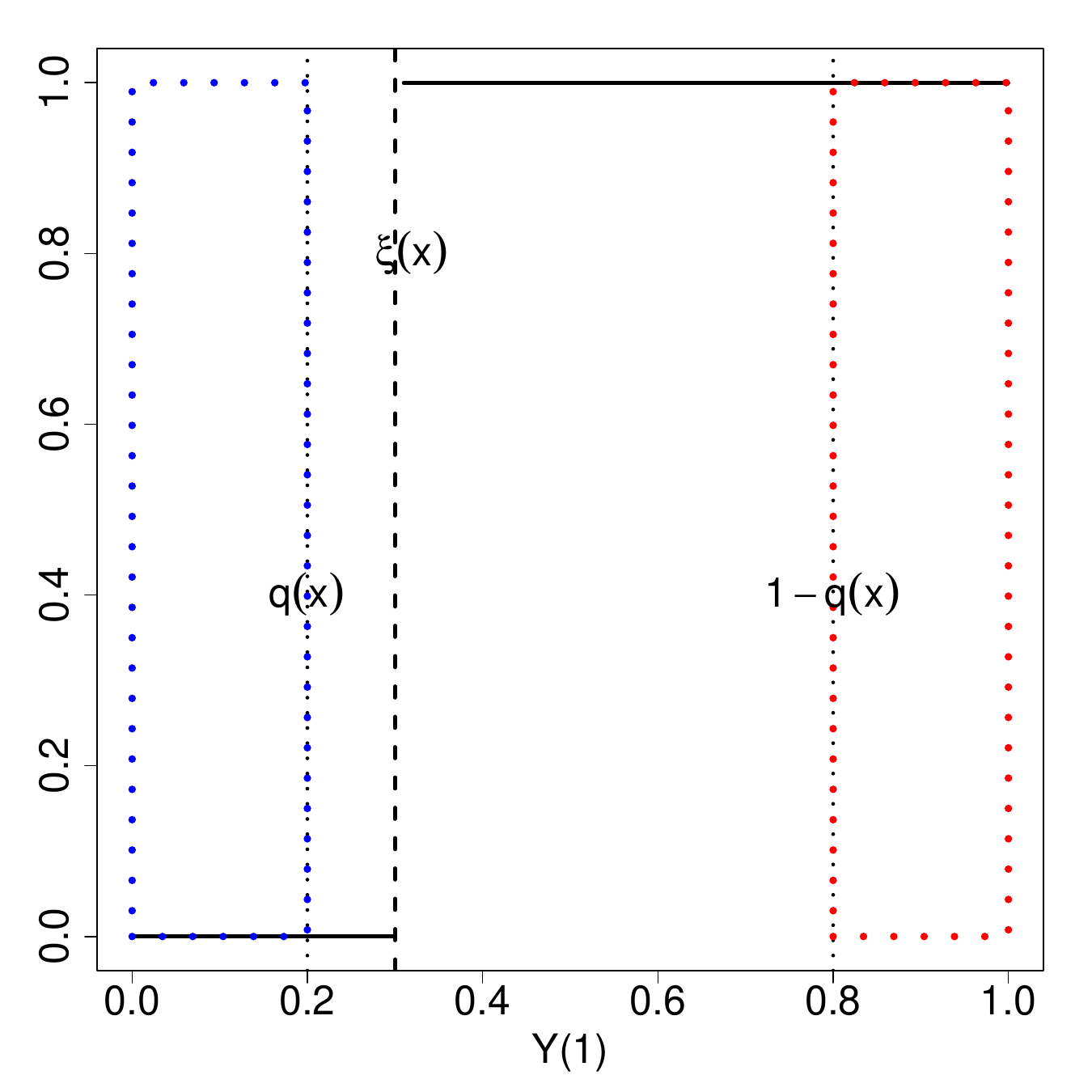}
\includegraphics[width=.45\linewidth, height=.45\linewidth]{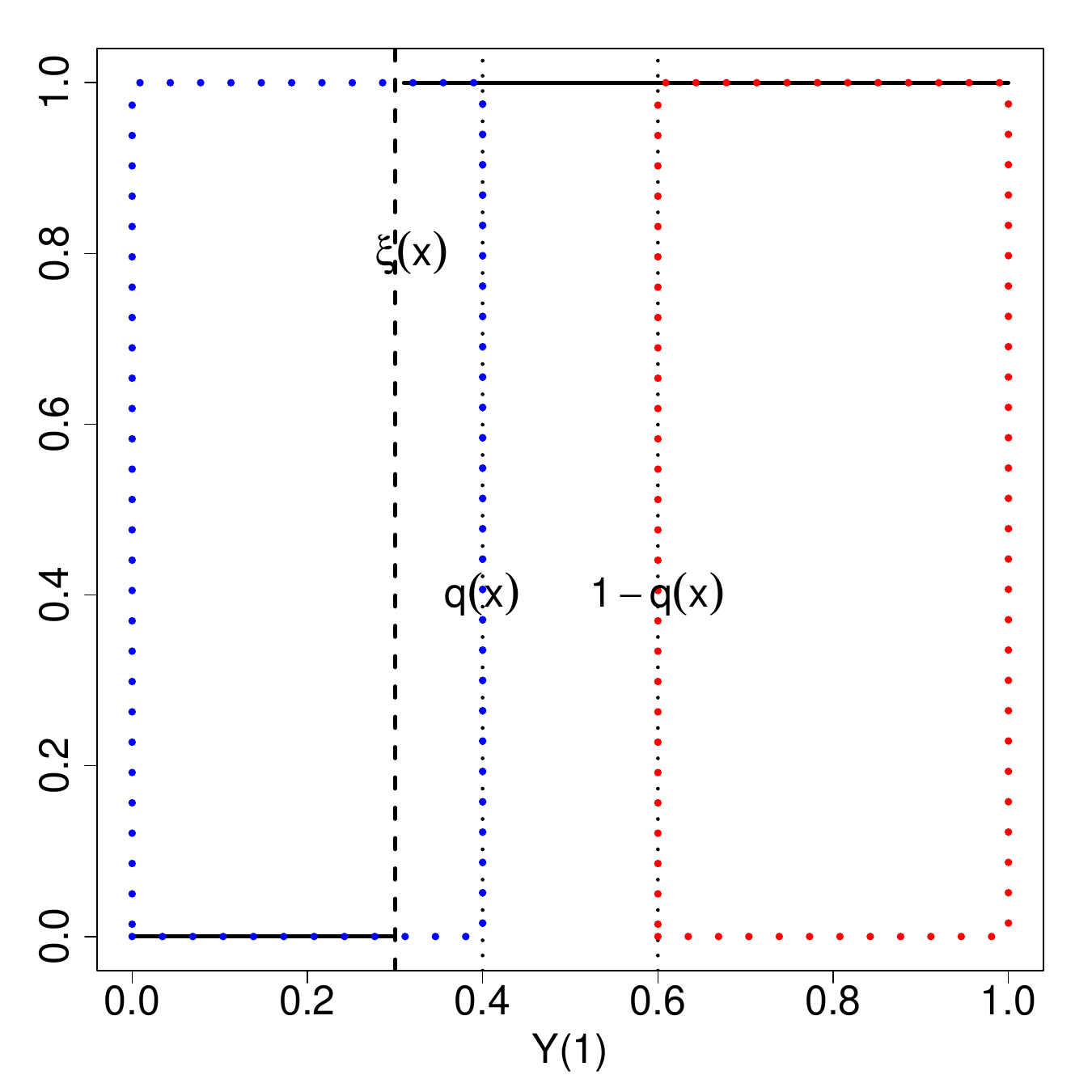}
\includegraphics[width=.45\linewidth, height=.45\linewidth]{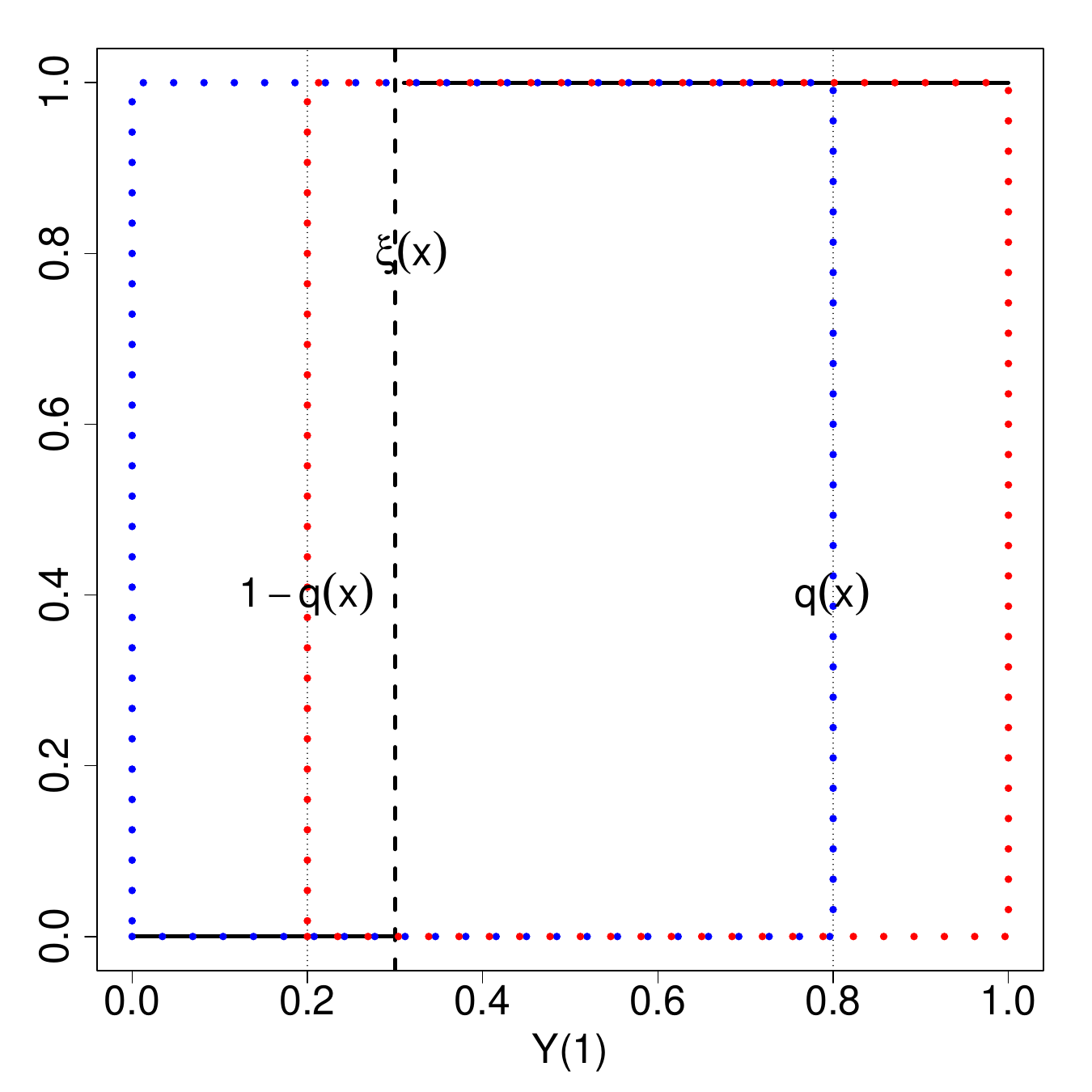}
 \end{center}
 \footnotesize\textbf{Note:} These plots demonstrate how to construct the trimming bounds when the outcome variable is binary. The top-left plot shows the distribution of $Y_i(1)$. The dashed line indicates the value of $\xi(\mathbf{x})$. In the rest three plots, the dotted lines mark the values of $q(\mathbf{x})$ and $1-q(\mathbf{x})$. The rectangle with a blue border represents $\theta_1^L(\mathbf{x})$ and the rectangle with a red border represents $\theta_1^U(\mathbf{x})$. From top-right to bottom-right, the plots show cases where $q(\mathbf{x}) \leq \xi(\mathbf{x}) \leq 1-q(\mathbf{x})$, $\xi(\mathbf{x}) \leq q(\mathbf{x}) \leq 1-q(\mathbf{x})$, and $1-q(\mathbf{x}) \leq \xi(\mathbf{x}) \leq q(\mathbf{x})$.
\end{figure}

Define $\xi(\mathbf{x}) = Pr(Y_i = 0| D_i = 1, S_i = 1, \mathbf{X}_i = \mathbf{x})$. From Figure \ref{fig_bo}, we can see that the trimmed means of the treated outcome are decided by the relative magnitude of $\xi(\mathbf{x})$ and $q(\mathbf{x})$. In the case where $1-q(\mathbf{x}) \leq \xi(\mathbf{x}) \leq q(\mathbf{x})$, for instance, the trimmed mean from below equals $\frac{q(\mathbf{x}) - \xi(\mathbf{x})}{q(\mathbf{x})}$ and the trimmed mean from above equals $\frac{1 - \xi(\mathbf{x})}{q(\mathbf{x})}$. In general, we have the following expressions for the lower and upper bound at point $\mathbf{x}$: 
$$
\begin{aligned}
& \hat \tau^L_{CTB, \mathbf{x}}(1, 1) = \frac{(\hat q(\mathbf{x}) - \hat \xi(\mathbf{x}))\mathbf{1}\{\hat \xi(\mathbf{x}) \leq \hat q(\mathbf{x})\}}{\hat q(\mathbf{x})} - \hat \theta_0(\mathbf{x}), \\
& \hat \tau^U_{CTB, \mathbf{x}}(1, 1) = \frac{\hat q(\mathbf{x}) + (1 - \hat q(\mathbf{x}) - \hat \xi(\mathbf{x}))\mathbf{1}\{\hat \xi(\mathbf{x}) \geq 1-\hat q(\mathbf{x})\}}{\hat q(\mathbf{x})} - \hat \theta_0(\mathbf{x}). \\
\end{aligned}
$$
We can use the probability forest to estimate $\xi(\mathbf{x})$ and obtain estimates for the conditional bounds. Estimates for the aggregated bounds can be constructed by aggregating these conditional bounds.

\paragraph{Estimated propensity scores}\label{est-pscore}
In observational studies, $p(\mathbf{X})$ is unknown and has to be estimated from data. Since $p(\mathbf{X})$ affects all the moment conditions in Section \ref{estimation}, we need to orthogonalize them with regard to this extra nuisance parameter. We illustrate how to implement this modification in the Online Appendix. We also generate an extra split of data for the estimation of $p(\mathbf{X})$. If there also exist missing covariates, then the estimation of $p(\mathbf{X})$ may be biased and the method can no longer be applied.

\section{Simulation}\label{simulation}
We conduct a Monte Carlo simulation experiment to study the performance of the proposed method. We set the sample size $N$ to be 1,000 and the number of covariates $P$ to be 10. Each of the covariates, $X_p$, is randomly drawn from the uniform distribution on $[0, 1]$. The unobservable factor, $U$, also obeys the uniform distribution on $[0, 1]$. We only allow $X_1$ and $U$ to affect both the potential responses ($S_i(0)$, $S_i(1)$) and the potential outcomes ($Y_i(0)$, $Y_i(1)$), while the other 9 covariates are pure noise. The setup resembles the reality in which we possess measures of multiple covariates but do not know {\it ex-ante} which one(s) should be controlled. Our exact data-generating process is as follows:
$$
\begin{aligned}
& Y_i(0) = 1.5 - 0.6*U_{i}^2 + 4*X_{1i} + \varepsilon_{i}, \\
& Y_i(1) = Y_i(0) + 2.5*U_{i} + 3*\sin(-0.7 + 2*X_{1i}), \\
& S_i = \mathbf{1}\{1 - 0.2*X_{1i} - 1.6*U_{i} + 0.4*D_i + 0.1*D_i*X_{1i} + 2*D_i*U_{i} + \nu_i > 0 \}, \\
& D_i \sim Bernoulli(0.5), \varepsilon_{i} \sim \mathcal{N}(0, 1), \nu_{i} \sim \mathcal{N}(0, 1).
\end{aligned}
$$
The individualistic treatment effect for unit $i$, $\tau_i = 2.5*U_i + 3*\sin(-0.7 + 2*X_{1i})$, is a nonlinear and non-monotonic function of the covariates. Both the selection indicator $S_i$ and outcome $Y_i$ are affected by $D_i$ and $U_i$, making the sample selection process endogenous. The construction of $S_i$ guarantees that $S_i(1) \geq S_i(0)$, thus Assumption \ref{ass:mono} holds.

We first estimate the conditional bounds, ($\tau^L_{CTB, \mathbf{x}}(1,1)$, $\tau^U_{CTB, \mathbf{x}}(1,1)$), across 19 points where $X_2$ to $X_{10}$ are fixed at their sample averages, while $X_1$'s value increases by 5\% each time from its minimum. The results are shown in the right plot of Figure \ref{fig_sim_results}. 
Each gray point is an individual-level treatment effect for an observation from one of the Monte Carlo samples.
The black line traces out the true CATE for the always-responders; this quantity would be unobservable in a real application and constitutes what the conditional bounds are intended to cover.
The upper bounds estimates are depicted as red dots and the lower bounds estimates as blue ones. 
The dotted lines surrounding them are the estimated 95\% confidence intervals. 
%We also plot the simulated true upper/lower bound from repeated assignments with the solid red/blue curve. 
The red and blue curves are the true bounds on the CATE for the always-responders, based on the simulation parameters.
The estimates are very accurate in the middle part of $X_1$’s support but slightly biased at the boundary points. 
%That said, they always compose a reasonable range for the variation of the treatment effects. 
%Most of the individualistic treatment effects, represented by gray spots, and the true ATE for the always-responders, depicted as the solid black line, fall inside the identified set formed by the bounds estimates.

\begin{figure}[htp]
\caption{Covariate-tightened Trimming Bounds in Simulated Data}
\label{fig_sim_results}
 \begin{center}
 \includegraphics[width=.49\linewidth, height=\linewidth]{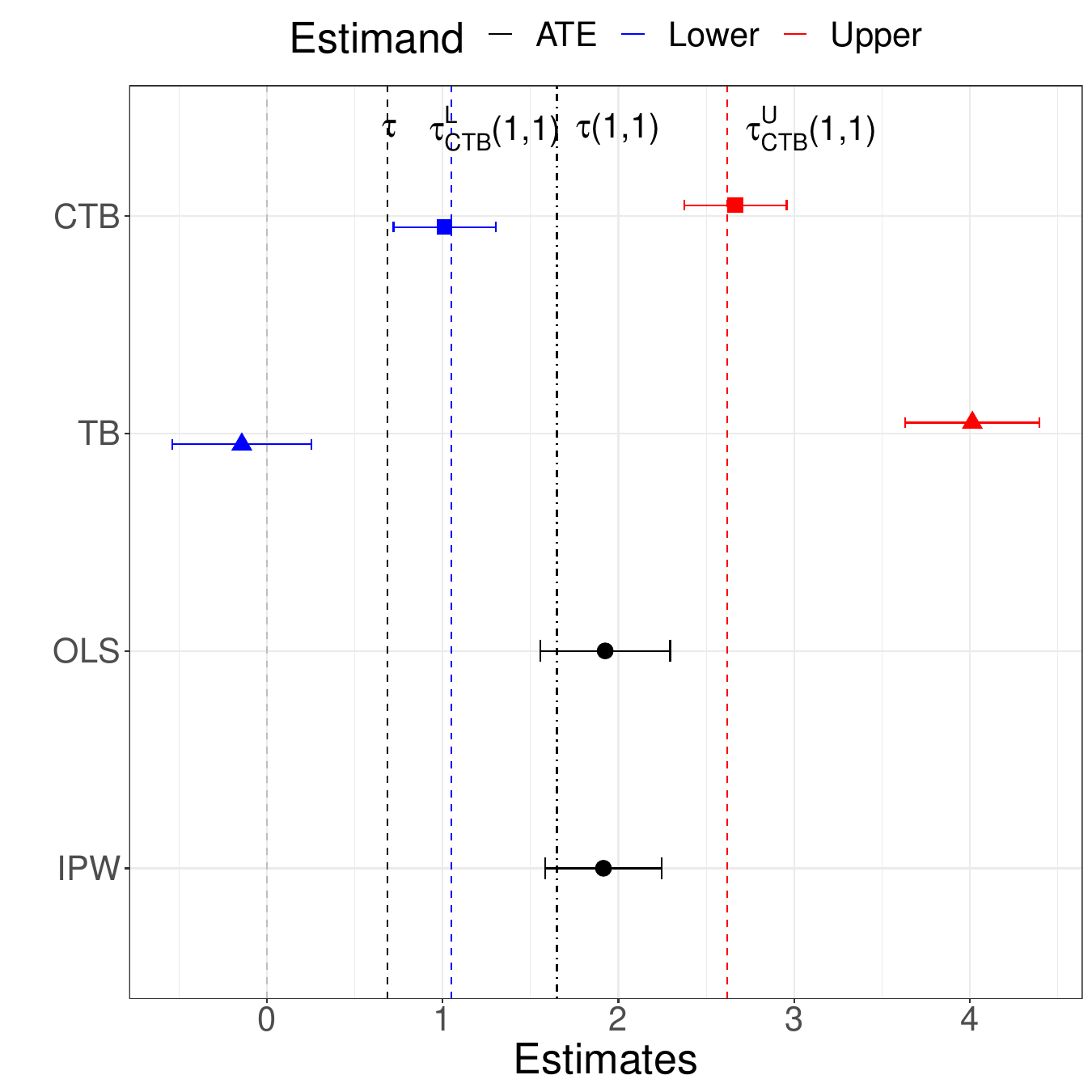}
\includegraphics[width=.49\linewidth, height=.6\linewidth]{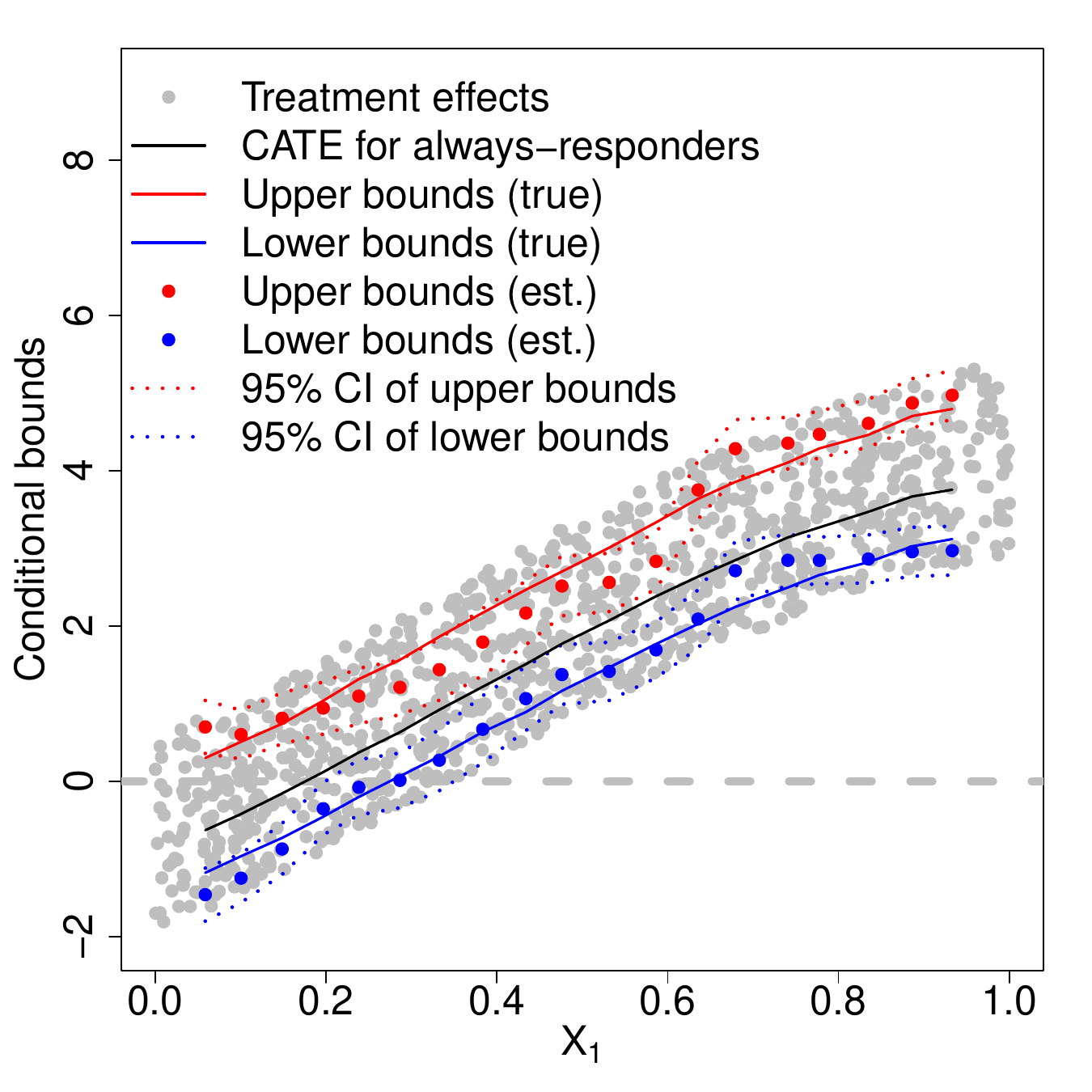}
 \end{center}
 \footnotesize\textbf{Note:} From top to bottom, the plot on the left shows the averages of the estimated covariate-tightened trimming bounds (CTB, in squares), basic trimming bounds (TB, in triangles), ATE estimate using OLS on the non-missing sample (OLS, in circles), and ATE estimate using OLS re-weighted by the inverse probability of attrition (IPW, in circles), over 1,000 assignments. Lower bound estimates are in blue and upper bound estimates are in red. The ATE estimate is in black. The segments represent the 95\% confidence intervals for the estimates. The red and blue dotted lines mark the true values of the CTB. The black dotted line represents the true ATE and the black dash-dotted line represents the true ATE for the always-responders. The plot on the right shows the estimates of the conditional covariate-tightened trimming bounds across 
 19 evaluation points along $X_1$, holding $X_2$ to $X_{10}$ to their sample averages.
 %19 observations for whom $X_2$ to $X_{10}$ are fixed at their sample averages, while $X_1$'s value increases by 5\% each time from its minimum. 
 The red and blue dots represent upper and lower bound estimates, respectively. The dotted lines around them are the 95\% confidence intervals. The red and blue solid lines are the true bounds. The black solid line represents the CATE for the always-responders. The gray dots mark the individualistic treatment effect's value for the corresponding observation.
\end{figure}

In the left plot of Figure \ref{fig_sim_results}, the top two rows present the averages of our estimates of the aggregated bounds and their 95\% confidence intervals across 1,000 treatment assignments. It is clear that our estimates lead to a narrower identified set of $\tau(1,1)$, the ATE for the always-responders. The estimates are also centered around the true CTB, represented by the red and blue dotted lines. In addition, the standard errors of our estimates are not larger, although more complicated models are exploited to infer the parameters. Consequently, the estimated CTBs are more informative of $\tau(1,1)$'s value than the basic trimming bounds (TB). In the third and the fourth row, we show the average estimates of the ATE using OLS on the non-missing sample and OLS  with inverse probability of attrition weighting (IPW), respectively. Both are biased for the true ATE (the black dashed line) and the true ATE for the always-responders (the black dash-dotted line), which is expected since $U$ makes it impossible to point identify the ATE by conditioning on the observable covariates.\footnote{We do not compare to outcome imputation, since it works under the same identifying conditions as IPW, or the Heckman correction because we do not have selection instruments that satisfy the exclusion restriction and so identification would depend exclusively on parametric assumptions \citep{honore2020selection}.} 

% We examine the asymptotic properties of our estimator by increasing the sample size from 400 to 2,000. We compare the mean squared error (MSE) and coverage rate over 1,000 repeated assignments of three approaches: the proposed one with both Neyman orthogonalization and cross-fitting, the one with only Neyman orthogonalization but no cross-fitting, and the one with neither Neyman orthogonalization nor cross-fitting. The results are presented in Figure \ref{fig_sim_asym}. In the top row, we have the MSE of the lower or upper bound estimates. With both Neyman orthogonalization and cross-fitting, the MSE declines to zero as $N$ grows. But this is not the case when either element is absent, as predicted by \citet{chernozhukov2018double}. From the bottom row, we can see that our method provides the correct (95\%) coverage across all the sample sizes. In Online Appendix, we present evidence that the performance of the method remains stable when the number of features in the data increases and is superior to alternative methods \citep{olma2020nonparametric, semenova2020better}. 
We examine the asymptotic properties of our estimators by varying the sample size from 400 to 2,000. In Figure \ref{fig_sim_asym}, we show the mean squared error (MSE) and the confidence interval coverage rate over 1,000 repeated assignments for both the lower and the upper bound estimate. From the left panel of Figure \ref{fig_sim_asym}, we can see that the MSE for either estimate declines toward zero as $N$ grows. The right panel indicates that our method provides the correct (95\%) coverage across the sample sizes. In the Online Appendix, we present evidence that the performance of the method remains stable when the number of features in the data increases and is superior to alternative methods proposed by \citet{olma2020nonparametric} and \citet{semenova2020better}. 

\begin{figure}[htp]
\caption{Asymptotic Performance of the Method}
\label{fig_sim_asym}
 \begin{center}
\includegraphics[width=.48\linewidth, height=.4\linewidth]{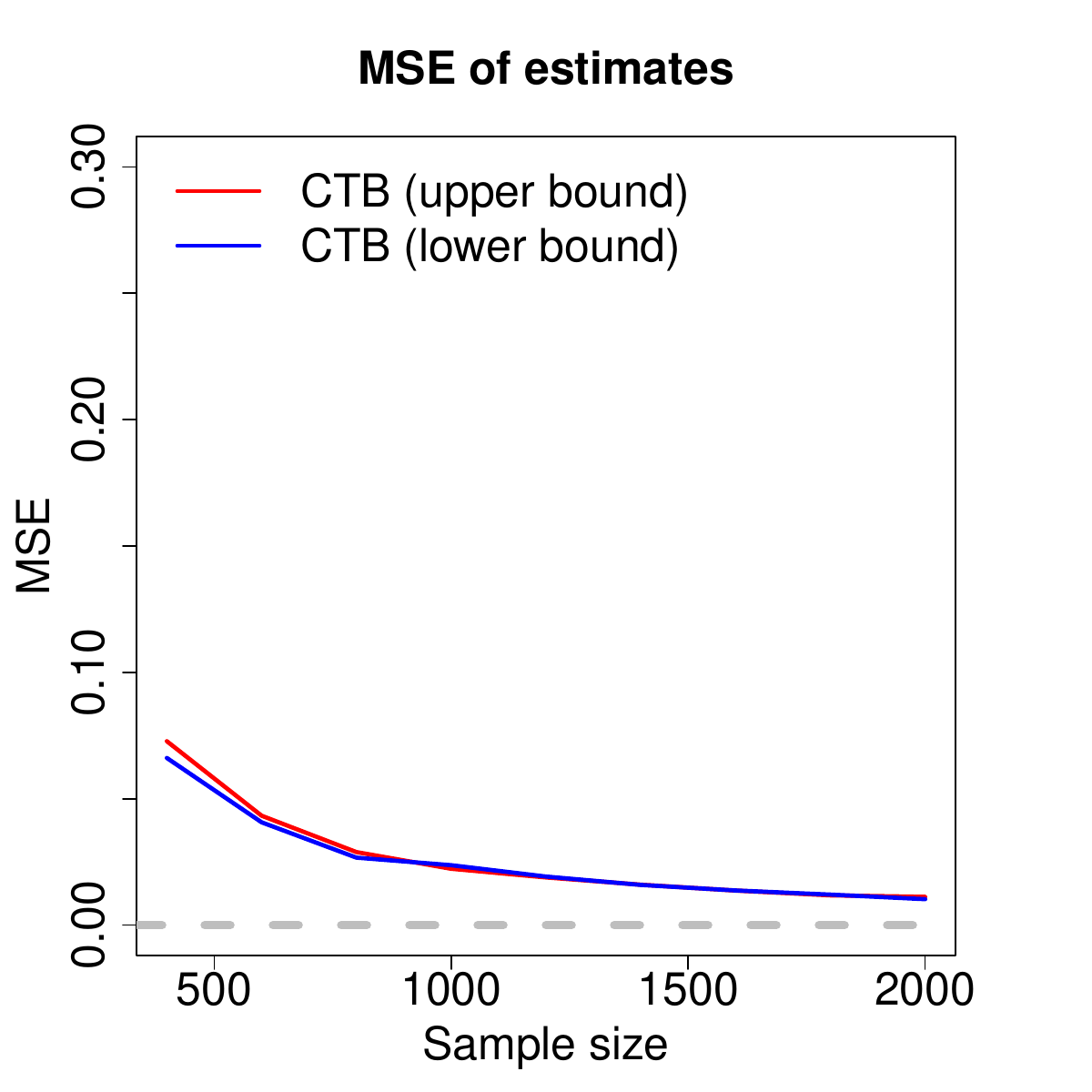}
\includegraphics[width=.48\linewidth, height=.4\linewidth]{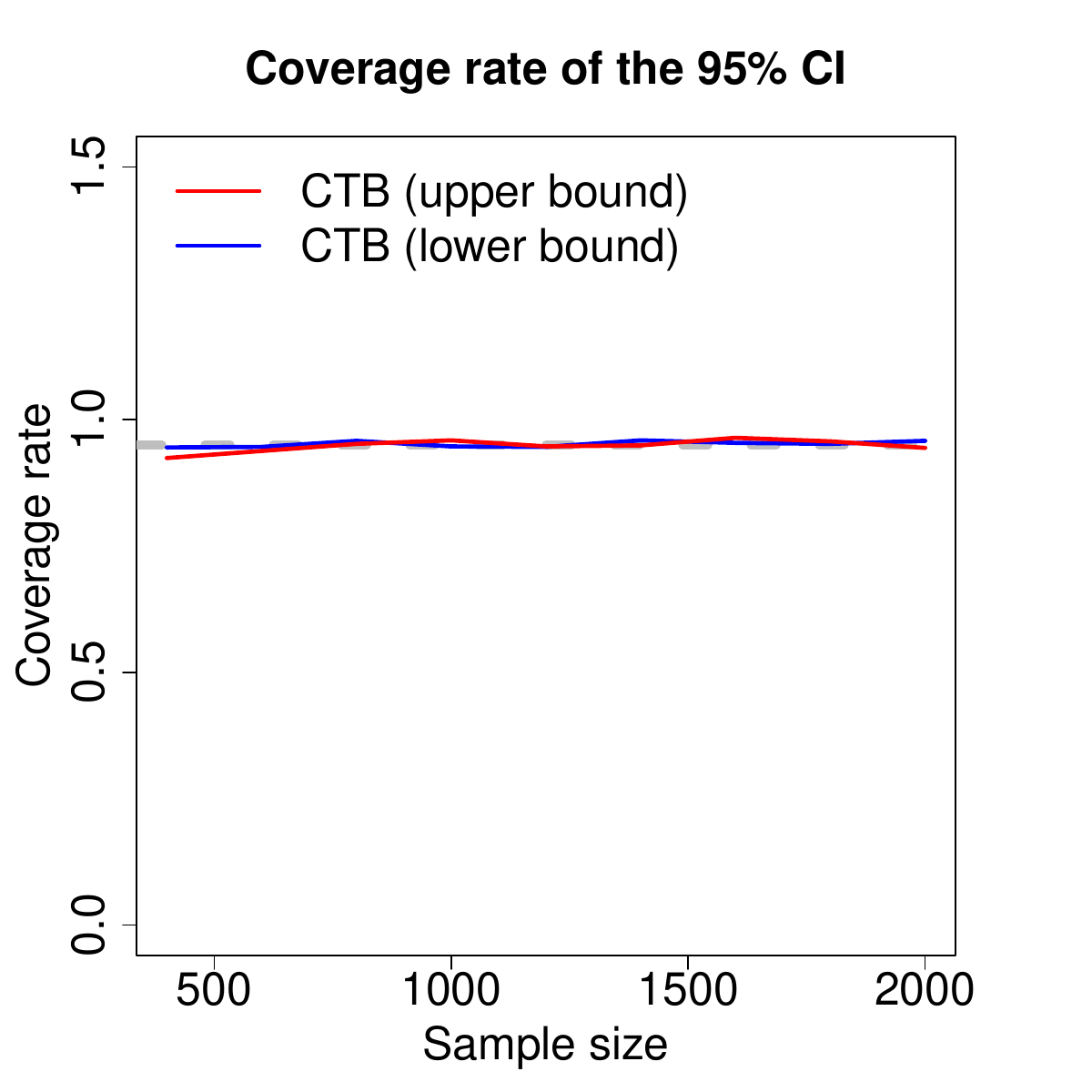} \\
 \end{center}
 \footnotesize\textbf{Note:} The left panel shows how the MSE of the lower bound estimate (in blue) and upper bound estimate (in red) varies with sample sizes. The right panel presents the variation of the coverage rate across sample sizes. The gray line marks the nominal level of coverage, 95\%.
\end{figure}

\section{Applications}\label{application}

In this section, we present applications of our covariate-adjusted bounds to two political science studies by \citet{santoro2022promise} and \citet{blattman2010consequences}.\footnote{For both studies, we use the publicly available replication data files. IRB approvals are as follows: for \citet{santoro2022promise}, Stanford University and University of California-Berkeley (2020-10-13766); for \citet{blattman2010consequences}, University of California-Berkeley (2005-5-7).} Our Appendix reanalyzes the Job Corps study, which was the original application in \citet{lee2009training}, as well as a study by \citet{kalla2022outside} to demonstrate bounds with a binary outcome. 

\subsection{Santoro and Broockman (2022)}\label{application1}
\citet[Study 1]{santoro2022promise} invited subjects to have a video chat on an online platform with a partner from a different party. 
The theme of the conversation is what their perfect day would be like. The study started with 986 subjects who satisfied the screening criteria. The subjects were then randomly assigned into either the treatment group ($D_i = 1$), in which they were informed that the partner would be an outpartisan, or the control group ($D_i = 0$), in which they received no extra information.\footnote{The actual assignment process is slightly more complex: When a subject logs into the platform, she/he will be matched to the subject who has the longest waiting time and inherits the treatment status of that subject. If no one is waiting, then she/he is randomly assigned a treatment status. The authors cluster their standard errors at the level of matched pairs. But if the timing of the subjects to log into the platform is independent, so will be the treatment assignment process across them. Our estimates do not change much with clustered standard errors.} 
Among subjects that were assigned to a treatment condition, 45.2\% of the treated subjects and 39.5\% of the control subjects completed the conversation and post-treatment survey.
The authors examined the treatment effect on a series of outcome variables that measure a subject's attitude toward the other parties and found significant effects of the treatment in the short run.

The authors indicate that the treatment effect on the response rates had a $p$-value of 0.055, and their omnibus test for covariate balance with respect to education, race/ethnicity, gender, age, and party identification had a $p$-value of 0.28. 
Nonetheless, the difference in response rates is not trivial, and other confounding factors may be imbalanced beyond those tested.
We can use our covariate-tightened trimming bounds to assess the robustness of their findings. 
Our inference targets the ``always responders'' that would complete the conversation and survey regardless of being informed about their partner's partisanship. 
We focus on the ``warmth toward outpartisan voters'' outcome (measured by a rescaled thermometer) and rely on the same covariates the authors selected for their analysis, including the age, gender, race, education level, and party identification of the subjects, as well as their pre-treatment outcome.

\begin{figure}[!t]
\caption{Covariate-tightened Trimming Bounds in Santoro and Broockman (2022)}
\label{fig_app_SB}
 \begin{center}
\includegraphics[width=.49\linewidth, height=\linewidth]{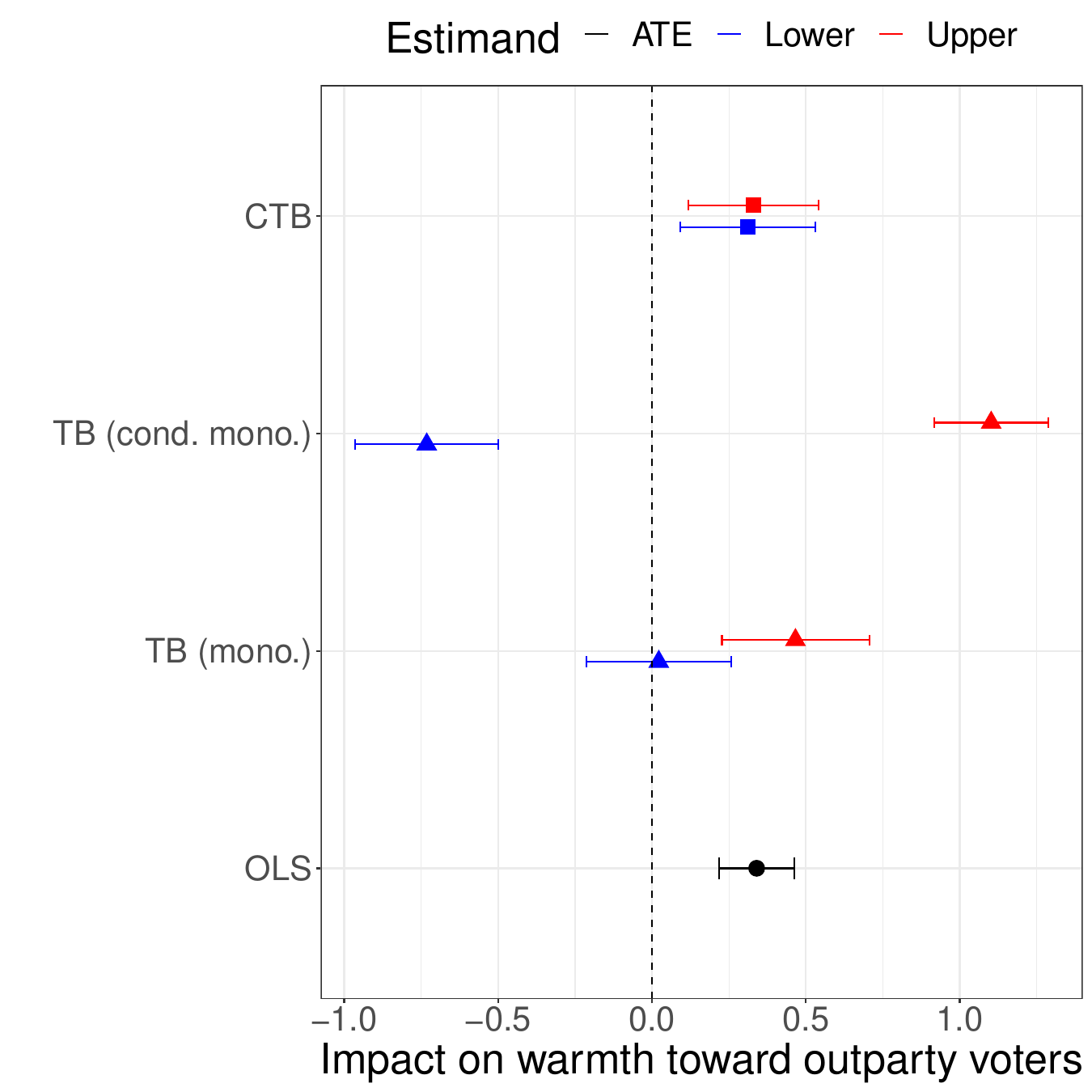}
\includegraphics[width=.49\linewidth, height=.6\linewidth]{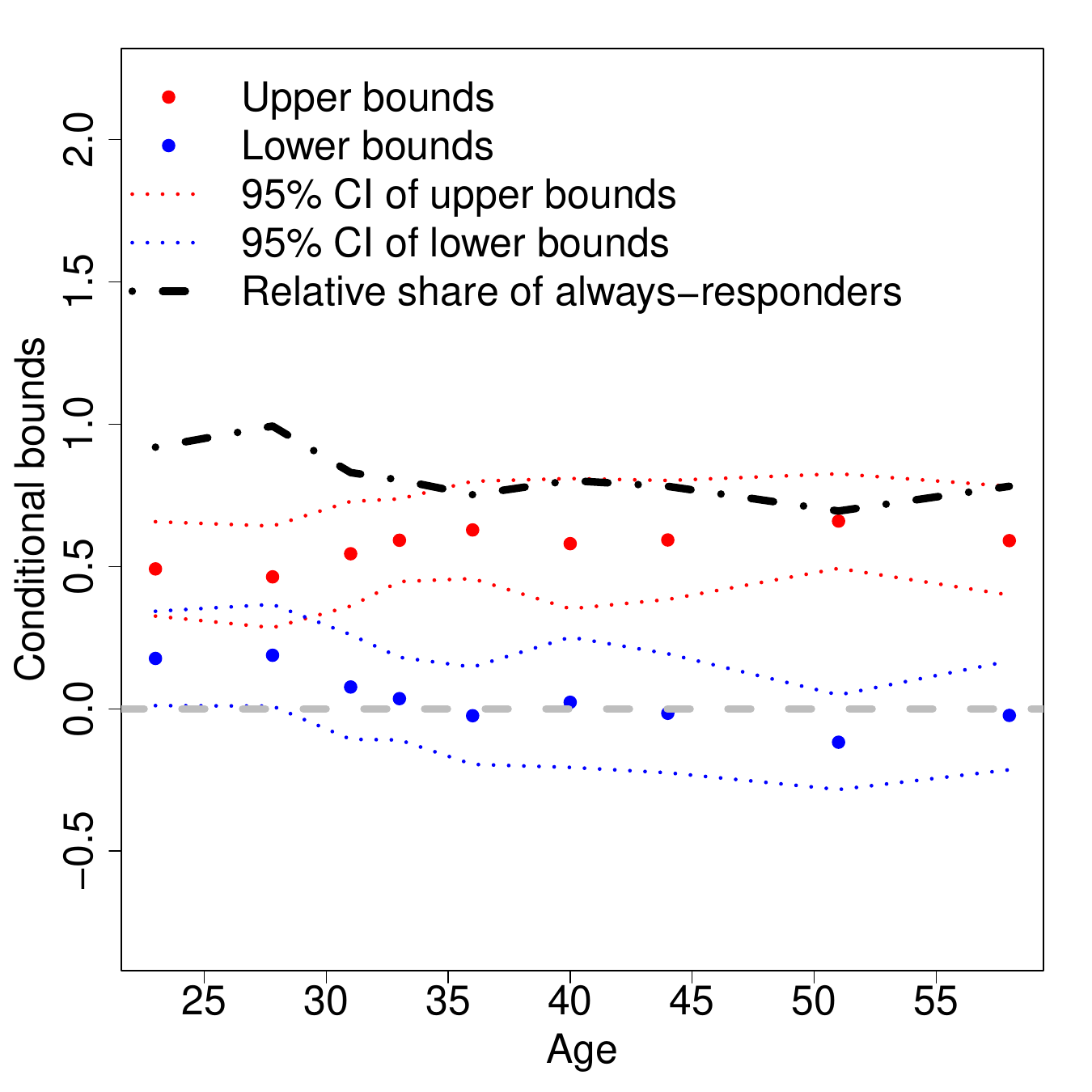}
 \end{center}
 \footnotesize\textbf{Note:} In both plots, the outcome variable is warmth toward outpartisan voters. From top to bottom, the plot on the left shows the estimated covariate-tightened trimming bounds (CTB, in squares), basic trimming bounds under monotonicity (TB (mono.), in triangles) and under conditional monotonicity (TB (cond. mono.), in triangles), and the ATE estimate using OLS on the non-missing sample (OLS, in circles). Lower bound estimates are in blue and upper bound estimates are in red. The ATE estimate is in black. The segments represent the 95\% confidence intervals for the estimates. The plot on the right shows the estimates of the conditional covariates-tightened trimming bounds across 9 evaluation points for which demographic attributes are fixed at the sample mean or mode while age varies across deciles. The red and blue dots represent upper and lower bound estimates, respectively. The dotted lines around them are the 95\% confidence intervals. The dash-dotted curve depicts the conditional trimming probability.
\end{figure}

Figure \ref{fig_app_SB} shows the results. 
At the bottom is the difference-in-means for the observed sample, controlling for the covariates (OLS), replicating the main result in the paper.
Above that (TB mono.) are the basic trimming bounds, assuming monotonic selection.
In this case, however, unconditional monotonicity is questionable. 
While the overall pattern was such that missingness was lower in the treatment group, some types of people (e.g., those with high affective polarization) may actually choose to drop out {\it when informed} about their partner's partisan status. 
When we perform the test proposed in Section~\ref{cond-mono} above, the results suggest that unconditional monotonicity is violated.\footnote{We plot the distribution of the response rates, $q_0(\mathbf{x})$ and $q_1(\mathbf{x})$, as well as the conditional trimming probability in the appendix. These also suggest the need to use the weaker Assumption \ref{ass:cond-mono} of conditional monotonicity.}
We can construct two sets of units ($\mathcal{X}^{+}$ and $\mathcal{X}^{-}$) for which treatment is positively and negatively associated with response.
Trimming bounds for the first group would be based on bounding the treated mean, and bounds for the second group are based on bounding the control mean. 
We can form the basic trimming bounds for the overall sample by constructing the appropriately weighted average of these bounds and the observed treated and control means.
The result is shown as ``TB (cond. mono.).''
We see that relaxing the monotonicity assumption has strong consequences for our inference.
Our proposed bounds covariate-tightened trimming bounds estimator (CTB) makes better use of the covariate information, while also allowing for conditional monotonicity. 
The result is much tighter than the basic bounds, whether or not one allows for violations of unconditional monotonicity.
Both the lower bound ($0.272$) and the upper bound ($0.355$) are significantly larger than zero, suggesting that the ATE for the always-responders is indeed positive. 

We can further examine how bounds on the effects vary over the age of the subjects. We create $9$ ``representative subjects'' whose demographic attributes are fixed at the sample mean or mode while their age varies across its deciles. The estimates of the conditional bounds are shown in the right plot of Figure \ref{fig_app_SB}. The red and blue dots represent the upper and the lower bounds, respectively, while the dotted lines around them are the 95\% confidence intervals. The dash-dotted curve depicts the conditional trimming probability. We can see that the response rate (which determines the share of always-responders) is much higher among subjects who are younger than 30. The conditional lower bounds are also significantly larger than zero for this sub-group, while the identified set crosses zero for other subjects. We show how the conditional bounds vary over the pre-treatment warmth or the ideology of the subjects in the Online Appendix.

\subsection{Blattman and Annan (2010)}\label{application2}
\citet{blattman2010consequences} is an observational study investigating the consequences of being abducted into a rebel organization as a youth in Northern Uganda. The authors argue that conditional on personal characteristics, whether an individual was abducted ($D_i$) was a random event.
This motivates an identification strategy based on covariate conditioning. 
To conduct the study, they constructed a roster of household members in 1996 from 1,100 households across eight rural sub-counties of Uganda. Then, a random sample of 1,216 males born between 1975 and 1991 was drawn from the roster. For each subject in the sample, the authors collaborated with the household head to determine his treatment status---whether they were abducted by the rebels---as well as their relevant demographic characteristics. Among these 1,216 males, 346 had died or not returned from abduction when the study was conducted. The survey enumerators succeeded in tracking 741 males among the remaining 870 and asked them to complete the survey questionnaire. 
Analysis shows that while abductees exhibited ``resilience'' on certain dimensions, abduction affected these subjects negatively along various dimensions: they have fewer years of education, worse labor market performance, and a higher level of distress.

Blattman and Annan note that selection rates differed by treatment status: ``abductees are half as likely to be unfound migrants, twice as likely to have perished, and comprise all of those who did not return from abduction'' (p. 888).
%An abducted subject may be more fragile in health hence more likely to die. Or he chooses to hide from his family to escape the past of being in the rebel. 
The authors applied several techniques to gauge the potential influence of sample selection on their findings, including re-weighting the subjects with the estimated probability of attrition, a sensitivity analysis, and basic trimming bounds. 
We use our proposed covariate-adjusted bounds to evaluate the robustness of their findings.

We focus on two outcome variables: years of education and the level of psychological distress. 
We observe years of education for all 870 males that returned from abduction (household heads were able to provide information in case the subject himself was unreachable for interview), but psychological distress is measured only for the 741 males that were reachable for interview.
We use the covariates from Blattman and Annan's original analysis that do not contain any missing values, such as the age and location of each male, as well as the wealth level of his household. We estimate the treatment propensity score using the probability forest on the entire sample. We adjust our moment conditions accordingly following the discussion in Section \ref{est-pscore}. 

\begin{figure}[!t]
\caption{Covariate-tightened Trimming Bounds in Blattman and Annan (2010)}
\label{fig_app_BA_agg}
 \begin{center}
\includegraphics[width=.49\linewidth, height=.7\linewidth]{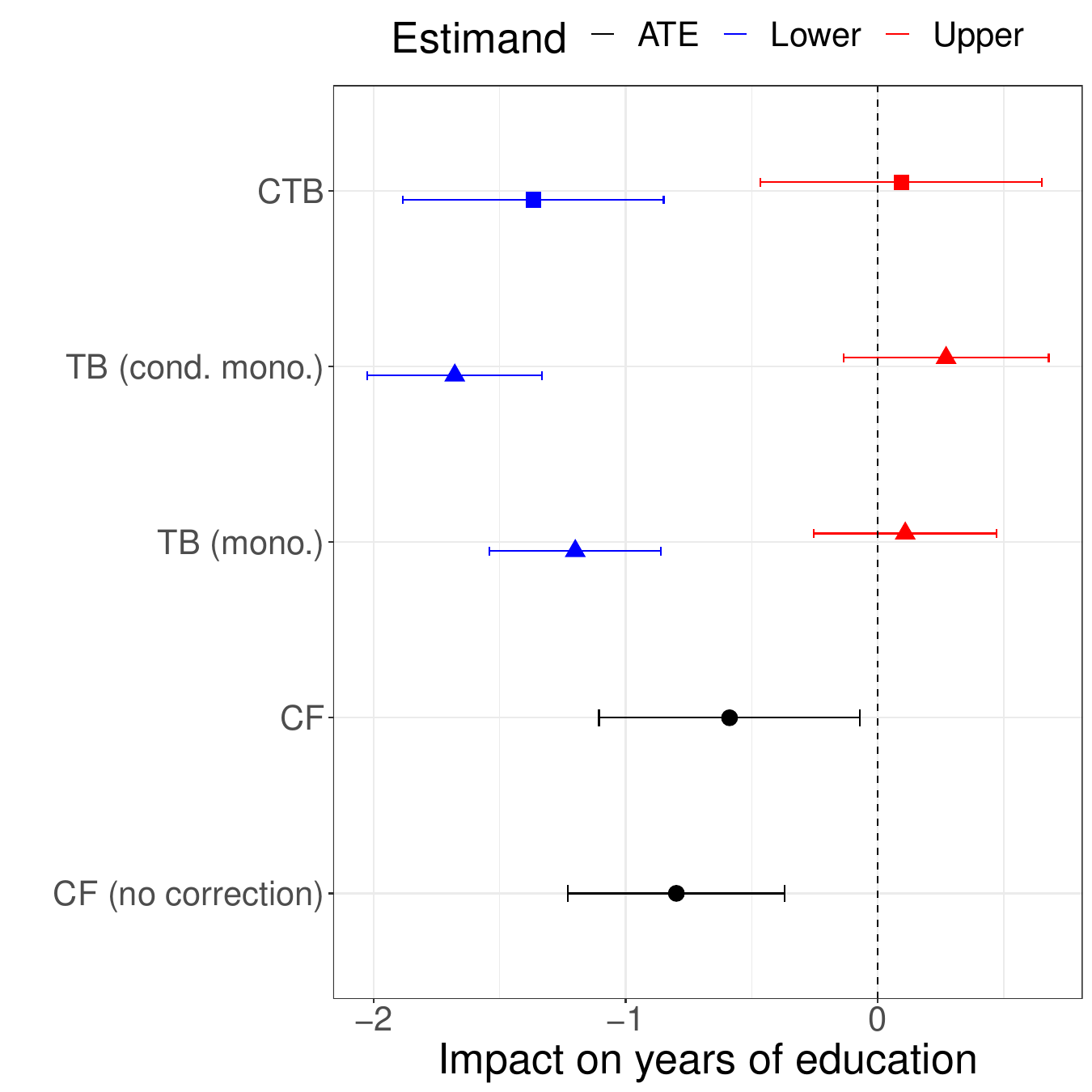}
\includegraphics[width=.49\linewidth, height=.7\linewidth]{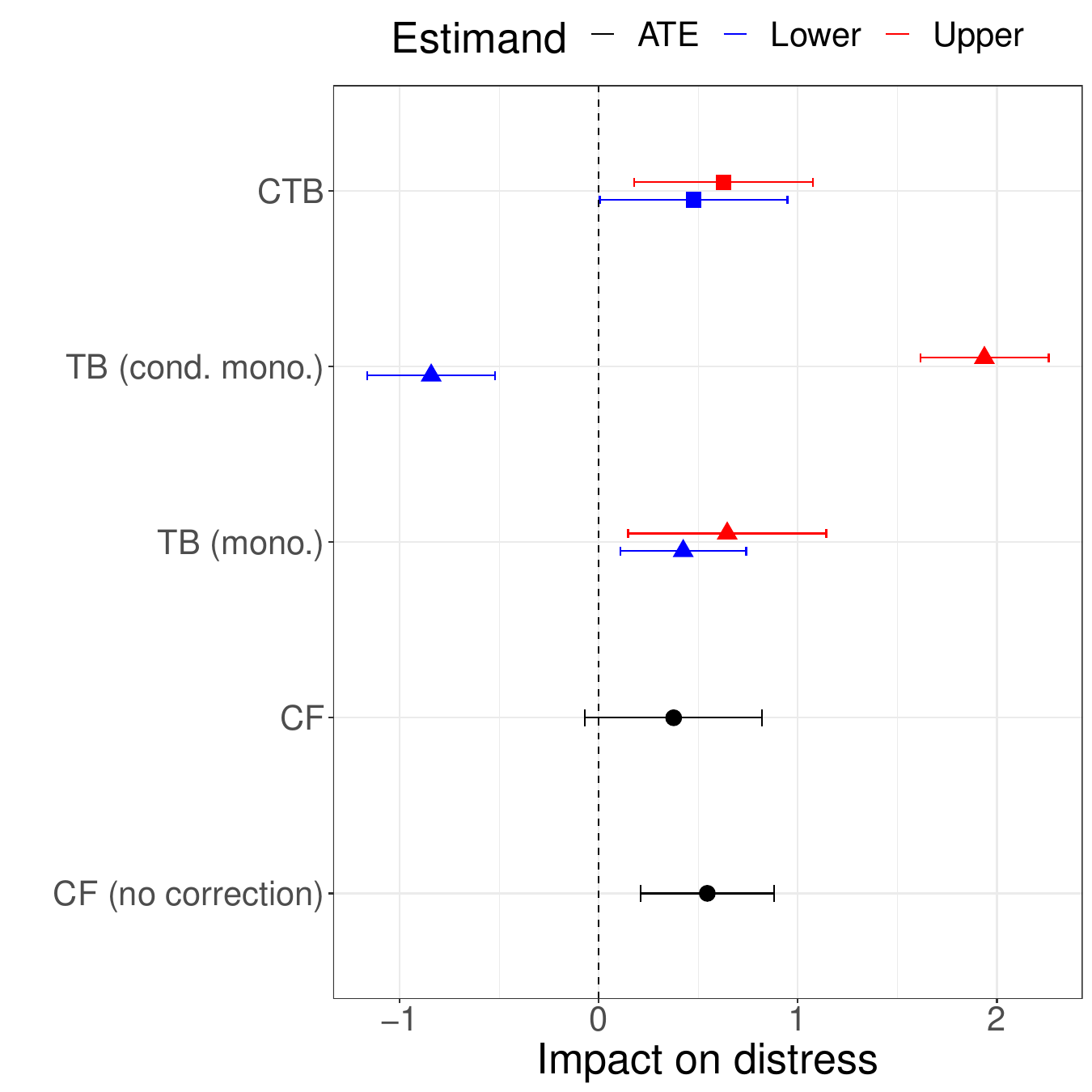}
 \end{center}
 \footnotesize\textbf{Note:} In the left plot, the outcome variable is years of education, while in the right one, it is psychological distress. From top to bottom, the plots show the estimated covariates-tightened trimming bounds (CTB, in crosses), basic trimming bounds under monotonicity (TB (mono.), in squares) and under conditional monotonicity (TB (cond. mono.), in triangles), the ATE estimate using causal forest on the non-missing sample weighted by the inverse probability of attrition (CF, in circles), and the same estimate without re-weighting (CF (no correction), in circles). Lower bound estimates are in blue and upper bound estimates are in red. The ATE estimates are in black. The segments represent the 95\% confidence intervals for the estimates.
\end{figure}

Estimates are presented in Figure \ref{fig_app_BA_agg}. 
At the very bottom are ``causal forest'' estimates that model outcomes in the treated and control group as a function of covariates using generalized random forest (``CF (no correction)'').
Above those are causal forest estimates that are also weighted by the inverse probability of attrition (``CF''). 
These are more elaborate specifications than in the original paper, which only used linear regression specifications for the covariates. 
Nonetheless, the findings are similar to what the original paper found.
Above the causal forest estimates are the basic trimming bounds, assuming monotonic selection. 
When we check for monotonicity using the approach discussed in Section~\ref{cond-mono}, we find evidence against it. 
Thus, we also show basic trimming bounds that allow for the direction of the monotonicity to vary by covariate-defined subgroups (``TB (cond. mono.)'').
Doing so yields extremely wide bounds. 
When we implement the full covariate-adjusted bound (``CTB''), which allows for conditional monotonicity, we obtain results that are more informative, and much more so for psychological distress. 

\section{Conclusion}\label{conclusion}
Even in experiments and plausible quasi-experiments, researchers often face problems of sample selection due to endogenous non-response or censoring. 
When the selection process is affected by treatment assignment, comparing the average outcome between the treated and untreated subjects no longer yields meaningful estimates, because the two groups consist of different principal strata. 
%It is closely related to the issue of ``post-treatment bias \citep{montgomery2018conditioning},'' ``phantom counterfactual \citep{slough2022phantom},'' or ``missing not at random \citep{liu2021latent}'' discussed in the literature of political methodology. 
Conventional methods attempt to solve the problem by modeling the selection process, using the Heckman correction, inverse probability of attrition weighting, or outcome imputation. 
But the validity of these methods requires that either conditioning on observables is sufficient or that one knows obscure details of the data-generating process.
Inferences are invalid if these conditions are not met in actuality. 
Such approaches do not allow us to be agnostic about the data-generating process and undermine the design-based logic of experimental or quasi-experimental analysis.

We propose an alternative approach that avoids making strong structural assumptions while also improving, in terms of precision, on existing methods. 
Our proposal combines the trimming bounds method developed by \citet{lee2009training} with a recently developed machine learning algorithm, generalized random forest \citep{wager2018estimation, athey2019generalized}, to generate conditional or covariate-adjusted aggregated bounds for the average treatment effect on the always-responders. Our approach allows us to incorporate information from a large set of covariates to tighten the bounds. 
Evidence from our simulation and replication studies shows that the proposed method can be informative when existing methods are not. 
Importantly, our approach allows one to loosen a ``monotonicity'' assumption and to still construct informative bounds; we that the monotonicity assumption is indeed problematic in applied settings.
The price we pay is to forgo a point estimate and rather settle for a range of plausible estimates. 
In cases where experiments or quasi-experiments are tainted by complex selection problems, we consider this to be a reasonable way to balance credibility and robustness with the need to provide informative estimates \citep{manski2019communicating}.

We have developed an open-source R package, \textit{CTB}, for practitioners to implement the method. We expect it to have wide applications in political science and in other social science disciplines.

%This is work in progress, and so future steps will involve a number of extensions.  This includes extending our analysis to the two-sided missingness case and to other principal strata problems for which trimming methods are useful (e.g., principal effects under non-compliance with more than one alternative to treatment).  

\bibliography{cttb}

\end{document}

\renewcommand\thesection{\Alph{section}}

\begin{center}
{\Large\bf Supplementary Information \emph{for}}\\\bigskip
{\large\bf Generalizing Trimming Bounds for Endogenously Missing Outcome Data Using Random Forests}\\\bigskip
{\normalsize Cyrus Samii (NYU)\hspace{2em}Ye Wang (UNC)\hspace{2em} Junlong Aaron Zhou (Tencent America)}
\end{center}
\bigskip

% \bigskip\noindent\textbf{{\large -- Not for Publication --}}
\vspace{2em}

\noindent\hspace{0em}{\large\bf \underline{Table of Contents}}

{\bf
\begin{enumerate}\itemsep0ex
    \item[A.] Proofs
    \begin{enumerate}\itemsep0ex\vspace{-0.5em}
    	\item[A.1.] Improvement from incorporating covariates
            \item[A.2.] Moment conditions in the general case
            \item[A.3.] Statistical theory
    	\item[A.4.] Unit missingness and missing covariates
	  \end{enumerate}
    \item[B.] Extra results
    \begin{enumerate}\itemsep0ex\vspace{-0.5em}
      \item[B.1.] Extra results from simulation
      \item[B.2.] Extra results from applications
      \item[B.3.] Revisit the Job Corps experiment
      \item[B.4.] Replication results of Kalla and Broockman (2022)
    \end{enumerate}
\end{enumerate}
}
\clearpage

\small

\newpage
\section{Proofs}\label{appx:A}
\subsection{Improvement from incorporating covariates}\label{appx:A1}
Here we show how incorporating the information from covariates reduces the length of the partially identified set (i.e., the length of the interval between the bounds). We first re-state a lemma from \citet{lee2002trimming}:
\begin{lemma}\label{lemma:mono}
Suppose $f_{Y}(y) = qf^{\dagger}_{Y}(y) + (1-q)f^{\ddagger}_{Y}(y)$, then
\begin{align*}
    & \frac{1}{q} \int_{-\infty}^{y_{q}}y f_{Y}(y)dy \leq \int y f_{Y}^{\dagger}(y)dy, \\
    & \frac{1}{1-q} \int_{y_{1-q}}^{\infty}y f_{Y}(y)dy \geq \int y f_{Y}^{\ddagger}(y)dy,
\end{align*}
where $\int_{-\infty}^{y_{q}} f_{Y}(y)dy = q$.
\end{lemma}

\begin{proof}
Let's denote $\int_{y_{q}}^{\infty} f_{Y}^{\dagger}(y)dy$ as $\kappa$, then
\begin{align*}
    & \frac{1}{\kappa}\left[\frac{1}{q} \int_{-\infty}^{y_{q}}y f_{Y}(y)dy - \int y f_{Y}^{\dagger}(y)dy \right] \\
    = & \int_{-\infty}^{y_{q}}y \frac{f_{Y}(y)/q - f_{Y}^{\dagger}(y)}{\kappa}dy - \int_{y_{q}}^{\infty} y \frac{f_{Y}^{\dagger}(y)}{\kappa}dy.
\end{align*}
Note that $\int_{-\infty}^{y_{q}} \frac{f_{Y}(y)/q - f_{Y}^{\dagger}(y)}{\kappa}dy = \frac{1}{\kappa} - \frac{1 - \kappa}{\kappa} = 1$ and $\int_{y_{q}}^{\infty} \frac{f_{Y}^{\dagger}(y)}{\kappa}dy = 1$, hence the first term is the expectation of $Y$ over $[-\infty, y_q]$ and the second is the expectation of $Y$ over $[y_q, \infty]$. As the first support is below the second one, the first expectation is always smaller than the second as long as $f_{Y}^{\dagger}(y)$ is not degenerate. The other inequality can be similarly proved.

\end{proof}

% Define $\tilde{f}_{Y}(y) = \begin{cases}\frac{f_{Y}(y)}{q} \text{ if } y \leq y_{q} \\ 0 \text{, otherwise} \end{cases}$. Then, $\int_{-\infty}^{y_{q}}y f_{Y}(y)dy = \int y \tilde{f}_{Y}(y)dy$

For the lower bound, we define
\begin{align*}
    & f_{Y|D=0, S=1, \mathbf{X} = \mathbf{x}}^{\dagger}(y) = f_{Y|D=0, S=1, \mathbf{X} = \mathbf{x}}(y) \mathbf{1}\{y \leq y_{q(\mathbf{x})}(\mathbf{x})\}/q(\mathbf{x}) \\
    & f_{Y|D=0, S=1, \mathbf{X} = \mathbf{x}}^{\ddagger}(y) = f_{Y|D=0, S=1, \mathbf{X} = \mathbf{x}}(y) \mathbf{1}\{y \geq y_{q(\mathbf{x})}(\mathbf{x})\}/(1-q(\mathbf{x})) \\
    & f_{Y|D=0, S=1}^{\dagger}(y) = \int f_{Y|D=0, S=1, \mathbf{X} = \mathbf{x}}^{\dagger}(y)\frac{q(\mathbf{x})}{q}  f_{\mathbf{X}|D=0, S=1}(\mathbf{x})d\mathbf{x} \\
    & f_{Y|D=0, S=1}^{\ddagger}(y) = \int f_{Y|D=0, S=1, \mathbf{X} = \mathbf{x}}^{\ddagger}(y) \frac{1-q(\mathbf{x})}{1-q} f_{\mathbf{X}|D=0, S=1}(\mathbf{x})d\mathbf{x}.
\end{align*}
Then, by definition,
\begin{align*}
 f_{Y|D=0, S=1}(y) = & \int f_{Y|D=0,S=1, \mathbf{X} = \mathbf{x}}(y) f_{\mathbf{X}|D=0, S=1}(\mathbf{x})d\mathbf{x} \\
    = & \int f_{Y|D=0, S=1, \mathbf{X} = \mathbf{x}}^{\dagger}(y) q(\mathbf{x})  f_{\mathbf{X}|D=0, S=1}(\mathbf{x})d\mathbf{x} \\
    & + \int f_{Y|D=0, S=1, \mathbf{X} = \mathbf{x}}^{\ddagger}(y) (1-q(\mathbf{x}))  f_{\mathbf{X}|D=0, S=1}(\mathbf{x})d\mathbf{x} \\
    = & q \int f_{Y|D=0, S=1, \mathbf{X} = \mathbf{x}}^{\dagger}(y)\frac{q(\mathbf{x})}{q}  f_{\mathbf{X}|D=0, S=1}(\mathbf{x})d\mathbf{x} \\
    & + (1-q) \int f_{Y|D=0, S=1, \mathbf{X} = \mathbf{x}}^{\ddagger}(y) \frac{1-q(\mathbf{x})}{1-q} f_{\mathbf{X}|D=0, S=1}(\mathbf{x})d\mathbf{x} \\
    = & q f_{Y|D=0, S=1}^{\dagger}(y) + (1-q)f_{Y|D=0, S=1}^{\ddagger}(y)
\end{align*}
From Lemma \ref{lemma:mono}, we know that
\begin{align*}
    & E[Y_{i} | D_i = 0, S_i = 1, Y_i \leq y_{q}] = \frac{1}{q}\int_{-\infty}^{y_{q}}y f_{Y|D=0, S=1}(y)dy \\
    \leq & \int y f_{Y|D=0, S=1}^{\dagger}(y)dy = \int \int yf_{Y|D=0, S=1, \mathbf{X} = \mathbf{x}}^{\dagger}(y)\frac{q(\mathbf{x})}{q}  f_{\mathbf{X}|D=0, S=1}(\mathbf{x})d\mathbf{x}dy \\
    = & \int \theta_{0}^{L}(\mathbf{x})\frac{q(\mathbf{x})}{q}  f_{\mathbf{X}|D=0, S=1}(\mathbf{x})d\mathbf{x} \\
    = & \int \theta_{0}^{L}(\mathbf{x})f_{\mathbf{X}|D=1, S=1}(\mathbf{x})d\mathbf{x} \\
    = & E[E[Y_{i} | D_i = 0, \mathbf{X}_i = \mathbf{x}, Y_i \leq y_{q(\mathbf{x})}(\mathbf{x})]]
\end{align*}
The second to the last equality holds since by definition
\begin{align*}
    q(\mathbf{x}) = & \frac{P[S=1, D=1 | \mathbf{X}=\mathbf{x}]}{p(\mathbf{x})} = \frac{P[S=1, D=1, \mathbf{X}=\mathbf{x}]}{p(\mathbf{x})P[\mathbf{X}=\mathbf{x}]} \\
    = & \frac{P[\mathbf{X}=\mathbf{x} | S=1, D=1]P[S=1, D=1]}{p(\mathbf{x})P[\mathbf{X}=\mathbf{x}]} \\
    = & \frac{f_{\mathbf{X}|D=1, S=1}(\mathbf{x})P[S=1, D=1]}{p(\mathbf{x})f_{\mathbf{X}|D=0, S=1}(\mathbf{x})} \\
    = & \frac{qf_{\mathbf{X}|D=1, S=1}(\mathbf{x})}{f_{\mathbf{X}|D=0, S=1}(\mathbf{x})}.
\end{align*}

From the derivation, we see that
\begin{align*}
    & E[Y_{i} | D_i = 0, S_i = 1, Y_i \leq y_{q}] -  E[E[Y_{i} | D_i = 0, \mathbf{X}_i = \mathbf{x}, Y_i \leq y_{q(\mathbf{x})}(\mathbf{x})]] \\ 
    = & \frac{1}{q}\int_{-\infty}^{y_{q}}y f_{Y|D=0, S=1}(y)dy - \int \theta_{0}^{L}(\mathbf{x})\frac{q(\mathbf{x})}{q}  f_{\mathbf{X}|D=0, S=1}(\mathbf{x})d\mathbf{x} \\
    = & \int \theta_{0}^{L}f_{\mathbf{X}|D=0, S=1}(\mathbf{x})d\mathbf{x} - \int \theta_{0}^{L}(\mathbf{x})\frac{q(\mathbf{x})}{q}  f_{\mathbf{X}|D=0, S=1}(\mathbf{x})d\mathbf{x} \\
    = & \int \frac{\theta_{0}^{L}q - \theta_{0}^{L}(\mathbf{x})q(\mathbf{x})}{q}f_{\mathbf{X}|D=0, S=1}(\mathbf{x})d\mathbf{x}
\end{align*}
Therefore, the improvement is more pronounced when the $\theta_{0}^{L}q - \theta_{0}^{L}(\mathbf{x})q(\mathbf{x})$ is small where $f_{\mathbf{X}|D=0, S=1}(\mathbf{x})$ is small.

\begin{figure}[htp]
 \begin{center}
\includegraphics[width=.48\linewidth, height=.8\linewidth]{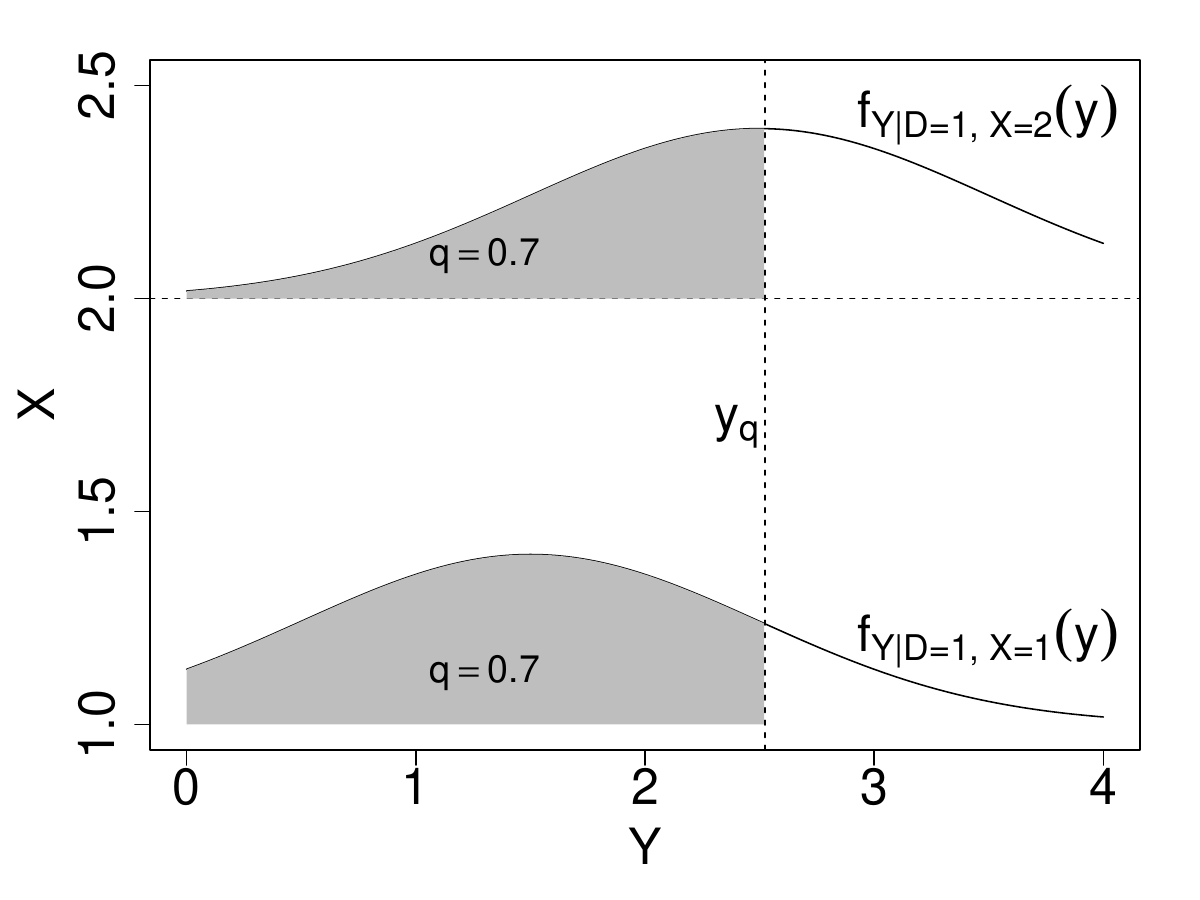}
\includegraphics[width=.48\linewidth, height=.8\linewidth]{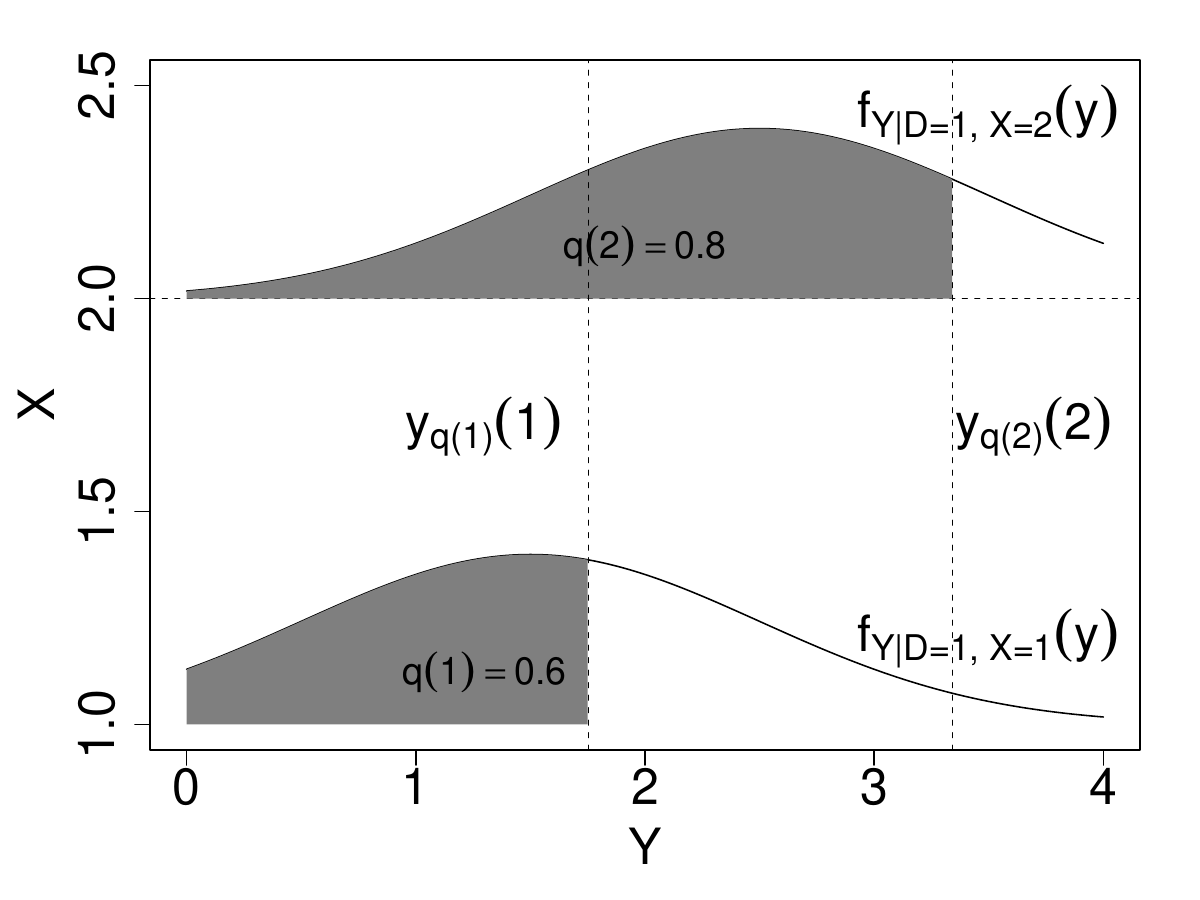} \\
\includegraphics[width=.48\linewidth, height=.8\linewidth]{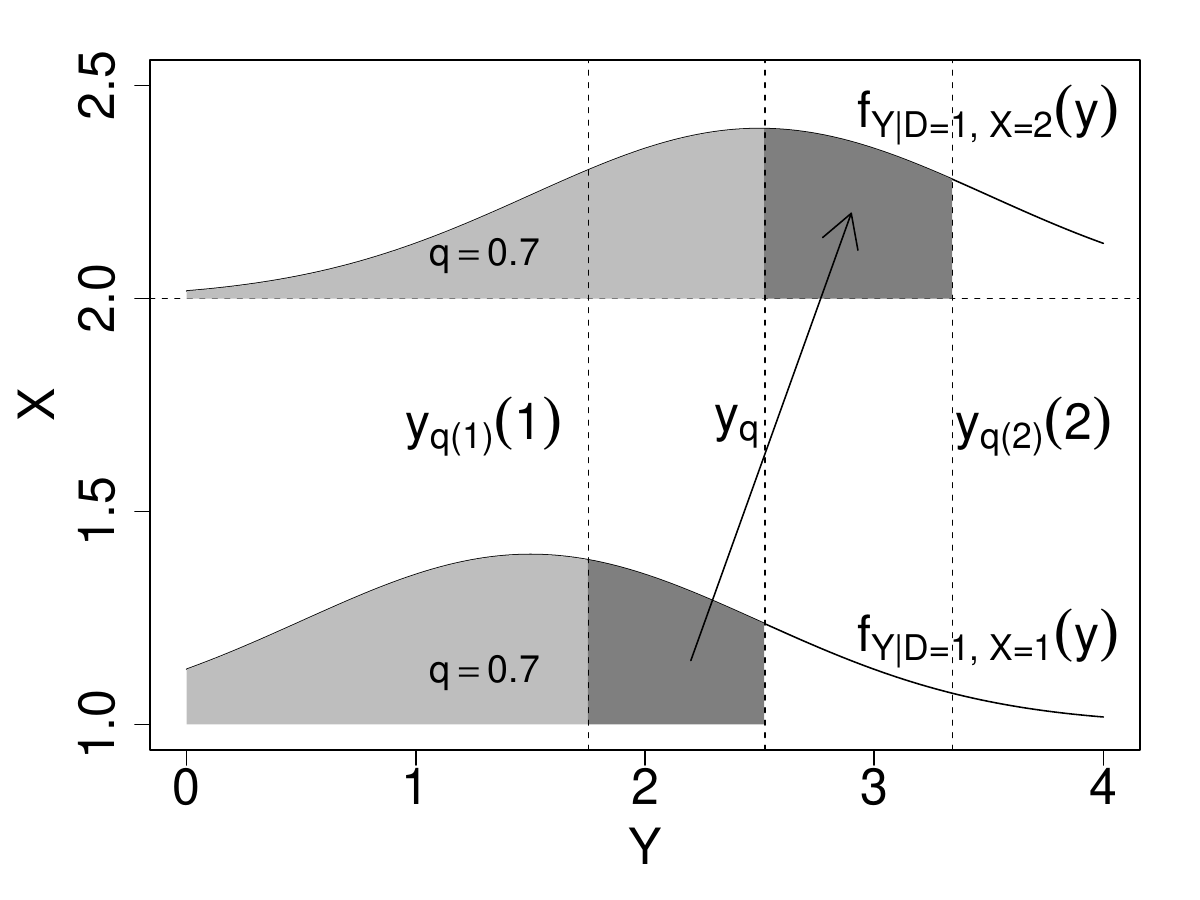}
\caption{An illustration of the CTB}
\label{fig_app_ctb_exp}
 \end{center}
\end{figure}

Let's consider a simple case where $X$ takes two values $\{1, 2\}$ with equal probability (0.5). Suppose $q(1) = 0.2$, $q(2) = 0.4$, then $q = 0.3$. Note that
$$
f_{Y|D=0, S=1}(y) = 0.5*f_{Y|D=0, S=1, X=1}(y) + 0.5*f_{Y|D=0, S=1, X=2}(y).
$$
We illustrate why the CTB lead to a smaller identified set using Figure \ref{fig_app_ctb_exp}. The top-left plot shows how we construct the lower bound using the basic method, where the trimming probability equals 0.3 for both values of $X$. We take expectation of $Y_i(1)$ over the shaded region under each condition distribution and their average equals the lower bound of the treated outcome. The top-right plot shows the idea behind the CTB, where the trimming probability equals 0.2 when $X = 1$ and 0.4 when $X = 2$. We follow the same steps, taking conditional expectations first and then calculating their average. Note that in the two plots, the total area of the shaded region is the same. But the value of $Y_i(1)$ is strictly larger in the top-right plot, as illustrated by the plot at the bottom. This explains where the improvement comes from under the CTB.

\subsection{Moment conditions}\label{appx:A2}
We start from the scenario where Assumptions 1 and 2 are satisfied. It is equivalent to conditioning on the set $\mathcal{X}^{+}$ under Assumption 3. Following the terminology in the generalized method of moments, we denote $\nu(\mathbf{x}) = c(q_0(\mathbf{x}), q_1(\mathbf{x}), y_{q(\mathbf{x})}(\mathbf{x}), y_{1-q(\mathbf{x})}(\mathbf{x}))$ as the nuisance parameters and $\theta(\mathbf{x}) = c(\theta_0(\mathbf{x}), \theta_1^L(\mathbf{x}), \theta_1^U(\mathbf{x}))$ as the target parameters. Note that $q(\mathbf{x}) = \frac{q_0(\mathbf{x})}{q_1(\mathbf{x})}$ and $p(\mathbf{x}) = P(D = 1 | \mathbf{X} = \mathbf{x})$. To facilitate illustration, we focus on the modified target parameters $\tilde \theta(\mathbf{x}) = c(\tilde \theta_0(\mathbf{x}), \tilde \theta_1^L(\mathbf{x}), \tilde \theta_1^U(\mathbf{x})) = c(q_0(\mathbf{x})\theta_0(\mathbf{x}), q_0(\mathbf{x})\theta_1^L(\mathbf{x}), q_0(\mathbf{x})\theta_1^U(\mathbf{x}))$. We use the superscript $^0$ to denote the true value of a parameter. From their definition, we can see that the nuisance parameters satisfy the following local moment conditions:
$$
\begin{aligned}
& E\left[m^{(1)}(\mathbf{O}_i; \nu(\mathbf{x})) | \mathbf{X}_i = \mathbf{x}\right] = E\left[S_i(1 - D_i) - q_0(\mathbf{x})(1 - p(\mathbf{x})) | \mathbf{X}_i = \mathbf{x}\right],\\
& E\left[m^{(2)}(\mathbf{O}_i; \nu(\mathbf{x})) | \mathbf{X}_i = \mathbf{x}\right] = E\left[S_i D_i - q_1(\mathbf{x})p(\mathbf{x}) | \mathbf{X}_i = \mathbf{x}\right],\\
& E\left[m^{(3)}(\mathbf{O}_i; \nu(\mathbf{x})) | \mathbf{X}_i = \mathbf{x}\right] = E\left[S_i D_i\mathbf{1} \{Y_i \leq y_{q(\mathbf{x})}(\mathbf{x})\} - q_0(\mathbf{x})p(\mathbf{x}) | \mathbf{X}_i = \mathbf{x}\right],\\
& E\left[m^{(4)}(\mathbf{O}_i; \nu(\mathbf{x})) | \mathbf{X}_i = \mathbf{x}\right] = E\left[S_i D_i\mathbf{1} \{Y_i \geq y_{1-q(\mathbf{x})}(\mathbf{x})\} - q_0(\mathbf{x})p(\mathbf{x}) | \mathbf{X}_i = \mathbf{x}\right]. \\
\end{aligned}
$$
For the target parameters, we can show that under Assumptions 1 and 2, they satisfy the following local moment conditions:
$$
\begin{aligned}
& E\left[\psi^{(1)}(\mathbf{O}_i; \tilde \theta(\mathbf{x}), \nu(\mathbf{x})) | \mathbf{X}_i = \mathbf{x}\right] = E\left[S_{i}(1-D_{i})Y_{i} - (1-p(\mathbf{x}))\tilde \theta_{0}(\mathbf{x}) | \mathbf{X}_i = \mathbf{x}\right] \\
= & (1-p(\mathbf{x}))\left[q_0^0(\mathbf{x})\int yf_{Y|D=0, S=1, \mathbf{X} = \mathbf{x}}(y)dy - \tilde \theta_{0}(\mathbf{x}) \right], \\
& E\left[\psi^{(2)}(\mathbf{O}_i; \tilde \theta(\mathbf{x}), \nu(\mathbf{x})) | \mathbf{X}_i = \mathbf{x}\right] = E\left[ S_i D_i Y_{i} \mathbf{1}\{Y_i \leq y_{q(\mathbf{x})}(\mathbf{x})\} - p(\mathbf{x})\tilde \theta_{1}^{L}(\mathbf{x}) | \mathbf{X}_i = \mathbf{x}\right] \\
= & p(\mathbf{x})\left[q_1^0(\mathbf{x})\int_{-\infty}^{y_{q(\mathbf{x})}(\mathbf{x})} yf_{Y|D=1, S=1, \mathbf{X} = \mathbf{x}}(y)dy - \tilde \theta_{1}^L(\mathbf{x}) \right], \\
& E\left[\psi^{(3)}(\mathbf{O}_i; \tilde \theta(\mathbf{x}), \nu(\mathbf{x})) | \mathbf{X}_i = \mathbf{x}\right] = E\left[ S_i D_i Y_{i} \mathbf{1}\{Y_i \geq y_{1-q(\mathbf{x})}(\mathbf{x})\} - p(\mathbf{x})\tilde \theta_{1}^{U}(\mathbf{x}) | \mathbf{X}_i = \mathbf{x}\right] \\
= & p(\mathbf{x})\left[q_1^0(\mathbf{x})\int_{y_{1-q(\mathbf{x})}(\mathbf{x})}^{\infty} yf_{Y|D=1, S=1, \mathbf{X} = \mathbf{x}}(y)dy - \tilde \theta_{1}^U(\mathbf{x}) \right].
\end{aligned}
$$
When $\mathbf{X}$ is discrete or low-dimensional, we can approximate each of the moment conditions with stratification or kernel regression \citep{lee2009training, olma2020nonparametric}. We then solve $\nu(\mathbf{x})$ from their moment conditions and plug their values into the moment conditions above to solve $\theta(\mathbf{x})$. Nevertheless, with a large number of covariates, this approach suffers from the ``curse of dimensionality'' \citep{robins1997toward}. Therefore, we adopt the generalized random forest algorithm (\textit{grf}) \citep{wager2018estimation, athey2019generalized} to approximate each moment condition with the following expression:
$$
E[\psi_{\theta(\mathbf{x}), \nu(\mathbf{x})}(\mathbf{O}_i) | \mathbf{X}_i = \mathbf{x}] \approx \sum_{i=1}^N \alpha_i(\mathbf{x})\psi_{\theta(\mathbf{x}), \nu(\mathbf{x})}(\mathbf{O}_i)
$$
where $\psi(\cdot)$ represents the moment function, $\mathbf{O}_i = (Y_i, S_i, D_i)$, and $\alpha_i(\mathbf{x})$ is an adaptive kernel that weights the contribution of each observation $i$ to the moment condition evaluated at $\mathbf{x}$. 
The {\it grf} estimation of $\alpha_i(\mathbf{x})$ is ``honest'' in the sense that only half of each sub-sample is used to train the tree and the other half is reserved for making predictions. We can estimate any parameter by minimizing its approximated local moment condition's $L_2$ norm, $||\sum_{i=1}^N \alpha_i(\mathbf{x})\psi_{\theta(\mathbf{x}), \hat \nu(\mathbf{x})}(\mathbf{O}_i)||_2$.

Nevertheless, as illustrated in the main text, using the same sample to estimate both $\nu(\mathbf{x})$ and $\theta(\mathbf{x})$ leads to the regularization bias. It can be avoided if we implement both cross-fitting and Neyman orthogonalization \citep{belloni2017program, chernozhukov2018double}. The idea behind Neyman orthogonalization is to adjust the local moment conditions for $\theta(\mathbf{x})$ such that they are invariant to small perturbations of the estimated $\nu(\mathbf{x})$. To be more precise, If a moment condition $\psi(\cdot)$ satisfies Neyman orthogonality, then the following two equations must hold:
$$
E[\psi_{\theta^0(\mathbf{x}), \nu^0(\mathbf{x})}(\mathbf{O}_i) | \mathbf{X}_i = \mathbf{x}] = 0,
$$
and
$$
\frac{\partial}{\partial \nu} E[\psi_{\theta^{0}(\mathbf{x}), \nu(\mathbf{x})}(\mathbf{O}_i) | \mathbf{X}_i = \mathbf{x}] |_{\nu(\mathbf{x})=\nu^{0}(\mathbf{x})} = 0.
$$
The process of modifying a moment condition such that it satisfies the two equations is termed as Neyman orthogonalization.

Note that our algorithm is essentially a three-step process: we first estimate $q(\mathbf{x})$, then ($y_{q(\mathbf{x})}(\mathbf{x}), y_{1-q(\mathbf{x})}(\mathbf{x})$), finally the conditional or aggregated bounds, with the estimates of nuisance parameters plugged in. Therefore, to conduct Neyman orthogonalization, we first orthogonalize $m^{(3)}(\mathbf{O}_i; \nu(\mathbf{x}))$ and $m^{(4)}(\mathbf{O}_i; \nu(\mathbf{x}))$, and then $\psi^{(2)}(\mathbf{O}_i; \tilde \theta(\mathbf{x}), \nu(\mathbf{x}))$ and $\psi^{(3)}(\mathbf{O}_i; \tilde \theta(\mathbf{x}), \nu(\mathbf{x}))$ with regards to the nuisance parameters. We use $\tilde{m}^{(3)}(\mathbf{O}_i; \nu(\mathbf{x}))$, $\tilde{m}^{(4)}(\mathbf{O}_i; \nu(\mathbf{x}))$, $\tilde{\psi}^{(2)}(\mathbf{O}_i; \tilde \theta(\mathbf{x}), \nu(\mathbf{x}))$, and $\tilde{\psi}^{(3)}(\mathbf{O}_i; \tilde \theta(\mathbf{x}), \nu(\mathbf{x}))$ to denote the orthogonalized moments. We have
$$
\begin{aligned}
& \tilde{m}^{(3)}(\mathbf{O}_i; \nu(\mathbf{x})) = m^{(3)}(\mathbf{O}_i; \nu(\mathbf{x})) - \frac{p(\mathbf{x})}{1-p(\mathbf{x})}m^{(1)}(\mathbf{O}_i; \nu(\mathbf{x})) \\
= & S_i D_i\mathbf{1} \{Y_i \leq y_{q(\mathbf{x})}(\mathbf{x})\} - \frac{S_i(1-D_i)p(\mathbf{x})}{1-p(\mathbf{x})} \\
& \tilde{m}^{(4)}(\mathbf{O}_i; \nu(\mathbf{x})) = m^{(4)}(\mathbf{O}_i; \nu(\mathbf{x})) - \frac{p(\mathbf{x})}{1-p(\mathbf{x})}m^{(1)}(\mathbf{O}_i; \nu(\mathbf{x})) \\
= & S_i D_i\mathbf{1} \{Y_i \geq y_{1-q(\mathbf{x})}(\mathbf{x})\} - \frac{S_i(1-D_i)p(\mathbf{x})}{1-p(\mathbf{x})} \\
& \tilde{\psi}^{(2)}(\mathbf{O}; \tilde \theta(\mathbf{x}), \nu(\mathbf{x})) = \psi^{(2)}(\mathbf{O}; \tilde \theta(\mathbf{x}), \nu(\mathbf{x})) - y^0_{q(\mathbf{x})}(\mathbf{x})\tilde{m}^{(3)}_i(\mathbf{O}_i; \nu(\mathbf{x})) \\
= & S_i D_i (Y_{i} - y^0_{q(\mathbf{x})}(\mathbf{x})) \mathbf{1}\{Y_i \leq y_{q(\mathbf{x})}(\mathbf{x})\} - p(\mathbf{x})\tilde \theta_{1}^{L}(\mathbf{x}) + \frac{y^0_{q(\mathbf{x})}(\mathbf{x})p(\mathbf{x})S_i(1-D_i)}{1-p(\mathbf{x})} \\
& \tilde{\psi}^{(3)}(\mathbf{O}; \tilde \theta(\mathbf{x}), \nu(\mathbf{x})) = \psi^{(3)}(\mathbf{O}; \tilde \theta(\mathbf{x}), \nu(\mathbf{x})) - y^0_{1-q(\mathbf{x})}(\mathbf{x})\tilde{m}^{(4)}_i(\mathbf{O}_i; \nu(\mathbf{x})) \\
= & S_i D_i (Y_{i} - y^0_{1-q(\mathbf{x})}(\mathbf{x})) \mathbf{1}\{Y_i \geq y_{1-q(\mathbf{x})}(\mathbf{x})\} - p(\mathbf{x})\tilde \theta_{1}^{U}(\mathbf{x}) + \frac{y^0_{1-q(\mathbf{x})}(\mathbf{x})p(\mathbf{x})S_i(1-D_i)}{1-p(\mathbf{x})}.
\end{aligned}
$$
As each of the orthogonalized moments is a linear combination of two original moments, the first requirement is satisfied. To verify that the second requirement is also satisfied, first note that $\tilde{m}^{(3)}(\mathbf{O}_i; \nu(\mathbf{x}))$ and $\tilde{m}^{(4)}(\mathbf{O}_i; \nu(\mathbf{x}))$ no longer depend on $q(\mathbf{x})$. Furthermore,
$$
\begin{aligned}
& \frac{\partial E\left[\tilde{\psi}^{(2)}(\mathbf{O}_i; \tilde \theta(\mathbf{x}), \nu(\mathbf{x})) | \mathbf{X}_i = \mathbf{x}\right]}{\partial y_{q(\mathbf{x})}(\mathbf{x})}\Bigg|_{\nu(\mathbf{x}) = \nu^0(\mathbf{x})} \\
= & \frac{\partial E\left[\psi^{(2)}(\mathbf{O}_i; \tilde \theta(\mathbf{x}), \nu(\mathbf{x})) | \mathbf{X}_i = \mathbf{x}\right]}{\partial y_{q(\mathbf{x})}(\mathbf{x})}\Bigg|_{\nu(\mathbf{x}) = \nu^0(\mathbf{x})} - y^0_{q(\mathbf{x})}(\mathbf{x})\frac{\partial E\left[\tilde{m}^{(3)}(\mathbf{O}_i; \nu(\mathbf{x})) | \mathbf{X}_i = \mathbf{x}\right]}{\partial y_{q(\mathbf{x})}(\mathbf{x})}\Bigg|_{\nu(\mathbf{x}) = \nu^0(\mathbf{x})} \\
= & p(\mathbf{x})q_1^0(\mathbf{x}) y_{q(\mathbf{x})}^0(\mathbf{x})f_{Y|D=1, S=1, \mathbf{X} = \mathbf{x}}(y^0_{q(\mathbf{x})}(\mathbf{x})) - y_{q(\mathbf{x})}^0(\mathbf{x})p(\mathbf{x})q_1^0(\mathbf{x}) f_{Y|D=1, S=1, \mathbf{X} = \mathbf{x}}(y_{q(\mathbf{x})}^0(\mathbf{x})) = 0 \\
& \frac{\partial E\left[\psi^{(3)}(\mathbf{O}_i; \tilde \theta(\mathbf{x}), \nu(\mathbf{x})) | \mathbf{X}_i = \mathbf{x}\right]}{\partial y_{1-q(\mathbf{x})}(\mathbf{x})}\Bigg|_{\nu(\mathbf{x}) = \nu^0(\mathbf{x})} \\
= & \frac{\partial E\left[\psi^{(3)}(\mathbf{O}_i; \tilde \theta(\mathbf{x}), \nu(\mathbf{x})) | \mathbf{X}_i = \mathbf{x}\right]}{\partial y_{1-q(\mathbf{x})}(\mathbf{x})}\Bigg|_{\nu(\mathbf{x}) = \nu^0(\mathbf{x})} - y^0_{1-q(\mathbf{x})}(\mathbf{x})\frac{\partial E\left[\tilde{m}^{(4)}(\mathbf{O}_i; \nu(\mathbf{x})) | \mathbf{X}_i = \mathbf{x}\right]}{\partial y_{1-q(\mathbf{x})}(\mathbf{x})}\Bigg|_{\nu(\mathbf{x}) = \nu^0(\mathbf{x})} \\
= & -p(\mathbf{x})q_1^0(\mathbf{x}) y_{1-q(\mathbf{x})}^0(\mathbf{x})f_{Y|D=1, S=1, \mathbf{X} = \mathbf{x}}(y^0_{1-q(\mathbf{x})}(\mathbf{x})) \\
& + y_{1-q(\mathbf{x})}^0(\mathbf{x})p(\mathbf{x})q_1^0(\mathbf{x}) f_{Y|D=1, S=1, \mathbf{X} = \mathbf{x}}(y^0_{1-q(\mathbf{x})}(\mathbf{x})) = 0.
\end{aligned}
$$
Hence, $\tilde{m}^{(3)}(\mathbf{O}_i; \nu(\mathbf{x}))$, $\tilde{m}^{(4)}(\mathbf{O}_i; \nu(\mathbf{x}))$, $\tilde{\psi}^{(2)}(\mathbf{O}_i; \tilde \theta(\mathbf{x}), \nu(\mathbf{x}))$, and $\tilde{\psi}^{(3)}(\mathbf{O}_i; \tilde \theta(\mathbf{x}), \nu(\mathbf{x}))$ all satisfy Neyman orthogonality, which justifies our algorithm in the main text. It is worth noting that $E\left[\frac{S_i(1-D_i)}{1-p(\mathbf{x})} | \mathbf{X}_i = \mathbf{x}\right] = q_0^0(\mathbf{x})$. Thus, Neyman orthogonalization does not affect the estimation of $y_{q(\mathbf{x})}(\mathbf{x})$ or $y_{1-q(\mathbf{x})}(\mathbf{x})$. For the conditional lower bound $\tilde \theta_{1}^{L}(\mathbf{x})$, we know that
$$
\begin{aligned}
& E\left[\tilde{\psi}^{(2)}(\mathbf{O}; \tilde \theta(\mathbf{x}), \nu(\mathbf{x})) | \mathbf{X}_i = \mathbf{x}\right] \\
= & E\left[S_i D_i (Y_{i} - y^0_{q(\mathbf{x})}(\mathbf{x})) \mathbf{1}\{Y_i \leq y_{q(\mathbf{x})}(\mathbf{x})\} - p(\mathbf{x})\tilde \theta_{1}^{L}(\mathbf{x}) + y^0_{q(\mathbf{x})}(\mathbf{x})p(\mathbf{x})q_0^0(\mathbf{x}) | \mathbf{X}_i = \mathbf{x}\right] \\
= & E\left[(Y_{i} - y^0_{q(\mathbf{x})}(\mathbf{x})) \mathbf{1}\{Y_i \leq y_{q(\mathbf{x})}(\mathbf{x})\} | S_i = 1, D_i = 1, \mathbf{X}_i = \mathbf{x}\right]p(\mathbf{x})q_1^0(\mathbf{x}) \\
& - p(\mathbf{x})\tilde \theta_{1}^{L}(\mathbf{x}) + y^0_{q(\mathbf{x})}(\mathbf{x})p(\mathbf{x})q_0^0(\mathbf{x}).
\end{aligned}
$$
Therefore,
$$
\begin{aligned}
& \theta_{1}^{L}(\mathbf{x}) = \frac{E\left[(Y_{i} - y^0_{q(\mathbf{x})}(\mathbf{x})) \mathbf{1}\{Y_i \leq y_{q(\mathbf{x})}(\mathbf{x})\} | S_i = 1, D_i = 1, \mathbf{X}_i = \mathbf{x}\right]}{q^0(\mathbf{x})} + y^0_{q(\mathbf{x})},
\end{aligned}
$$
and a similar formula holds for $\theta_{1}^{U}(\mathbf{x})$.

For the aggregated bounds, we conduct Neyman orthogonalization on the unconditional moments, which leads to the following score functions:
$$
\begin{aligned}
& s^L(\mathbf{X}_i) = \frac{S_i D_i (Y_{i} - \hat y_{\hat q(\mathbf{X}_i)}(\mathbf{X}_i)) \mathbf{1}\{Y_i \leq \hat y_{\hat q(\mathbf{X}_i)}(\mathbf{X}_i))\}}{p(\mathbf{X}_i)} - \frac{S_i(1-D_i)(Y_i - \hat{y}_{\hat q(\mathbf{X}_i)}(\mathbf{X}_i))}{1-p(\mathbf{X}_i)} \\
& s^U(\mathbf{X}_i) = \frac{S_i D_i (Y_{i} - \hat{y}_{1-\hat{q}(\mathbf{X}_i)}(\mathbf{X}_i)) \mathbf{1}\{Y_i \geq \hat{y}_{1-\hat{q}(\mathbf{X}_i)}(\mathbf{X}_i)\}}{p(\mathbf{X}_i)} - \frac{S_i(1-D_i)(Y_i - \hat{y}_{1-\hat{q}(\mathbf{X}_i)}(\mathbf{X}_i))}{1-p(\mathbf{X}_i)}.
\end{aligned}
$$
Then,
$$
\begin{aligned}
\hat{\tau}^L_{CTB}(1,1) = & \frac{\frac{1}{N}\sum_{i=1}^N s^L(\mathbf{X}_i)}{\hat q_0}, \hat{\tau}^U_{CTB}(1,1) = \frac{\frac{1}{N}\sum_{i=1}^N s^U(\mathbf{X}_i)}{\hat q_0},
\end{aligned}
$$
where $\hat q_0 = \frac{1}{N} \sum_{i=1}^N \frac{S_i(1-D_i)}{1-p(\mathbf{X}_i)}$.

When the propensity scores need to be estimated, we have an extra moment condition
$$
E\left[m^{(0)}(\mathbf{O}_i; \nu(\mathbf{x})) | \mathbf{X}_i = \mathbf{x}\right] = E[D_i - p(\mathbf{x}) | \mathbf{X}_i = \mathbf{x}]  = 0.
$$
We need to orthogonalize our moment conditions to account for the estimation error from $m^{(0)}(\cdot)$. It turns out that there will be an extra correction term in $\tilde{\psi}^{(1)}(\cdot)$, $\tilde{\psi}^{(2)}(\cdot)$, and $\tilde{\psi}^{(3)}(\cdot)$, which are $\tilde \theta_0(\mathbf{x})[D_i - p(\mathbf{x})]$, $y_{q(\mathbf{x})}q_0(\mathbf{x})\left[\frac{D_i}{p(\mathbf{x})} - \frac{1-D_i}{1-p(\mathbf{x})}\right] - \tilde \theta_1^L(\mathbf{x})[D_i - p(\mathbf{x})]$ and $y_{1-q(\mathbf{x})}q_0(\mathbf{x})\left[\frac{D_i}{p(\mathbf{x})} - \frac{1-D_i}{1-p(\mathbf{x})}\right] - \tilde \theta_1^U(\mathbf{x})[D_i - p(\mathbf{x})]$, respectively.

If the assumption of monotonic selection holds for the other direction: $S_i(0) \geq S_i(1)$, then the moment conditions $m^{(3)}(\mathbf{O}_i; \nu(\mathbf{x}))$, $m^{(4)}(\mathbf{O}_i; \nu(\mathbf{x}))$, and $\psi(\mathbf{O};  \tilde \theta(\mathbf{x}), \nu(\mathbf{x}))$ become:
$$
\begin{aligned}
E\left[m^{(3)}_i(\mathbf{O}_i; \nu(\mathbf{x}))\right] & = E\left[S_i (1-D_i)\mathbf{1} \{Y_i \leq y_{q(\mathbf{x})}(\mathbf{x})\} - q_1^0(\mathbf{x})(1-p(\mathbf{x})) | \mathbf{X}_i = \mathbf{x}\right],\\
E\left[m^{(4)}_i(\mathbf{O}_i; \nu(\mathbf{x}))\right] & = E\left[S_i (1-D_i\mathbf{1}) \{Y_i \geq y_{1-q(\mathbf{x})}(\mathbf{x})\} - q_1^0(\mathbf{x})(1-p(\mathbf{x})) | \mathbf{X}_i = \mathbf{x}\right]\\
E\left[\psi^{(1)}(\mathbf{O}; \tilde \theta(\mathbf{x}), \nu(\mathbf{x}))\right] & = E\left[S_{i}D_{i}Y_{i} - p(\mathbf{x})\tilde \theta_{1}(\mathbf{x}) | \mathbf{X}_i = \mathbf{x}\right] \\
& = p(\mathbf{x})\left[q_1^0(\mathbf{x}) \int yf_{Y|D=1, S=1, \mathbf{X} = \mathbf{x}}(y)dy - \tilde \theta_{1}(\mathbf{x}) \right], \\
E\left[\psi^{(2)}(\mathbf{O}; \tilde \theta(\mathbf{x}), \nu(\mathbf{x}))\right] & = E\left[ S_i (1-D_i) Y_{i} \mathbf{1}\{Y_i \leq y_{q(\mathbf{x})}(\mathbf{x})\} - (1-p(\mathbf{x}))\tilde \theta_{0}^{L}(\mathbf{x}) | \mathbf{X}_i = \mathbf{x}\right] \\
& = (1-p(\mathbf{x}))\left[q_0^0(\mathbf{x})\int_{-\infty}^{y_{q(\mathbf{x})}(\mathbf{x})} yf_{Y|D=0, S=1, \mathbf{X} = \mathbf{x}}(y)dy - \tilde \theta_{0}^L(\mathbf{x}) \right], \\
E\left[\psi^{(3)}(\mathbf{O}; \tilde \theta(\mathbf{x}), \nu(\mathbf{x}))\right] & = E\left[ S_i (1-D_i) Y_{i} \mathbf{1}\{Y_i \geq y_{1-q(\mathbf{x})}(\mathbf{x})\} - (1-p(\mathbf{x}))\tilde \theta_{0}^{U}(\mathbf{x}) | \mathbf{X}_i = \mathbf{x}\right] \\
& = (1-p(\mathbf{x}))\left[q_0^0(\mathbf{x})\int_{y_{1-q(\mathbf{x})}(\mathbf{x})}^{\infty} yf_{Y|D=0, S=1, \mathbf{X} = \mathbf{x}}(y)dy - \tilde \theta_{0}^U(\mathbf{x}) \right].
\end{aligned}
$$
We can similarly obtain the orthogonalized moment conditions:
$$
\begin{aligned}
& \tilde{\psi}^{(2)}(\mathbf{O}_i; \tilde \theta(\mathbf{x}), \nu(\mathbf{x})) \\
= & S_i (1-D_i) (Y_{i} - y^0_{q(\mathbf{x})}(\mathbf{x})) \mathbf{1}\{Y_i \leq y_{q(\mathbf{x})}(\mathbf{x})\} - (1-p(\mathbf{x}))\tilde \theta_{1}^{L}(\mathbf{x}) + \frac{(1-p(\mathbf{x}))y^0_{q(\mathbf{x})}(\mathbf{x})S_iD_i}{p(\mathbf{x})} \\
& \tilde{\psi}^{(3)}(\mathbf{O}; \tilde \theta(\mathbf{x}), \nu(\mathbf{x})) \\
= & S_i (1-D_i) (Y_{i} - y^0_{1-q(\mathbf{x})}(\mathbf{x})) \mathbf{1}\{Y_i \geq y_{1-q(\mathbf{x})}(\mathbf{x})\} - (1-p(\mathbf{x}))\tilde \theta_{1}^{U}(\mathbf{x}) \\
& + \frac{(1-p(\mathbf{x}))y^0_{1-q(\mathbf{x})}(\mathbf{x})S_i D_i}{p(\mathbf{x})}.
\end{aligned}
$$
When we have conditionally monotonic selection, $S_i(0) \leq S_i(1)$ for certain units and $S_i(0) \geq S_i(1)$ for the others. We define a binary variable $direction_i \coloneqq \mathbf{1}\{\mathbf{X}_i \in \mathcal{X}^{+}\}$. When $\mathbf{X}_i \in \mathcal{X}^{+}$, we have score functions $s^{L, help}(\mathbf{X}_i)$ and $s_i^{U, help}(\mathbf{X}_i)$. When $\mathbf{X}_i \in \mathcal{X}^{-}$, we have $s^{L, hurt}(\mathbf{X}_i)$ and $s_i^{U, hurt}(\mathbf{X}_i)$. Finally, $s^{L}(\mathbf{X}_i) = direction_i * s^{L, help}(\mathbf{X}_i) + (1-direction_i) * s^{L, hurt}(\mathbf{X}_i)$ and $s^{U}(\mathbf{X}_i) = direction_i * s^{U, help}(\mathbf{X}_i) + (1-direction_i) * s^{U, hurt}(\mathbf{X}_i)$. 

\subsection{Statistical theory}\label{appx:A5}
We first derive the asymptotic behavior of the estimator for the aggregated bounds, using Theorem 3.1 in \citet{chernozhukov2018double}. Our moment conditions are linear in the target parameters and satisfy Neyman orthogonality, hence Assumption 3.1 required by the theorem holds. Another required assumption, Assumption 3.2, is ensured by the convergence rate of the \textit{grf} algorithm. Therefore, under monotonic selection, we know that
$$
\begin{aligned}
& \frac{1}{\sqrt{N}} \sum_{i=1}^N \left[s^L(\mathbf{X}_i) - q_0 \tau^L_{CTB}(1,1)\right] \rightarrow \mathcal{N}(0, V_s^L) \\
& \frac{1}{\sqrt{N}} \sum_{i=1}^N \left[s^U(\mathbf{X}_i) - q_0 \tau^U_{CTB}(1,1)\right] \rightarrow \mathcal{N}(0, V_s^U) \\
& \sqrt{N} (\hat q_0 - q_0) \rightarrow \mathcal{N}(0, V_{q_0}),
\end{aligned}
$$
Using the Delta method, we can obtain
$$
\begin{aligned}
& \frac{1}{\sqrt{N}} \sum_{i=1}^N \left(\hat{\tau}^L_{CTB}(1,1) - \tau^L_{CTB}(1,1)\right) \rightarrow \mathcal{N}(0, V^L) \\
& \frac{1}{\sqrt{N}} \sum_{i=1}^N \left(\hat{\tau}^U_{CTB}(1,1) - \tau^U_{CTB}(1,1)\right) \rightarrow \mathcal{N}(0, V^U),
\end{aligned}
$$
% where $V^L = \frac{V_s^L}{q_0^2} + (\frac{\tau^L}{q_0^2})^2V_{q0}$ and $V^U = \frac{V_s^U}{q_0^2} + (\frac{\tau^U}{q_0^2})^2V_{q0}$. 
We can estimate $V^L$ via its sample analogue:
$$
\begin{aligned}
& \hat V^L = \begin{pmatrix}\frac{1}{\hat q_0} & -\frac{\hat \tau^L_{CTB}(1,1)}{\hat q_0} \end{pmatrix}\begin{pmatrix}\hat V_s^L & \widehat{Cov}_{(s^L, q_0)} \\ \widehat{Cov}_{(s^L, q_0)} & \hat V_{q_0} \end{pmatrix}\begin{pmatrix}\frac{1}{\hat q_0} \\ -\frac{\hat \tau^L_{CTB}(1,1)}{\hat q_0} \end{pmatrix},
\end{aligned}
$$
where
$$
\begin{aligned}
& \hat V_s^L = \frac{1}{N^2} \sum_{i=1}^N \left[s^L(\mathbf{X}_i) - \frac{1}{N} \sum_{i=1}^N s^L(\mathbf{X}_i)\right]^2, \hat V_{q_0} = \frac{1}{N^2} \sum_{i=1}^N \left[\frac{S_i(1-D_i)}{1-p(\mathbf{X}_i)} - \hat q_0 \right]^2, \\
& \widehat{Cov}_{(s^L, q_0)} = \frac{1}{N^2} \sum_{i=1}^N \left[s^L(\mathbf{X}_i) - \frac{1}{N} \sum_{i=1}^N s^L(\mathbf{X}_i)\right] \left[\frac{S_i(1-D_i)}{1-p(\mathbf{X}_i)} - \hat q_0 \right].
\end{aligned}
$$
Similar results hold under conditionally monotonic selection, as the estimate is a weighted average of two estimates under monotonic selection.
% Then,
% $$
% \begin{aligned}
% & \hat{\tau}^L(1,1) = \frac{\frac{1}{N}\sum_{i=1}^N s^L(\mathbf{X}_i)}{\hat q_0}, \hat{\tau}^U(1,1) = \frac{\frac{1}{N}\sum_{i=1}^N s^U(\mathbf{X}_i)}{\hat q_0}, \\
% & \sqrt{N}(\hat{\tau}^L(1,1) - \tau^L(1,1)) \rightarrow \mathcal{N}(0, V^L), \sqrt{N}(\hat{\tau}^U(1,1) - \tau^U(1,1)) \rightarrow \mathcal{N}(0, V^U)
% \end{aligned}
% $$

For the conditional bounds, it is easy to verify that the moment conditions for $\nu(\mathbf{x})$ satisfy all the assumptions imposed in \citet{athey2019generalized}. Therefore, $\hat \nu(\mathbf{x})$ are consistent and asymptotically Normal, according to Theorem 5 in the paper. In addition, the theorem guarantees that these estimates converge to their true values at a sufficiently fast rate (higher than $1/N^{\frac{1}{4}}$). Since our algorithm incorporates both Neyman orthogonalization and cross-fitting, it is honest for the estimation of the target parameters. Consequently, the consistency and asymptotically Normality of $\hat \theta(\mathbf{x})$ and the conditional bounds, ($\hat{\tau}_{CTB, \mathbf{x}}^{U}(1,1), \hat{\tau}_{CTB, \mathbf{x}}^{L}(1,1)$), can be similarly derived from Theorem 5 in \citet{athey2019generalized}. To summarize, we have the following formal results:

\begin{theorem}\label{thm:asym}
Under regularity conditions required by Theorem 5 in \citet{athey2019generalized}, estimates of $\tau^{L}_{CTB, \mathbf{x}}(1,1)$, $\tau^{U}_{CTB, \mathbf{x}}(1,1)$, $\tau_{CTB}^{L}(1,1)$, and $\tau_{CTB}^{U}(1,1)$ from Algorithm 1 are consistent and asymptotically Normal:
$$
\frac{\hat \tau_{CTB, \mathbf{x}}^L(1,1) - \tau_{CTB, \mathbf{x}}^L(1,1)}{\sqrt{Var(\hat \tau_{CTB, \mathbf{x}}^L(1,1))}} \rightarrow \mathcal{N}(0, 1) \text{ ,} \frac{\hat \tau_{CTB, \mathbf{x}}^U(1,1) - \tau_{CTB, \mathbf{x}}^U(1,1)}{\sqrt{Var(\hat \tau_{CTB, \mathbf{x}}^U(1,1))}} \rightarrow \mathcal{N}(0, 1),
$$
$$
\frac{\hat \tau_{CTB}^L(1,1) - \tau_{CTB}^L(1,1)}{\sqrt{Var(\hat \tau_{CTB}^L(1,1))}} \rightarrow \mathcal{N}(0, 1) \text{ ,} \frac{\hat \tau_{CTB}^U(1,1) - \tau_{CTB}^U(1,1)}{\sqrt{Var(\hat \tau_{CTB}^U(1,1))}} \rightarrow \mathcal{N}(0, 1).
$$
\end{theorem}

Theorem \ref{thm:asym} allows us to construct valid confidence intervals for either the conditional bounds or the aggregated bounds using estimates from our algorithm.\footnote{As there is no guarantee that our lower bound is necessarily smaller than the upper bound, the super-efficiency condition required by \citet{imbens2004confidence} may not hold. Hence, a prefereable approach is to rely on the critical values in proposition 3 of \citet{stoye2009more}.} For the conditional bounds, the standard error estimates are available in the output of the regression forest. This is due to the fact that our moment conditions for the target parameters have been orthogonalized and accounted for the extra uncertainties caused by the estimation of $\nu(\mathbf{x})$ in the first stage. Finally, if we only have conditionally monotonic selection, some extra conditions are needed to ensure that we can predict $\mathcal{X}^{+}$ and $\mathcal{X}^{-}$ accurately. Details of these conditions are discussed in \citet{semenova2020better}.

\subsection{Unit missingness and missing covariates}\label{appx:A3}
When there exist missing covariates, we have:
$$
\begin{aligned}
q(\mathbf{x}) = &  P(S_i(1)=S_i(0)=1|\mathbf{X}_i = \mathbf{x},S=1,D=1) \\
= & \frac{P(S_i(0)=1|D=1,\mathbf{X}_i = \mathbf{x})}{P(S_i(1)=1|D=1,\mathbf{X}_i = \mathbf{x})} \\
= & \frac{P(S_i(0)=1|\mathbf{X}_i = \mathbf{x})}{P(S_i(1)=1|\mathbf{X}_i = \mathbf{x})} \\
= & \frac{p(x)P(D=0|S=1,\mathbf{X}_i = \mathbf{x})}{(1-p(x))P(D=1|S=1,\mathbf{X}_i = \mathbf{x})}\\
%= & 1- P(S_i(1)=0|D_i=0,\mathbf{X}_i = \mathbf{x},   S_i=1) \\
%= &1- \frac{P(S_i(1)=0 | \mathbf{X}_i = \mathbf{x},   S_i=1)}{P(D_i=1 | \mathbf{X}_i = \mathbf{x},  S_i=1)} \\
%= &1- \frac{P(D=1|\mathbf{X} = \mathbf{x},S_i=1)-p(\mathbf{x})}{(1-P(D=1|\mathbf{X} = \mathbf{x},S_i=1))p(\mathbf{x})},
\end{aligned}
$$
where the last equality comes from
$$
\begin{aligned}
P(D=1|S=1,\mathbf{X}_i = \mathbf{x})= & P(S=1|D=1,\mathbf{X}_i = \mathbf{x})\frac{P(D=1|\mathbf{X}_i = \mathbf{x})}{P(S=1|\mathbf{X}_i = \mathbf{x})} \\
= & P(S_i(1)=1|D=1,\mathbf{X}_i = \mathbf{x})\frac{p(x)}{P(S=1|\mathbf{X}_i = \mathbf{x})} \\
= & P(S_i(1)=1|\mathbf{X}_i = \mathbf{x})\frac{p(x)}{P(S=1|\mathbf{X}_i = \mathbf{x})}, \\
P(D=0|S=1,\mathbf{X}_i = \mathbf{x})= & P(S=1|D=0,\mathbf{X}_i = \mathbf{x})\frac{P(D=0|\mathbf{X}_i = \mathbf{x})}{P(S=1|\mathbf{X}_i = \mathbf{x})} \\
= & P(S_i(0)=1|D=0,\mathbf{X}_i = \mathbf{x})\frac{1-p(x)}{P(S=1|\mathbf{X}_i = \mathbf{x})} \\
= & P(S_i(0)=1|\mathbf{X}_i = \mathbf{x})\frac{1-p(x)}{P(S=1|\mathbf{X}_i = \mathbf{x})}. \\
%1-p(\mathbf{x}) = &  P(D_i=0 |\mathbf{X}_i = \mathbf{x})\\
%= & P(D_i=0 | S_i(1) = 1, S_i=1, \mathbf{X}_i = \mathbf{x}) P(S_i(1) = 1 | S_i=1, \mathbf{X}_i = \mathbf{x}) + \\ 
%& P(D_i=0 | S_i(1) = 0, S_i=1, \mathbf{X}_i = \mathbf{x}) P(S_i(1) = 0 | S_i=1, \mathbf{X}_i = \mathbf{x})\\
%= & P(D_i=0 | A, \mathbf{X}_i = \mathbf{x},S_i=1) P(A | S_i=1, \mathbf{X}_i = \mathbf{x}) + P(S_i(1) = 0 | S_i=1, \mathbf{X}_i = \mathbf{x})\\
%= & \frac{P(D_i=0,A |  \mathbf{X}_i = \mathbf{x},S_i=1)}{P(A|S_i=1,\mathbf{X}_i = \mathbf{x})} P(A | S_i=1, \mathbf{X}_i = \mathbf{x}) +  P(S_i(1) = 0 | S_i=1, \mathbf{X}_i = \mathbf{x})\\
%= & P(D_i=0 | \mathbf{X}_i = \mathbf{x},S_i=1) +  P(S_i(1) = 0 | S_i=1, \mathbf{X}_i = \mathbf{x}).
\end{aligned}
$$ 
Hence, we can estimate $q(\mathbf{x})$ via estimating two probability forest models, $P(D=1|\mathbf{X} = \mathbf{x},S_i=1)$ and $P(D=0|\mathbf{X} = \mathbf{x},S_i=1)$.

\newpage
\section{Extra results}\label{appx:B}
\subsection{Extra results from simulation}\label{appx:B1}
\begin{figure}[htp]
 \begin{center}
\includegraphics[width=2\linewidth, height=.8\linewidth]{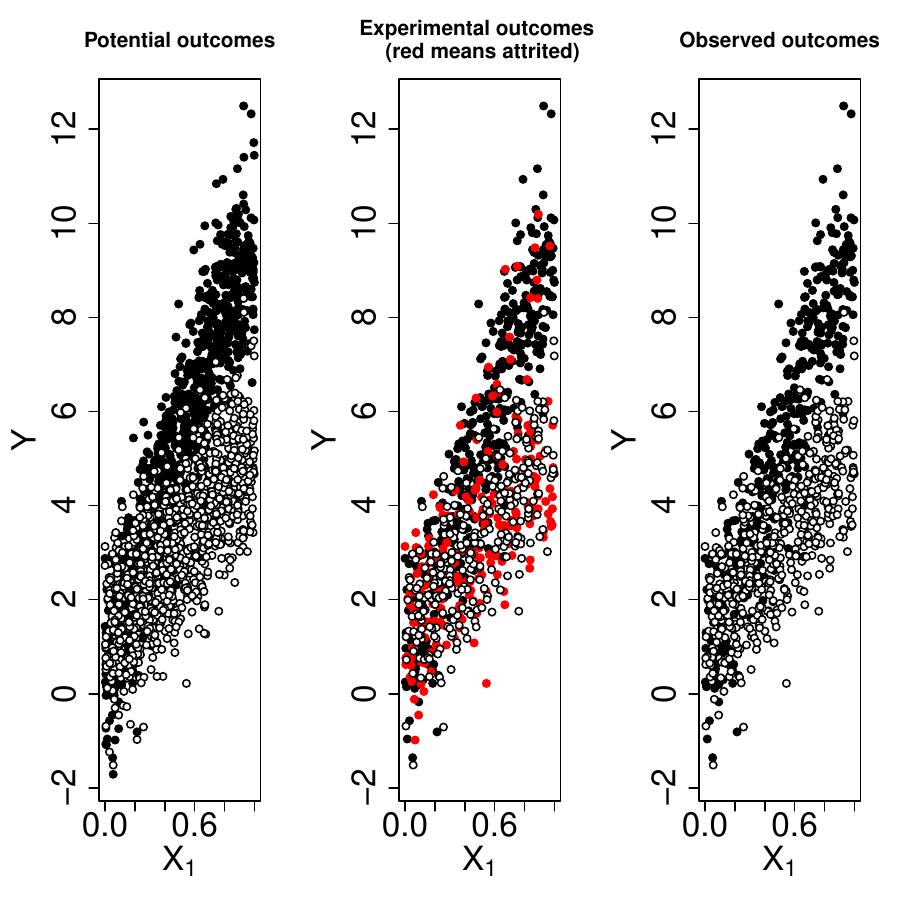}
\caption{Simulated Data}
\label{fig_simu_data}
 \end{center}
  \footnotesize\textbf{Note:} The left plot presents how $Y_i(0)$ (in white) and $Y_i(1)$ (in black) vary across the value of the first covariate for all the units in the simulated sample. The middle plot shows realized outcome for all the units under one assignment, with missing outcome values marked by red spots. The right plots shows outcome values that can be observed by the researcher.
\end{figure}

\begin{figure}[htp]
\caption{Estimates of Nuisance Parameters}
\label{fig_simu_qs}
 \begin{center}
\includegraphics[width=.48\linewidth, height=.4\linewidth]{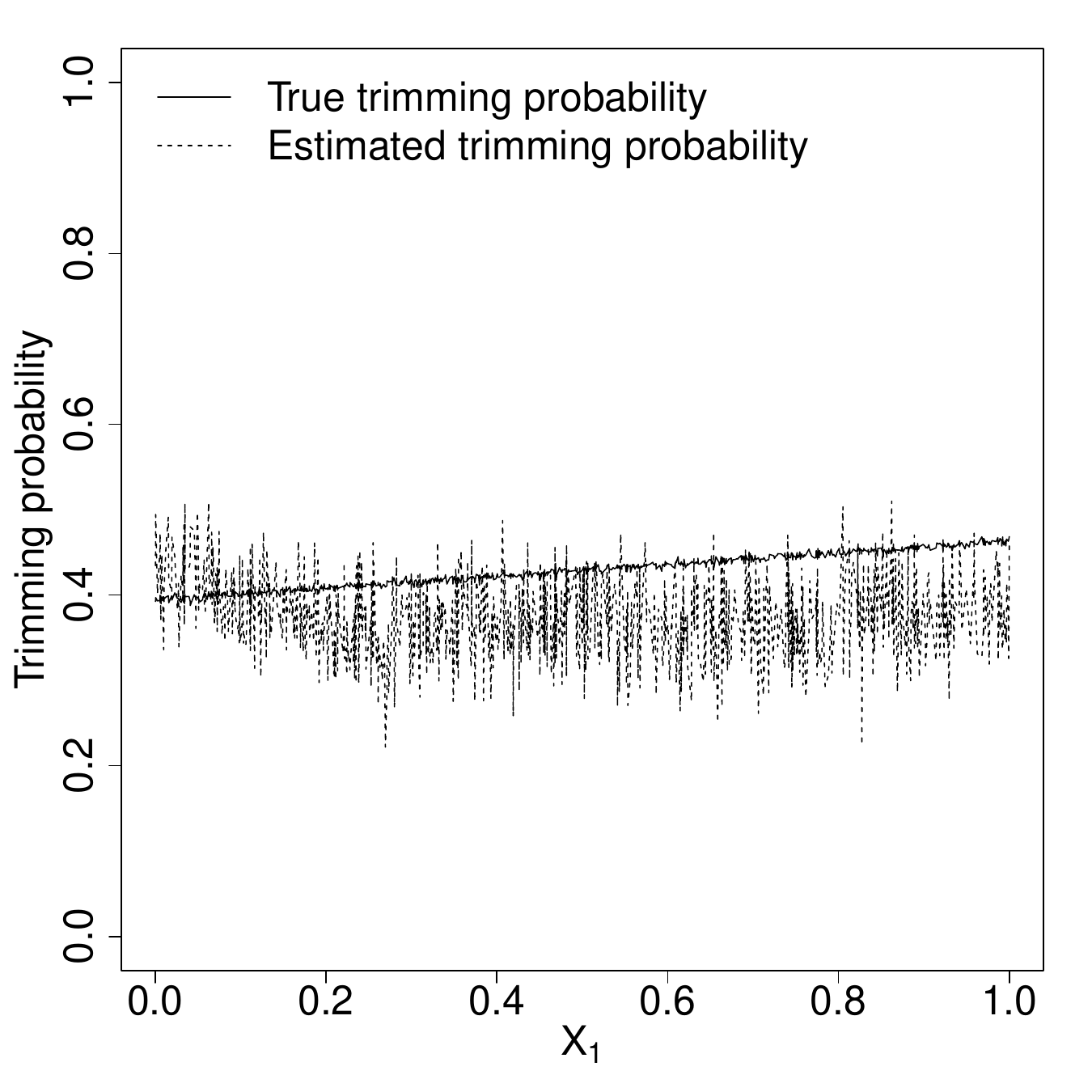} \\
\includegraphics[width=.48\linewidth, height=.4\linewidth]{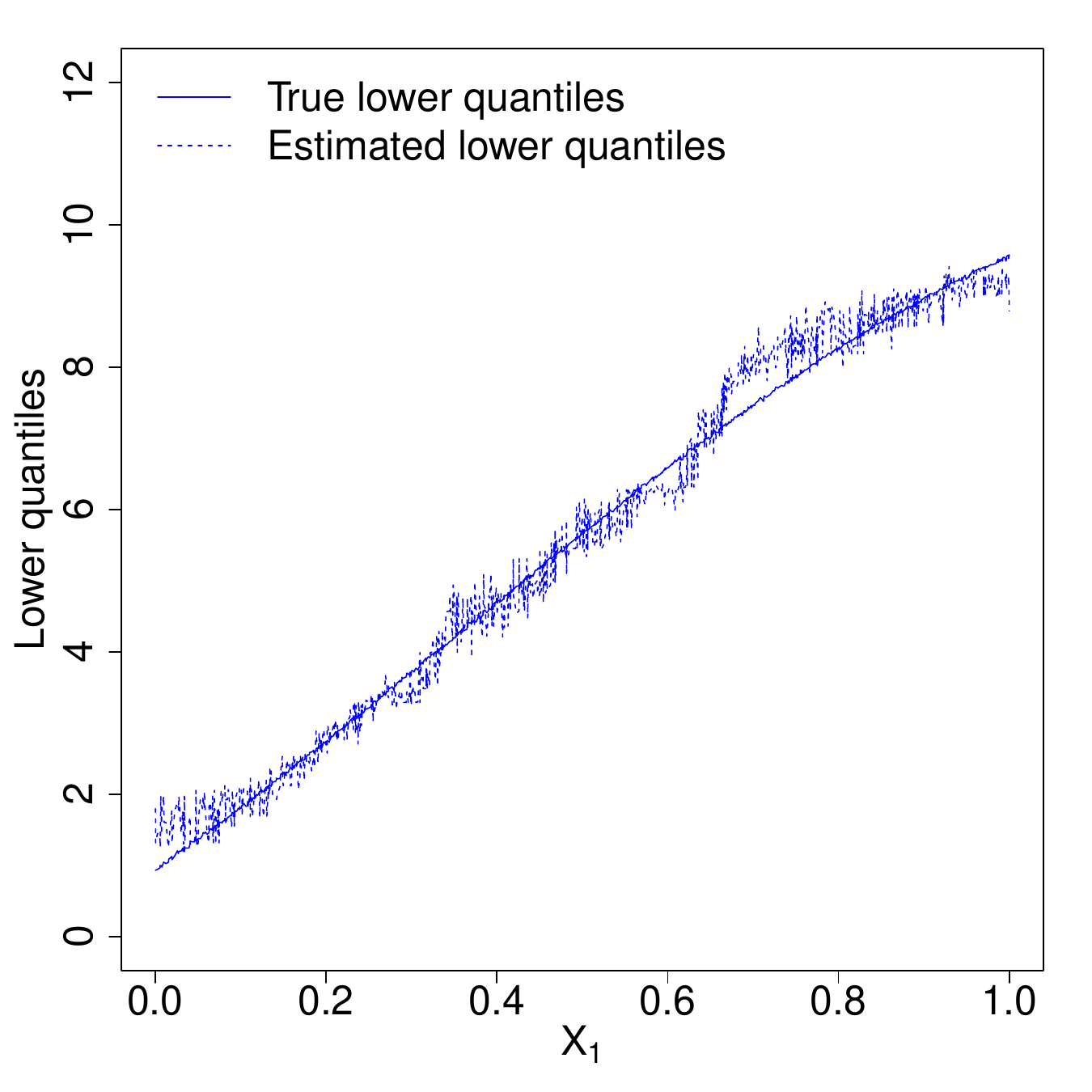}
\includegraphics[width=.48\linewidth, height=.4\linewidth]{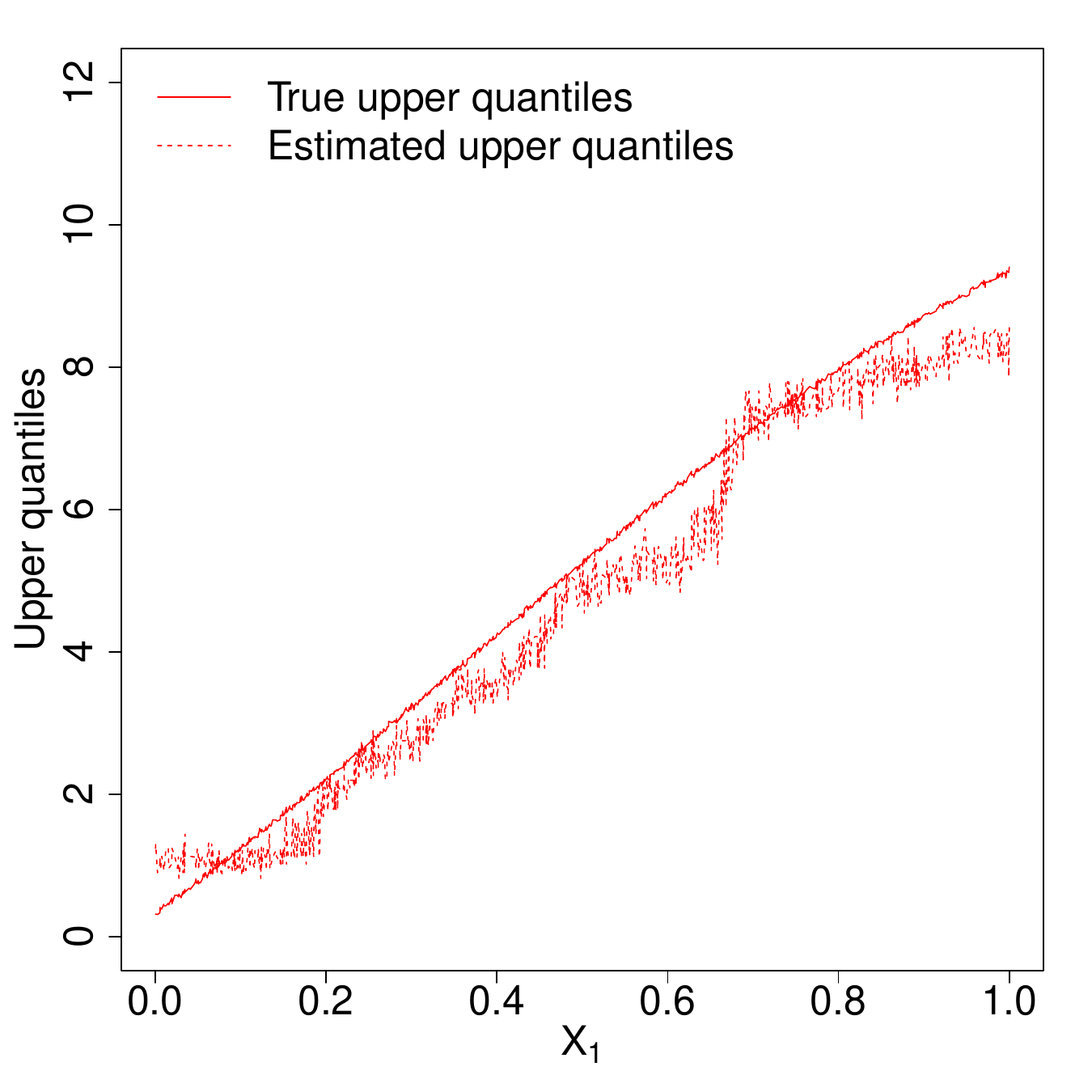}
 \end{center}
 \footnotesize\textbf{Note:} These plots compare the estimated nuisance parameters with their true values. The x-axis is $X_1$, the only observable covariate that affects the response rates. Top: trimming probability $q_0(\mathbf{x})/q_1(\mathbf{x})$; Bottom-left: lower quantile $y_{q\mathbf{x}}(\mathbf{x})$; Bottom-right: upper quantile $y_{1-q\mathbf{x}}(\mathbf{x})$.
\end{figure}

\begin{figure}[htp]
 \begin{center}
\includegraphics[width=.48\linewidth, height=.4\linewidth]{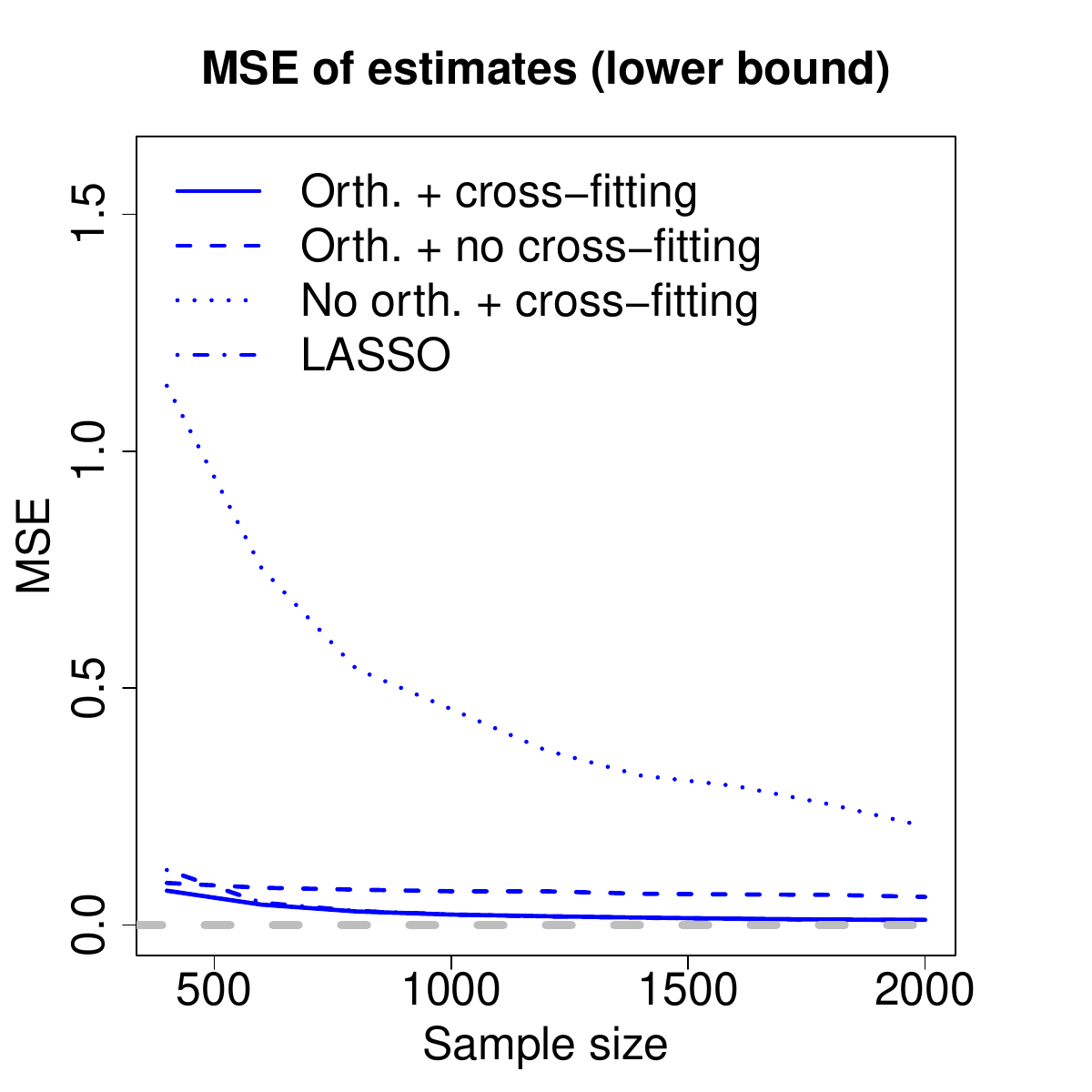}
\includegraphics[width=.48\linewidth, height=.4\linewidth]{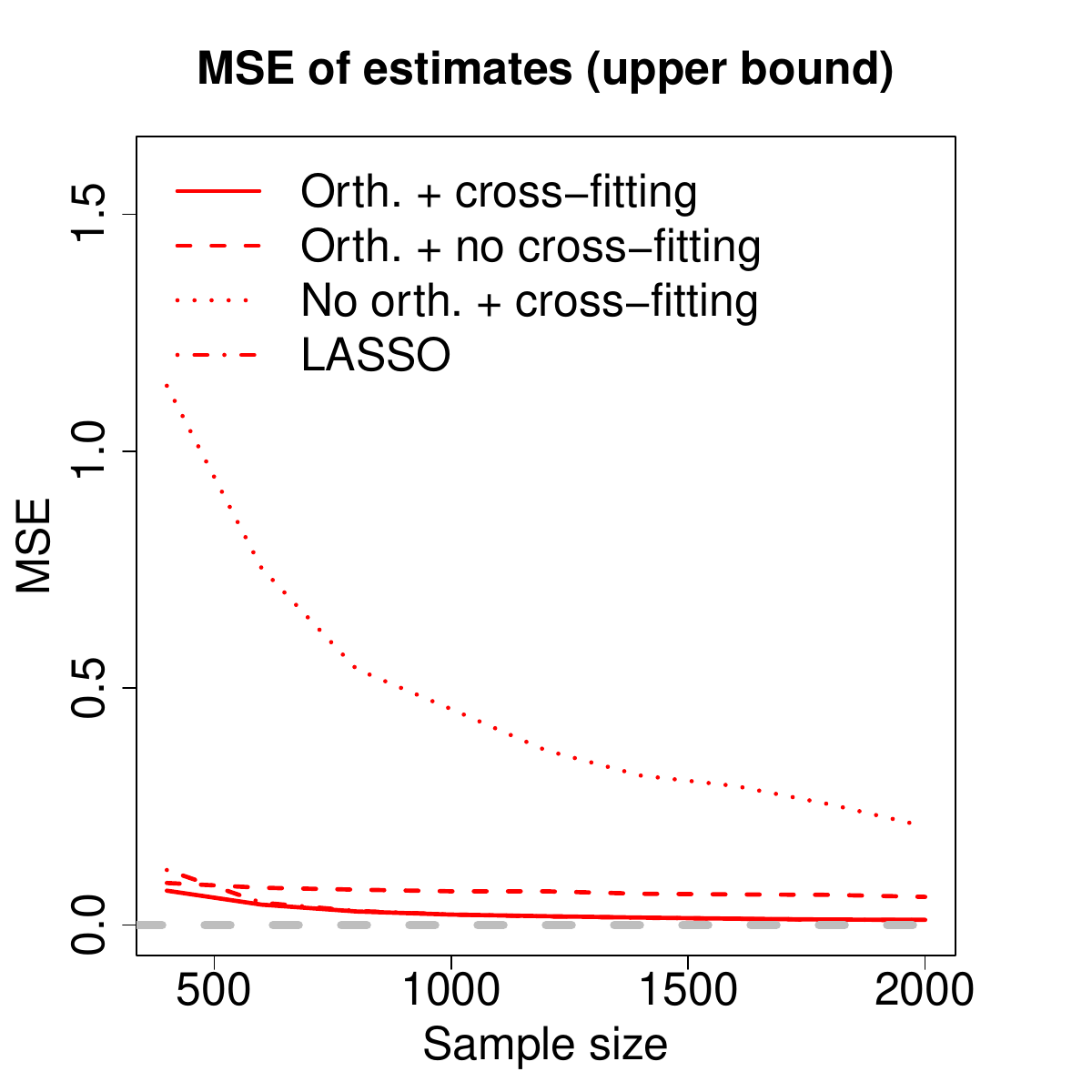} \\
\includegraphics[width=.48\linewidth, height=.4\linewidth]{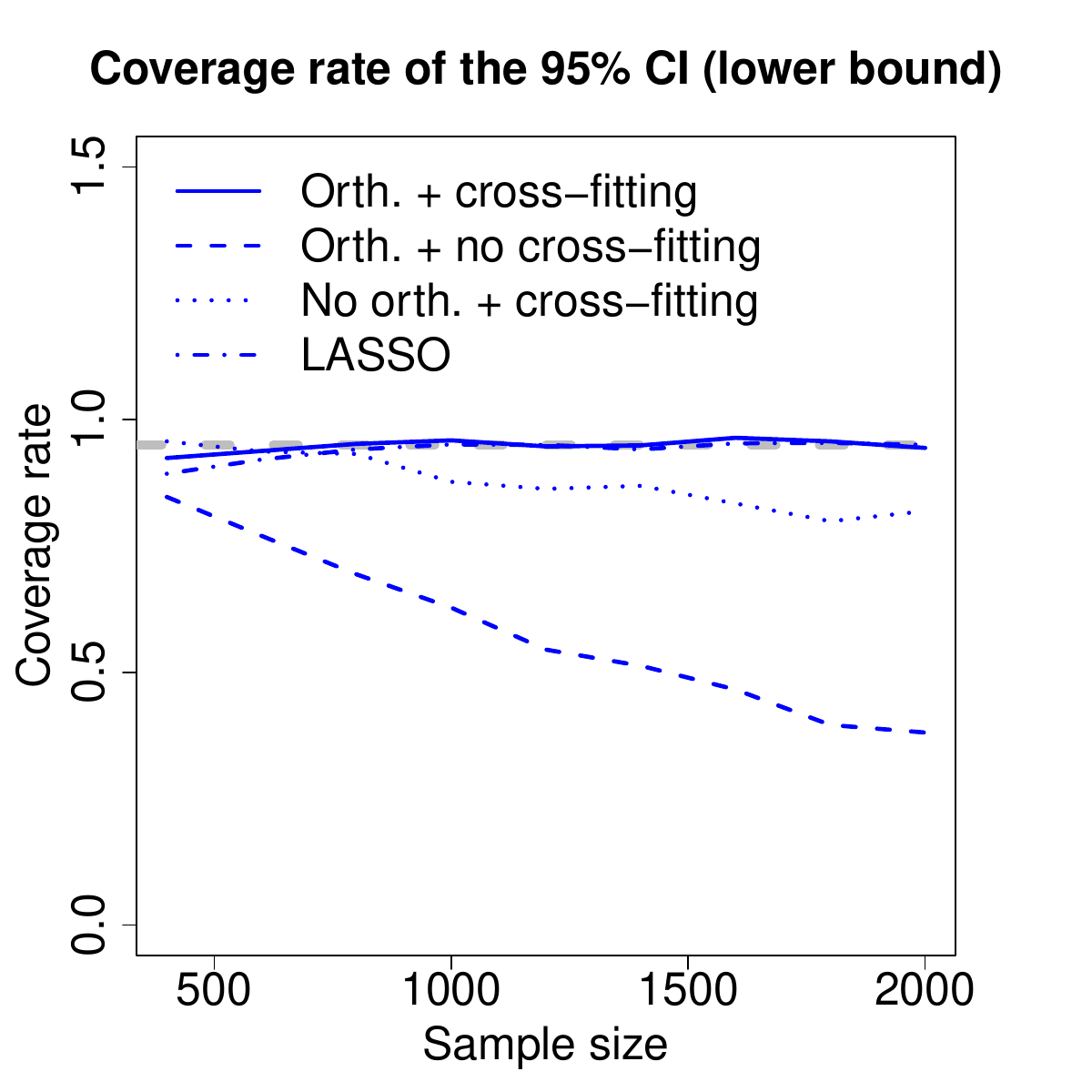}
\includegraphics[width=.48\linewidth, height=.4\linewidth]{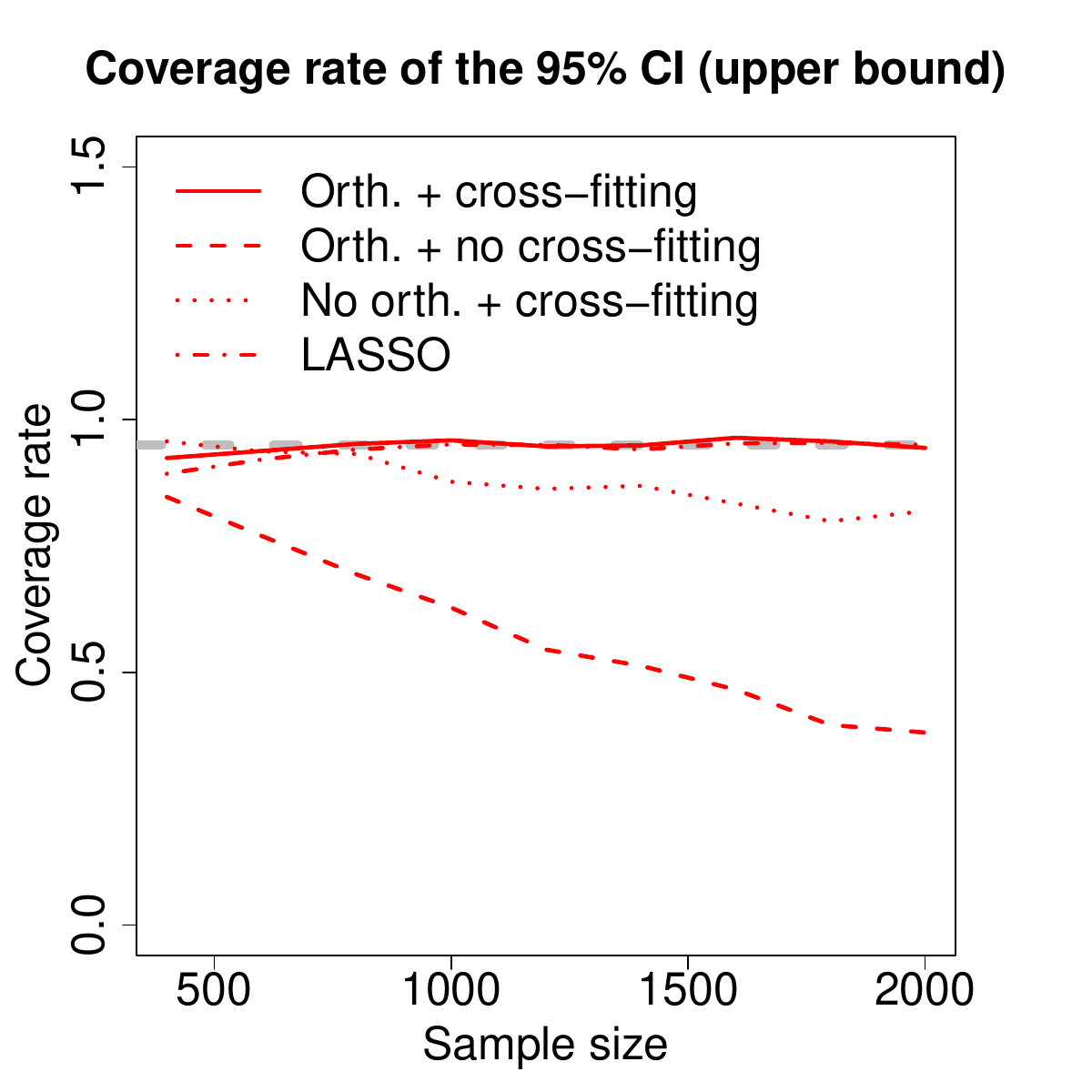}
\caption{Asymptotic Performance of Methods (Increasing Sample Size)}
\label{fig_simu_varyingP}
 \end{center}
 \footnotesize\textbf{Note:} Plots on the left show how the MSE and coverage rate of the lower bound estimates vary with sample sizes. Plots on the right show the same for the upper bound estimates. The solid lines represent the performance of the proposed method, with both Neyman orthogonalization and cross-fitting. The dashed lines represent the performance of the method with Neyman orthogonalization but without cross-fitting. The dotted lines represent the performance of the method with cross-fitting but without Neyman orthogonalization. The dash-dotted lines represent the performance of the LASSO method proposed by \cite{semenova2020better}. The gray lines on the bottom mark the nominal level of coverage, 95\%. We can see that the proposed method's MSE declines to zero as the sample size grows, and its coverage rate remains at the level of 95\%. It outperforms the LASSO-based approach under small samples, although the difference diminishes gradually. Without either Neyman orthogonalization or cross-fitting, the results are affected by the regularization bias hence do not lead to accurate inference.
\end{figure}

\begin{figure}[htp]
 \begin{center}
\includegraphics[width=.48\linewidth, height=.4\linewidth]{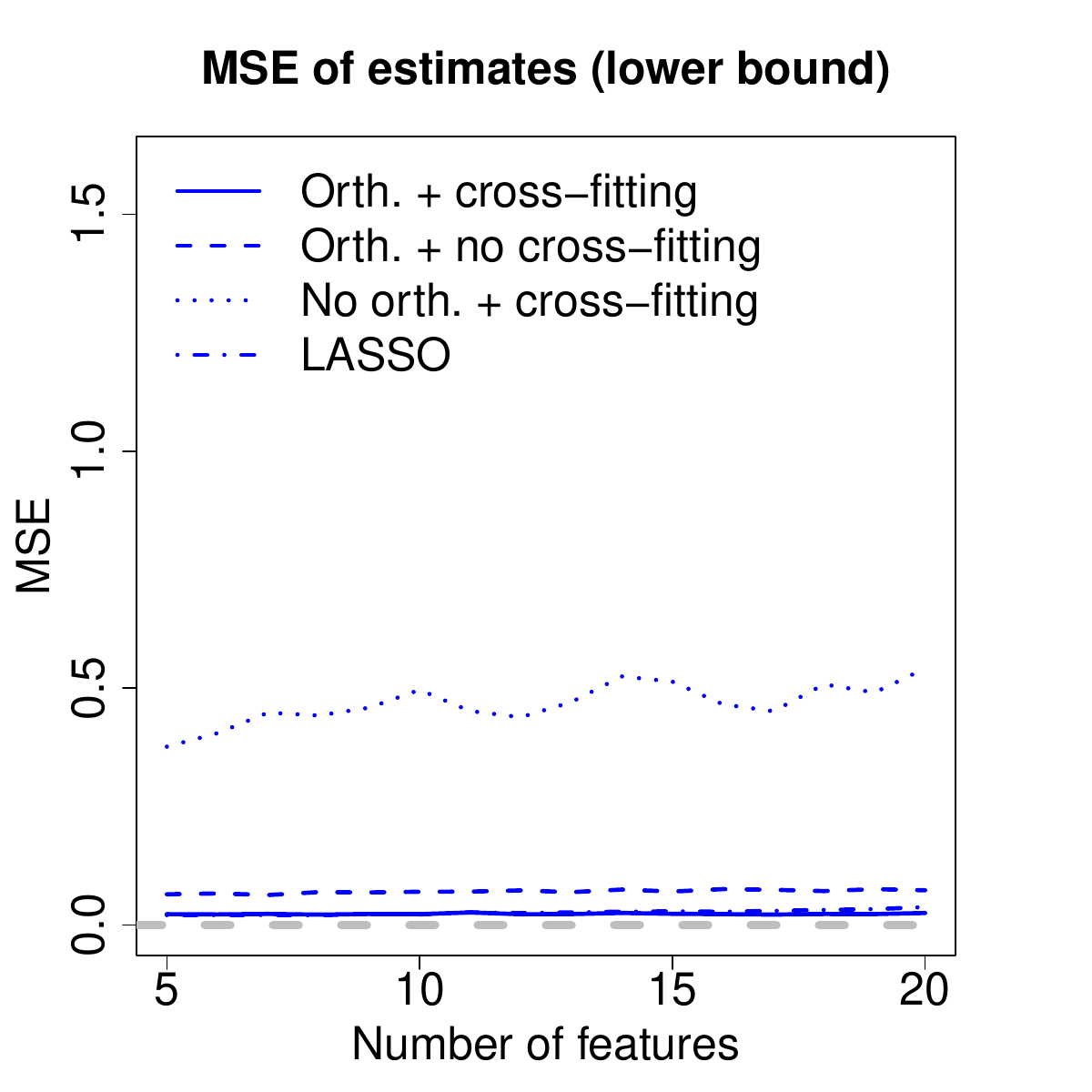}
\includegraphics[width=.48\linewidth, height=.4\linewidth]{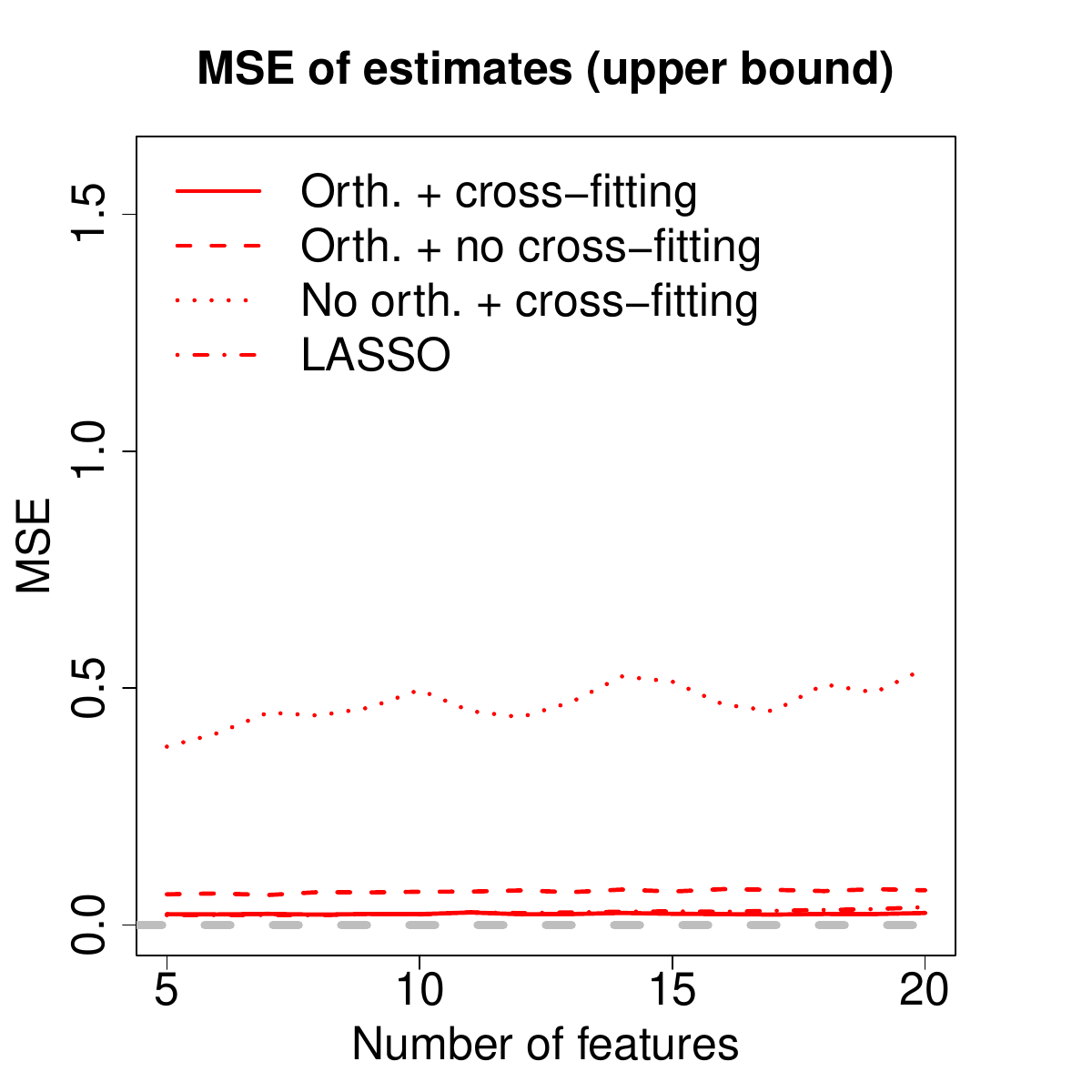} \\
\includegraphics[width=.48\linewidth, height=.4\linewidth]{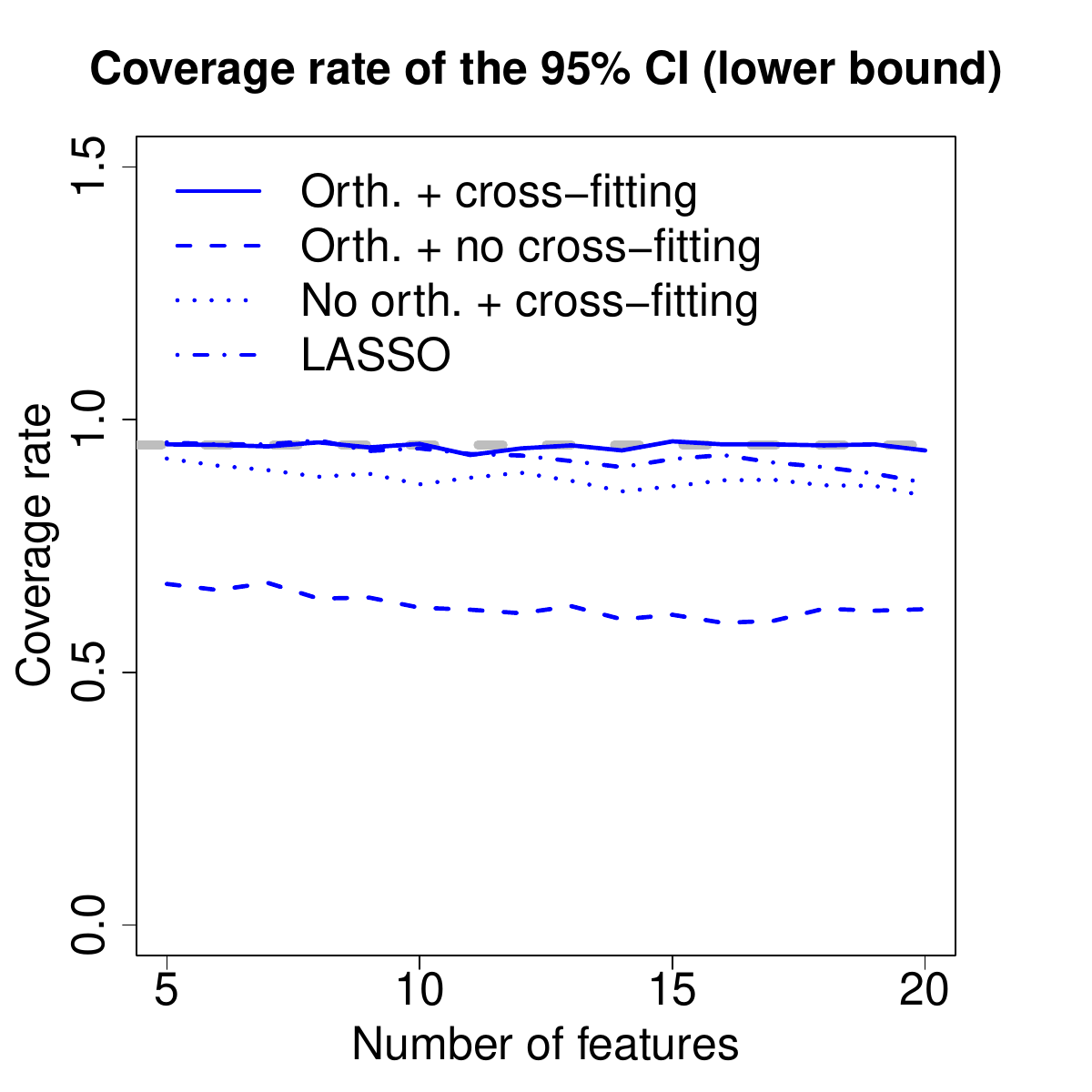}
\includegraphics[width=.48\linewidth, height=.4\linewidth]{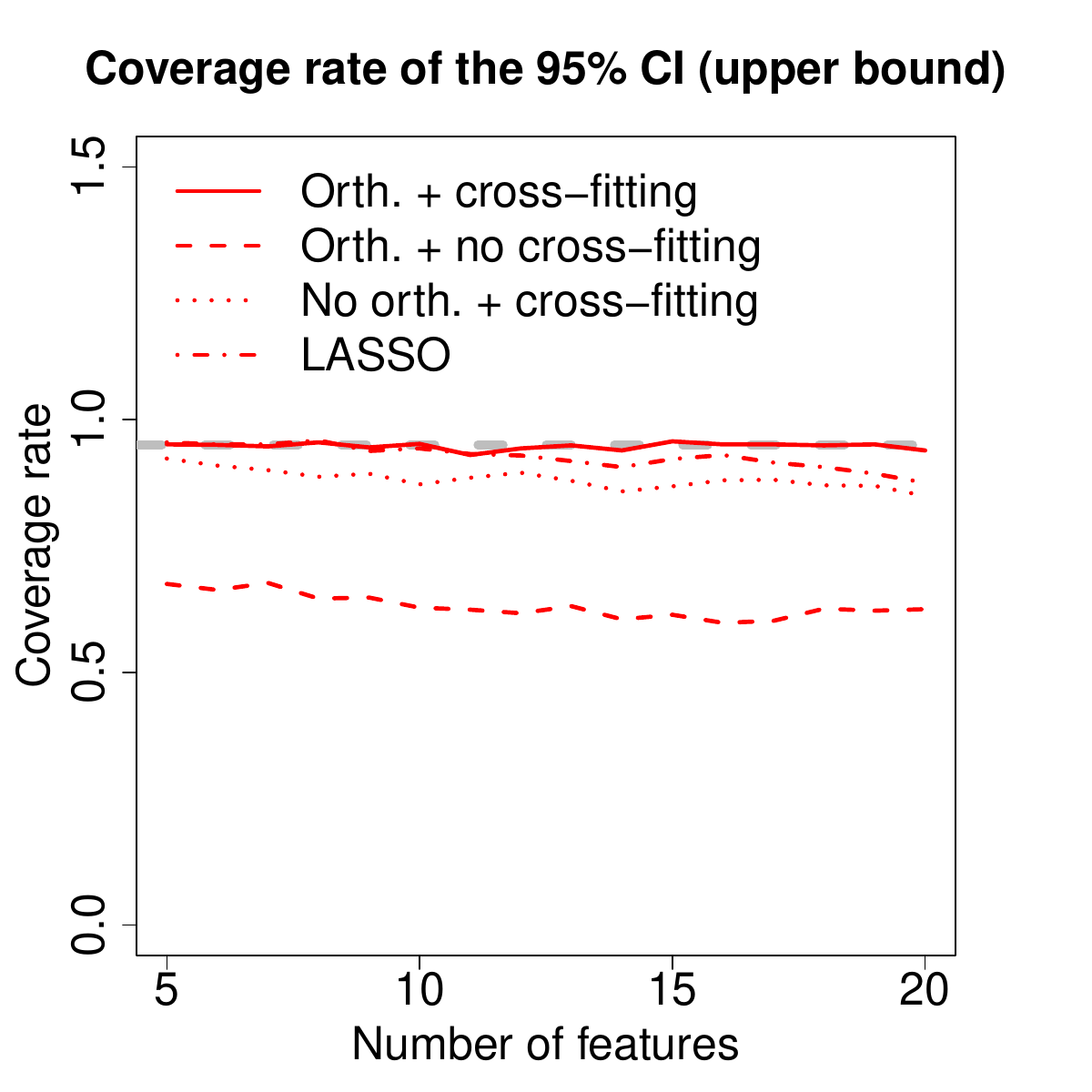}
\caption{Asymptotic Performance of Methods (Increasing Number of Covariates)}
\label{fig_simu_varyingP}
 \end{center}
 \footnotesize\textbf{Note:} Plots on the left show how the MSE and coverage rate of the lower bound estimates vary with the number of irrelevant covariates. Plots on the right show the same for the upper bound estimates. The solid lines represent the performance of the proposed method, with both Neyman orthogonalization and cross-fitting. The dashed lines represent the performance of the method with Neyman orthogonalization but without cross-fitting. The dotted lines represent the performance of the method with cross-fitting but without Neyman orthogonalization. The dash-dotted lines represent the performance of the LASSO method proposed by \cite{semenova2020better}. The gray lines on the bottom mark the nominal level of coverage, 95\%. We can see that the proposed method's MSE declines to zero and its coverage rate remains at the level of 95\%, even when the number of covariates is large. In contrast, the coverage rate of the LASSO-based approach falls below 95\% when there are more than 15 covariates. Without either Neyman orthogonalization or cross-fitting, the results are affected by the regularization bias hence do not lead to accurate inference.
\end{figure}

\begin{figure}[htp]
\caption{Performance with Unknown Propensity Scores}
\label{ap_fig_app_SK_cond}
 \begin{center}
\includegraphics[width=.7\linewidth, height=.6\linewidth]{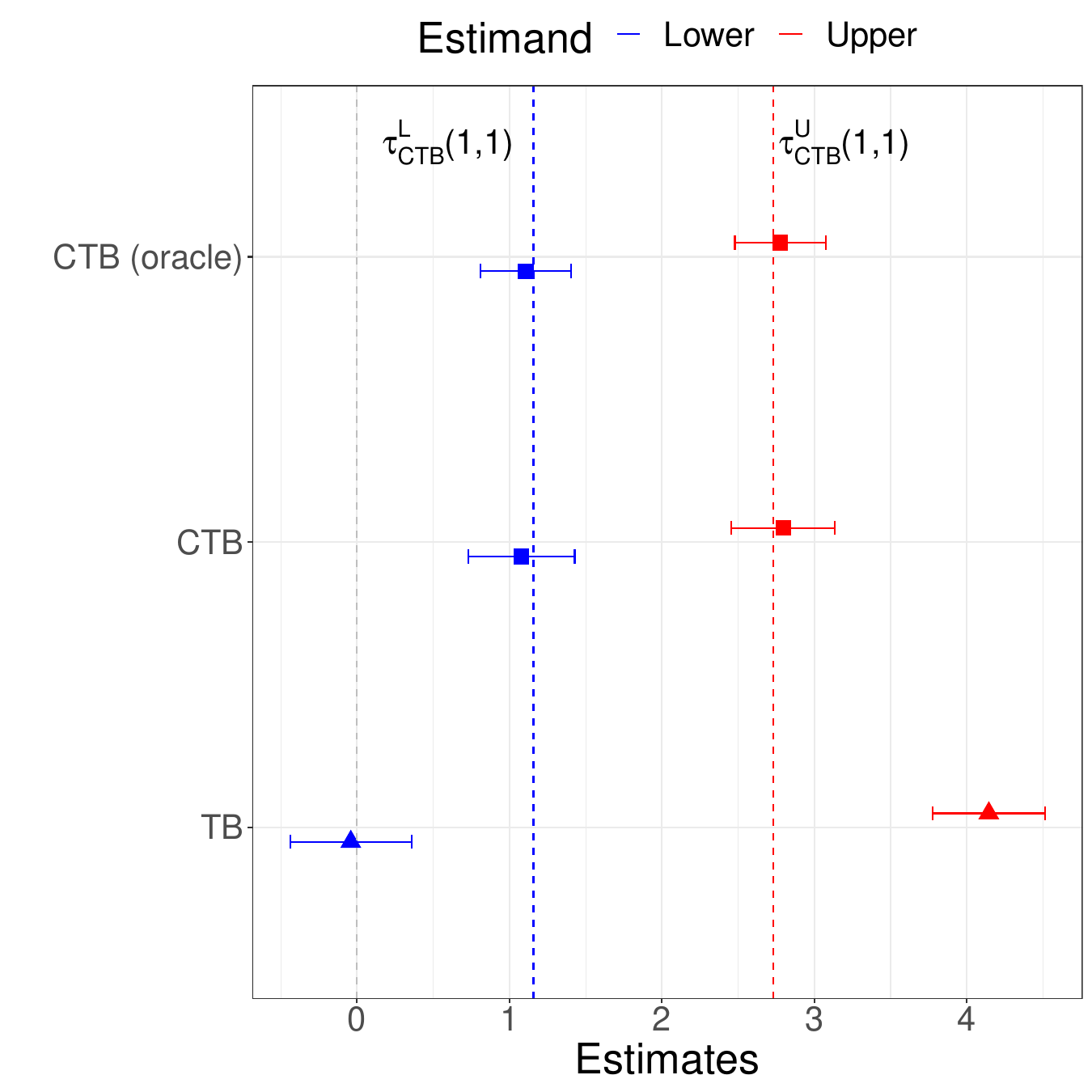}
 \end{center}
 \footnotesize\textbf{Note:} This plot shows the estimated covariates-tightened trimming bounds (CTB, in squares) and basic trimming bounds (TB, in triangles) in simulation, when the propensity scores have to be estimated from data. It also compares the results with the estimated covariates-tightened trimming bounds under the true propensity scores (CTB (oracle), in squares). Lower bound estimates are in blue and upper bounds estimates are in red. The segments represent the 95\% confidence intervals for the estimates. The results are averaged across 1,000 treatment assignments. The red and blue dotted lines mark the true values of the CTB. The black dotted line represents the true ATE and the black dash-dotted line represents the true ATE for the always-responders. We can see that the averages of the CTB estimates are close to their true values even when the propensity scores are unknown.
\end{figure}

\newpage
\subsection{Extra results from applications}\label{appx:B2}
\begin{figure}[htp]
\caption{Trimming Probability Estimates in Santoro and Broockman (2022)}
\label{ap_fig_app_SK_cond_q}
 \begin{center}
\includegraphics[width=.5\linewidth, height=.4\linewidth]{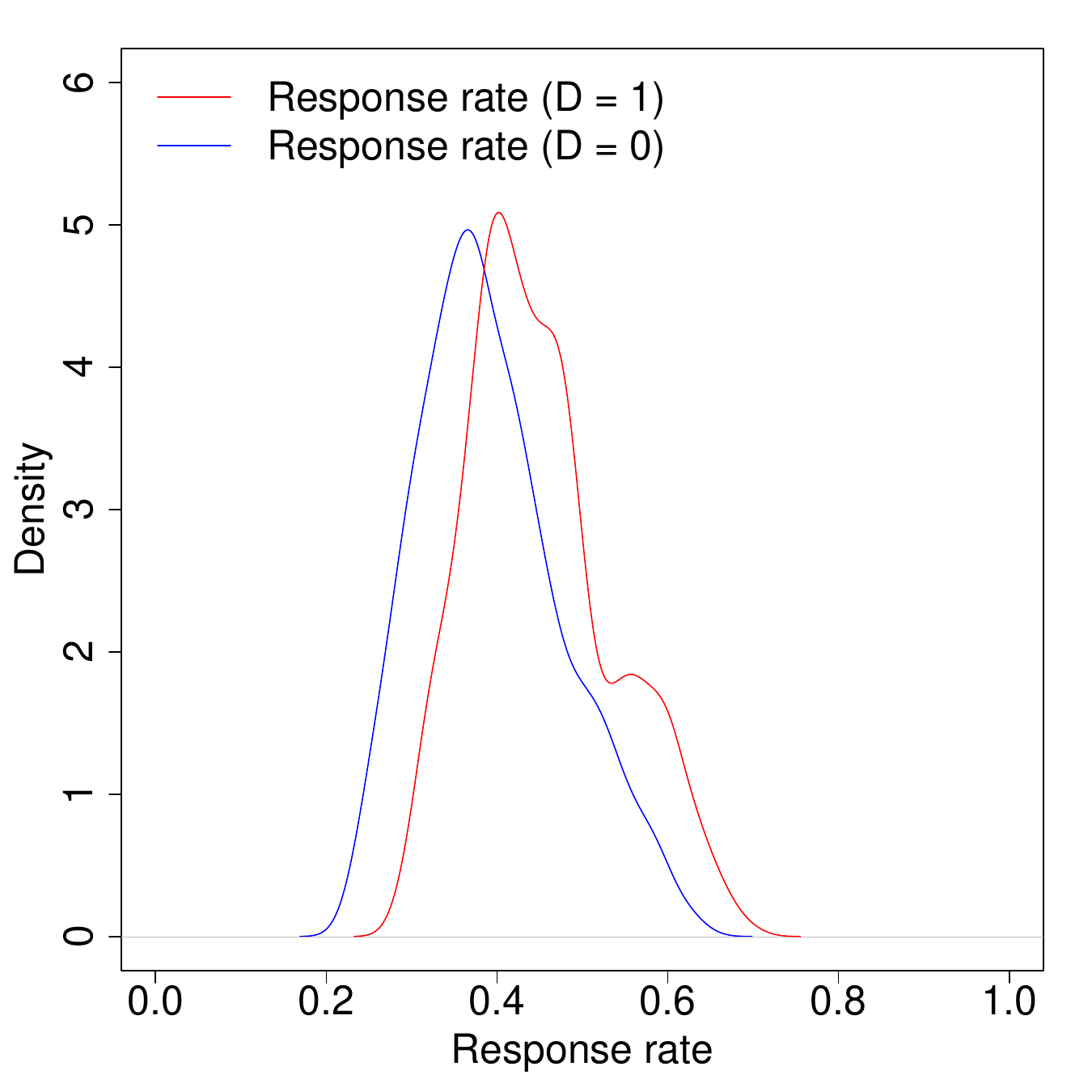}
\includegraphics[width=.5\linewidth, height=.4\linewidth]{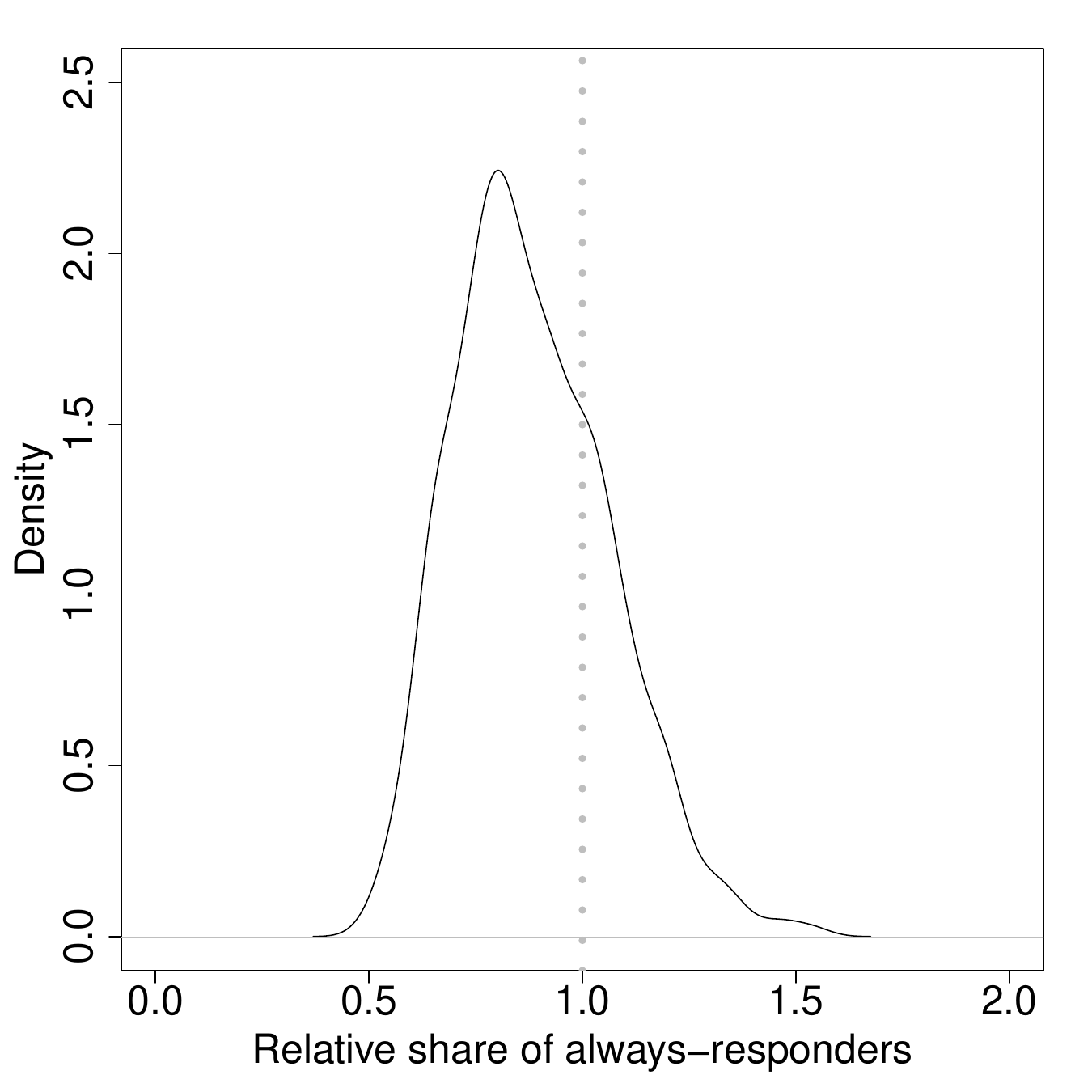}
 \end{center}
 \footnotesize\textbf{Note:} The left plot presents the distribution of the estimated conditional response rate under treatment ($q_1(\mathbf{x})$, in red) or under control ($q_0(\mathbf{x})$, in blue). The right plot presents the distribution of the estimated conditional trimming probability ($q(\mathbf{x}) = q_0(\mathbf{x}) / q_1(\mathbf{x})$).
\end{figure}

\newpage 
\begin{figure}[htp]
\caption{Conditional Bounds in Santoro and Broockman (2022)}
\label{ap_fig_app_SK_cond}
 \begin{center}
\includegraphics[width=.8\linewidth, height=.4\linewidth]{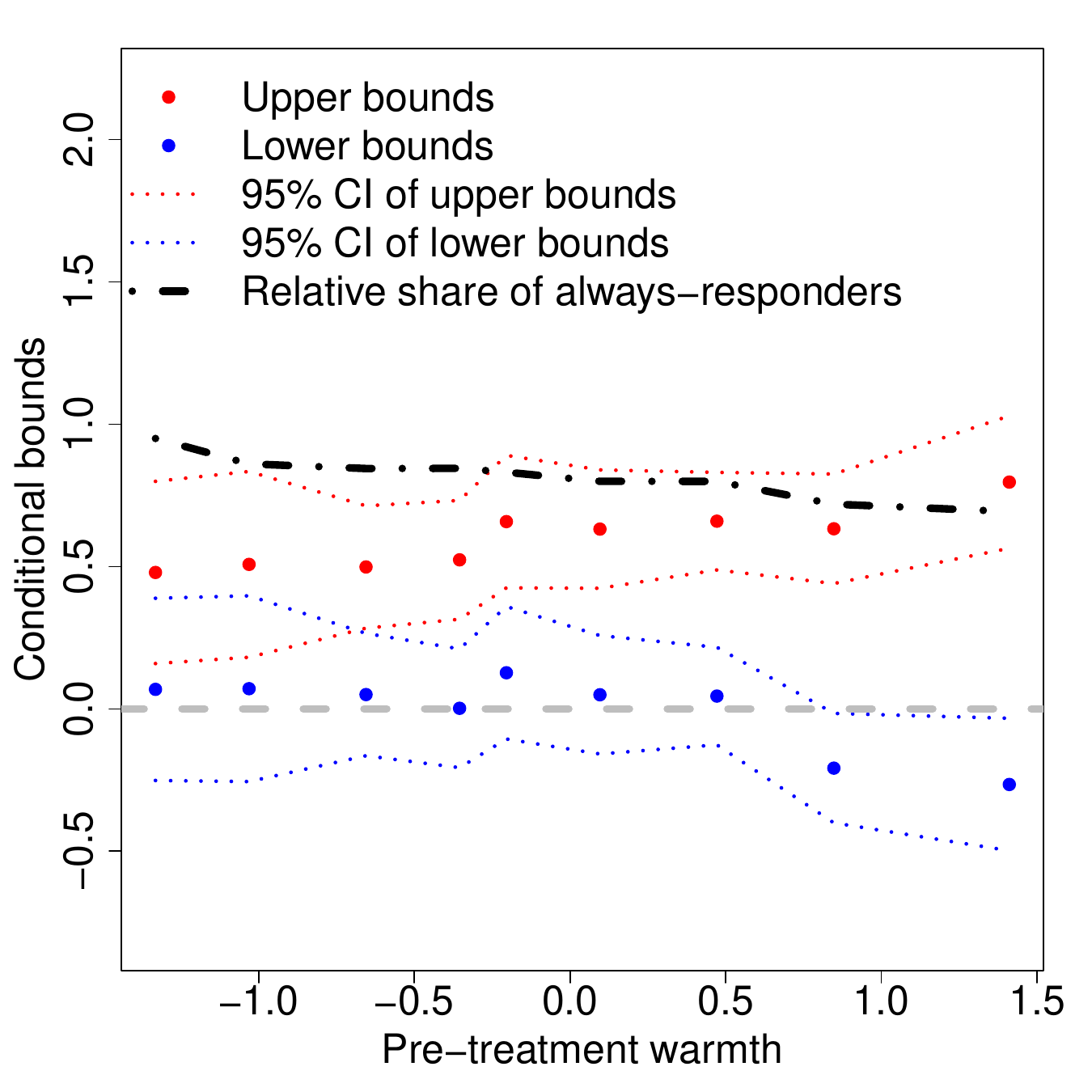}
\includegraphics[width=.8\linewidth, height=.4\linewidth]{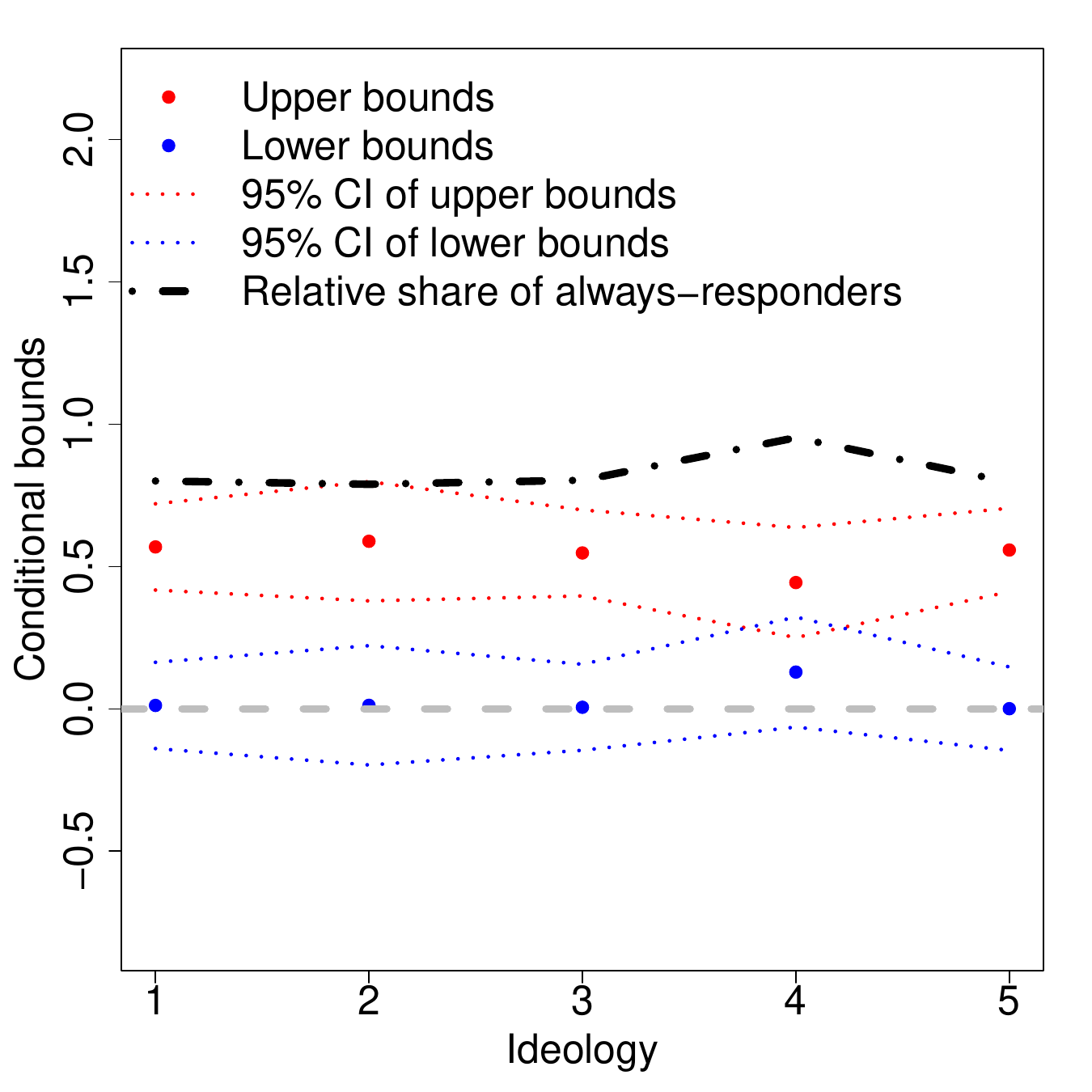}
 \end{center}
 \footnotesize\textbf{Note:} These plots show the estimates of the conditional covariates-tightened trimming bounds across 9 observations whose demographic attributes are fixed at the sample mean or mode while their pre-treatment warmth toward outpartisan voters or ideology varies across its quantiles. The red and blue dots represent upper and lower bound estimate, respectively. The dotted lines around them are the 95\% confidence intervals. The dash-dotted curve depicts the conditional trimming probability.
\end{figure}

\newpage 
\begin{figure}[htp]
\caption{Trimming Probability Estimates in Blattman and Annan (2010)}
\label{ap_fig_app_BA_cond_q}
 \begin{center}
\includegraphics[width=.5\linewidth, height=.4\linewidth]{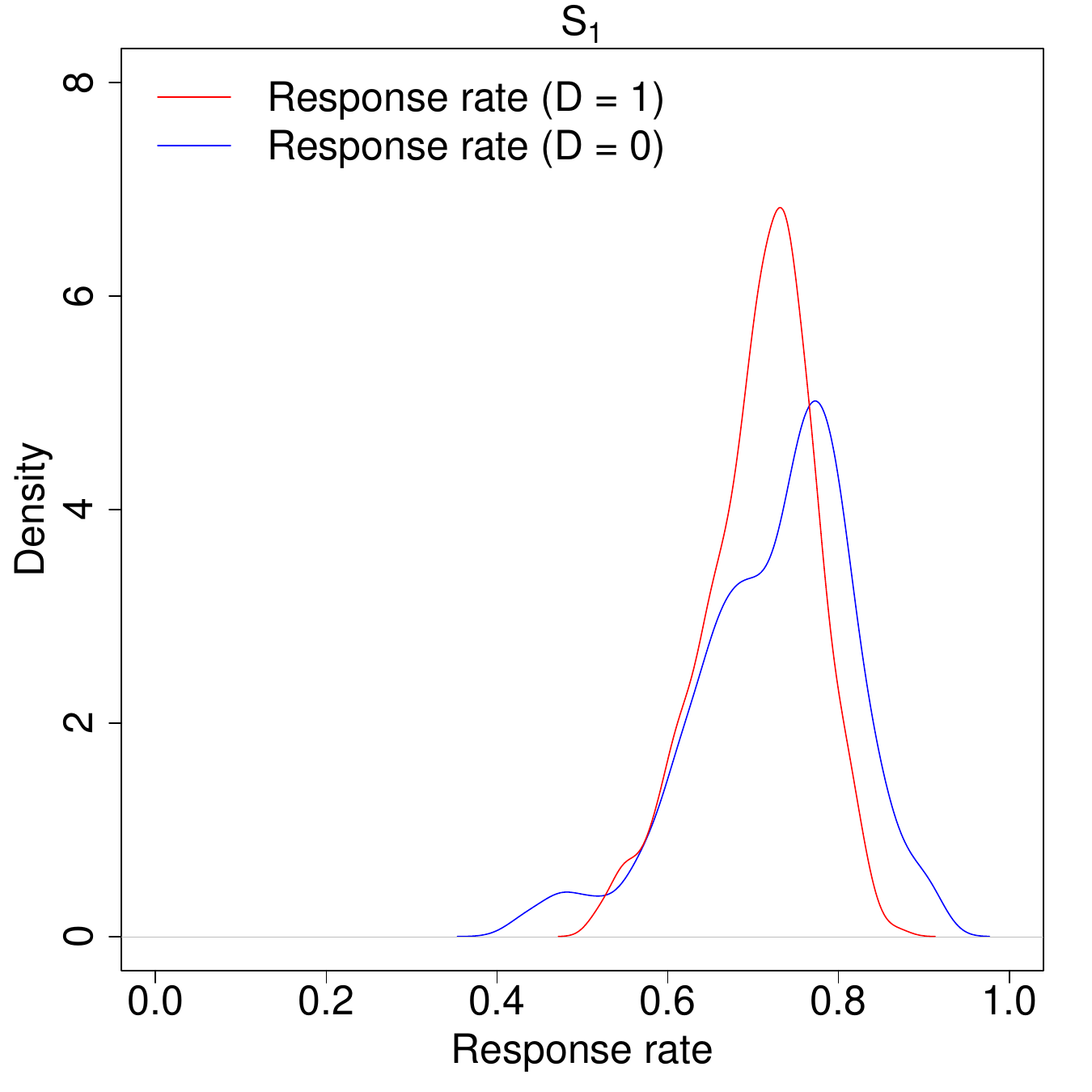}
\includegraphics[width=.5\linewidth, height=.4\linewidth]{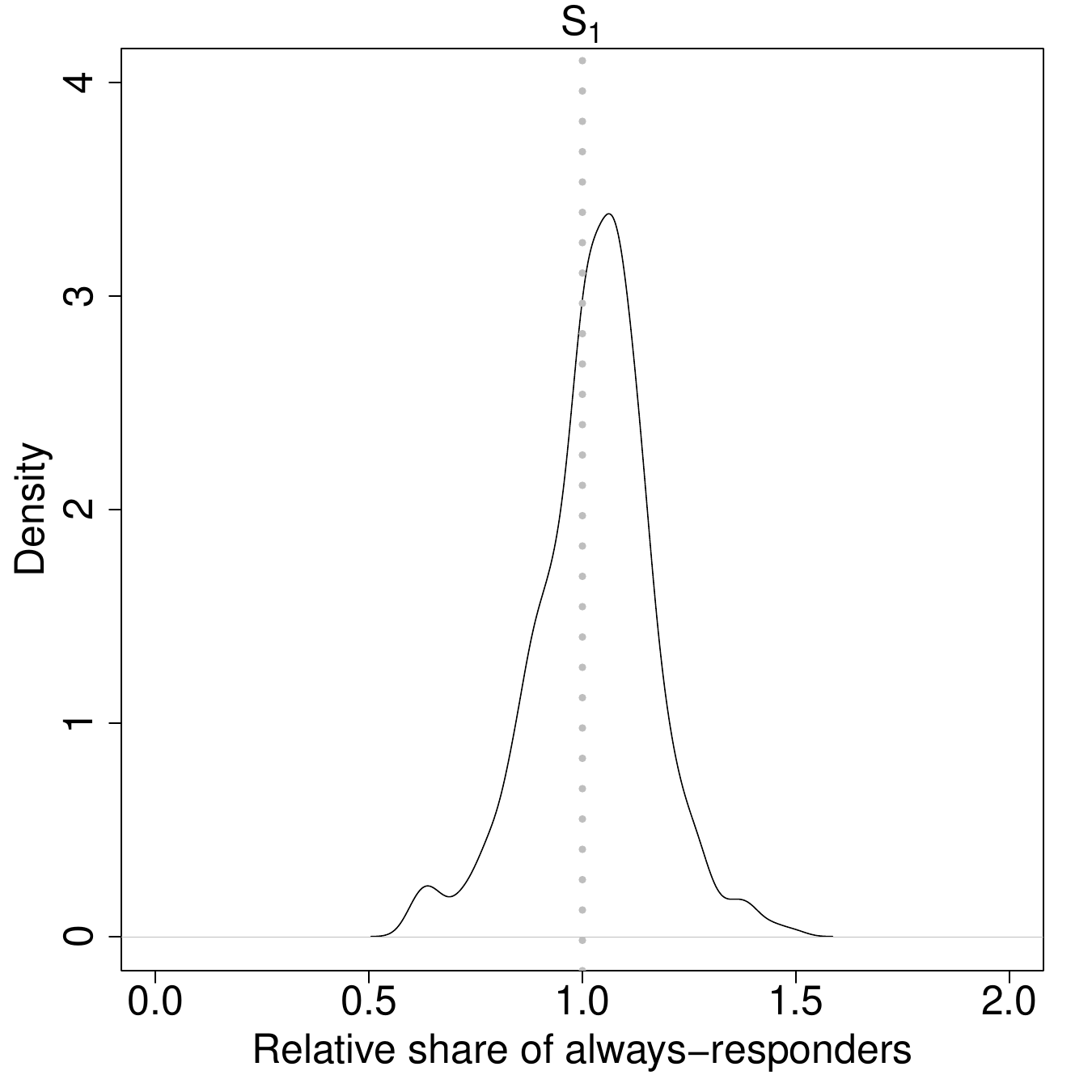}\\
\includegraphics[width=.5\linewidth, height=.4\linewidth]{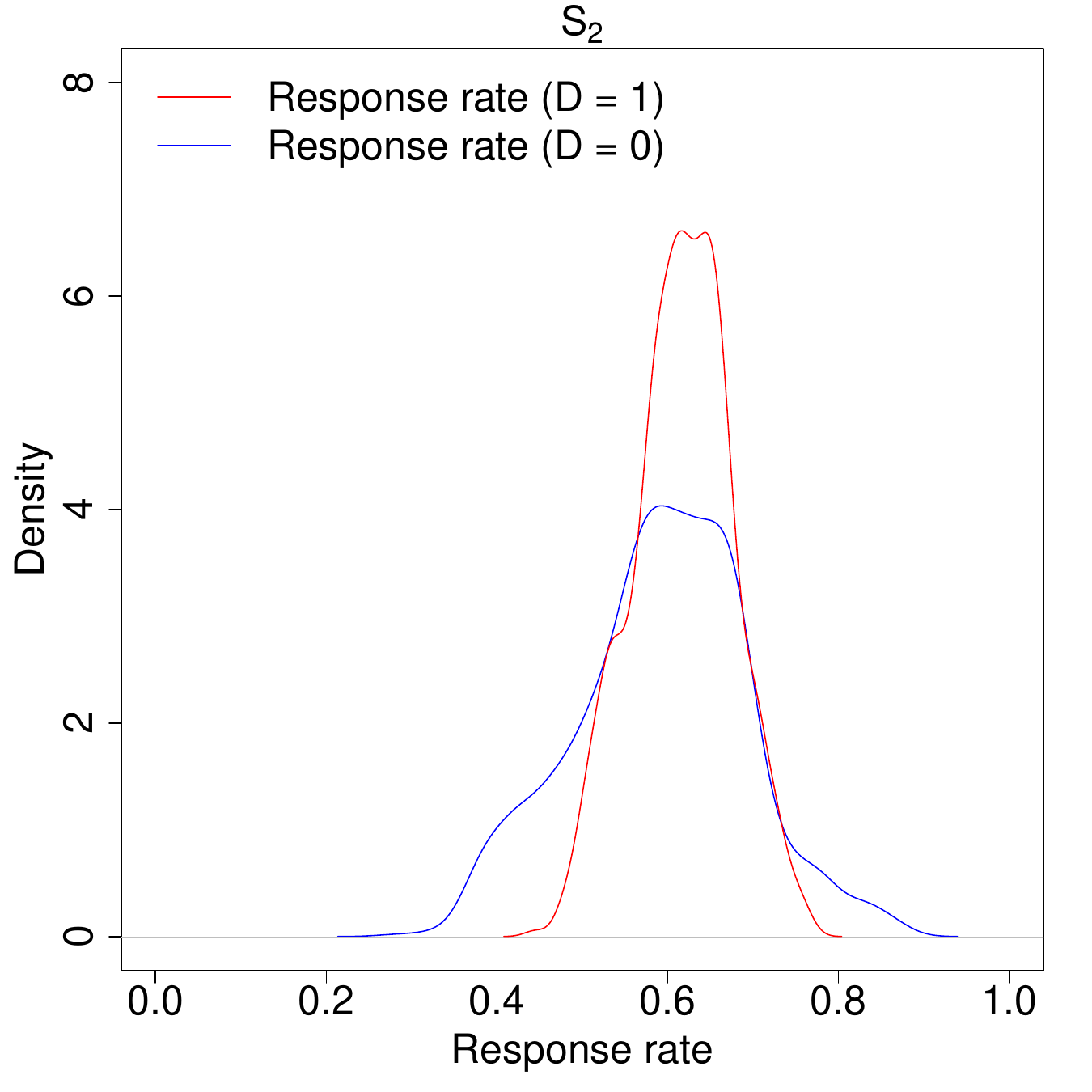}
\includegraphics[width=.5\linewidth, height=.4\linewidth]{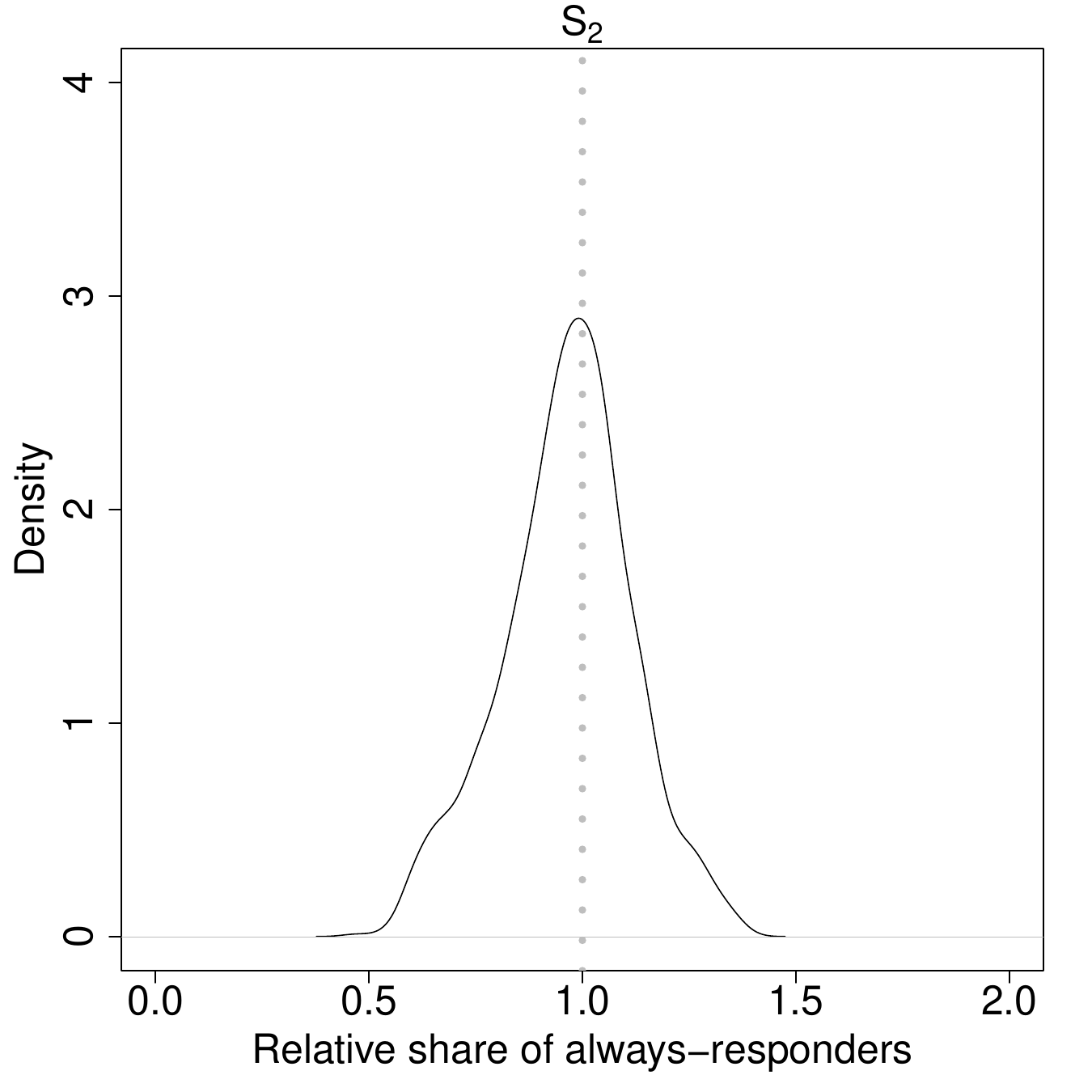}
 \end{center}
 \footnotesize\textbf{Note:} The left plots present the distribution of the estimated conditional response rate under treatment ($q_1(\mathbf{x})$, in red) or under control ($q_0(\mathbf{x})$, in blue). The right plots present the distribution of the estimated conditional trimming probability ($q(\mathbf{x}) = q_0(\mathbf{x}) / q_1(\mathbf{x})$).
\end{figure}

\newpage 
\begin{table}[!htbp]\centering
\footnotesize
   \begin{threeparttable}
   \caption{Summary of Results} 
 \label{tab_lee_res}  
\begin{tabular}{@{\extracolsep{5pt}} lccccc} 
\\[-1.8ex]\hline 
\hline \\[-1.8ex]  
\\[-1.8ex] 
& CTB & TB (mono.) & TB (cond. mono.) & OLS & IPW\\ 
\hline \\[-1.8ex]   
& \multicolumn{5}{c}{Panel A: Simulation (ATE for the always-responders = 1.650)} \\
\hline \\[-1.8ex]   
$\hat \tau^L(1, 1)$ & 1.012 & -0.142 & - & - & - \\ 
& [0.721, 1.303] & [-0.536, 0.253] & - & - & - \\  
$\hat \tau^U(1, 1)$ & 2.665 & 4.013 & - & - & - \\ 
& [2.373, 2.957] & [3.632, 4.395] & - & - & - \\  
$\hat \tau$ & - & - & - & 1.925 & 1.915 \\ 
& - & - & - & [1.555, 2.295] & [1.584, 2.246] \\  
N & 737/1,000 & 737/1,000 & 737/1,000 & 737/1,000 & 737/1,000 \\
 \hline \\[-1.8ex] 
 & \multicolumn{5}{c}{Panel B: Santoro and Broockman (2022)} \\
\hline \\[-1.8ex]  
$\hat \tau^L(1, 1)$ & 0.272 & 0.022 & -0.758 & - & - \\ 
& [0.060, 0.484] & [-0.213, 0.257] & [-0.996, -0.521]  & - & - \\  
$\hat \tau^U(1, 1)$ &  0.355 & 0.466 & 1.086 & - & - \\ 
& [0.156, 0.554] & [0.227, 0.706] & [0.913, 1.258] & - & - \\  
$\hat \tau$ & - & - & - & 0.340 & - \\ 
& - & - & - & [0.218, 0.463] & - \\ 
N & 469/986 & 469/986 & 469/986 & 469/986 & 469/986 \\
 \hline \\[-1.8ex] 
 & \multicolumn{5}{c}{Panel C: Blattman and Annan (2010), Education} \\
\hline \\[-1.8ex]  
$\hat \tau^L(1, 1)$ &  -1.366 & -1.200 & -1.678 & - & - \\ 
& [-1.883 , -0.849] & [-1.540, -0.860] & [-2.025, -1.331] & - & - \\  
$\hat \tau^U(1, 1)$ &  0.093 & 0.109 & 0.271  & - & - \\ 
& [-0.465, 0.652] & [-0.253, 0.471] & [-0.136, 0.678] & - & - \\  
$\hat \tau$ & - & - & - & -0.799 & -0.588 \\ 
& - & - & - & [-1.229, -0.369] & [-1.106, -0.070] \\ 
N & 870/1,216 & 870/1,216 & 870/1,216 & 870/1,216 & 870/1,216 \\
 \hline \\[-1.8ex] 
  & \multicolumn{5}{c}{Panel D: Blattman and Annan (2010), Distress} \\
\hline \\[-1.8ex]  
$\hat \tau^L(1, 1)$ & 0.477 & 0.424 & -0.842 & - & - \\ 
& [0.005, 0.948] & [0.109, 0.740] & [-1.163, -0.522] & - & - \\  
$\hat \tau^U(1, 1)$ & 0.627 & 0.645 & 1.937 & - & - \\ 
& [0.177, 1.077] & [0.147, 1.143] & [1.615, 2.260] & - & - \\   
$\hat \tau$ & - & - & - & 0.376 & 0.545 \\ 
& - & - & - & [-0.070, 0.821] & [0.209, 0.881] \\ 
N & 741/1,216 & 741/1,216 & 741/1,216 & 741/1,216 & 741/1,216 \\
 \hline \\[-1.8ex] 
\end{tabular}%
    \begin{tablenotes}
      \small
      \item \textit{Note:} This table summarizes the results from our simulation study (panel A), replication of \cite{santoro2022promise} (Panel B), and replication of \cite{blattman2010consequences} (Panels C and D). Numbers in the brackets are the endpoints of the 95\% confidence intervals. The first number in the last row of each panel is the sample size with observed outcome values and the second number is the total sample size.
    \end{tablenotes}
\end{threeparttable}
\end{table} 

\newpage 
\subsection{Revisiting the Job Corps experiment}\label{appx:B3}
We apply our method to estimate the intensive margin effect of job training on wages using data from the Job Corps program, the original example investigated by \citet{lee2009training}. It is one of the largest federal-funded job training programs in the U.S., in which disadvantaged youth aged 16-24 years were offered residence at a Job Corps center and about 1100 hours of vocational and academic training. In the mid-1990s, the program conducted lottery-based admission to assess its effect on the basis of randomization. The sample consists of 9,145 Job Corps applicants. Among them, 5,977 applicants in the control group were embargoed from the program for 3 years, while the remaining ones (the treatment group) could enroll in Job Corps as usual. The data includes the lottery outcome, hours worked, and wages for 208 consecutive weeks after randomization of each applicant, as well as their background information, such as educational attainment, employment record, recruiting experience, household composition, income, drug use, and arrest record.
 
To be consistent with \citet{lee2009training}, we focus the intensive margin effect generated by the program, which is defined as the effect on wages among those who would be employed regardless of the access to Job Corps.\footnote{This is distinct from the extensive margin effect, which applies to those whose employment status is affected by the program.} Such an effect is not identified by randomization alone, since whether applicants are employed may result from other unobservable factors.  \citet{lee2009training} constructed bounds for the intensive margin effect on the logarithm of wage in week 208 under the assumption of monotonic selection: treated applicants are more likely to be employed. Yet as pointed out by \citet{semenova2020better}, this assumption may not always hold in the data. We replicate Lee's analysis using all the available covariates in the dataset and focus on the wage level in both week 90 and week 208.
\begin{figure}
\caption{Covariate-tightened Trimming Bounds in JobCorps}
\label{fig_job_agg}
 \begin{center}
\includegraphics[width=.49\linewidth, height=.7\linewidth]{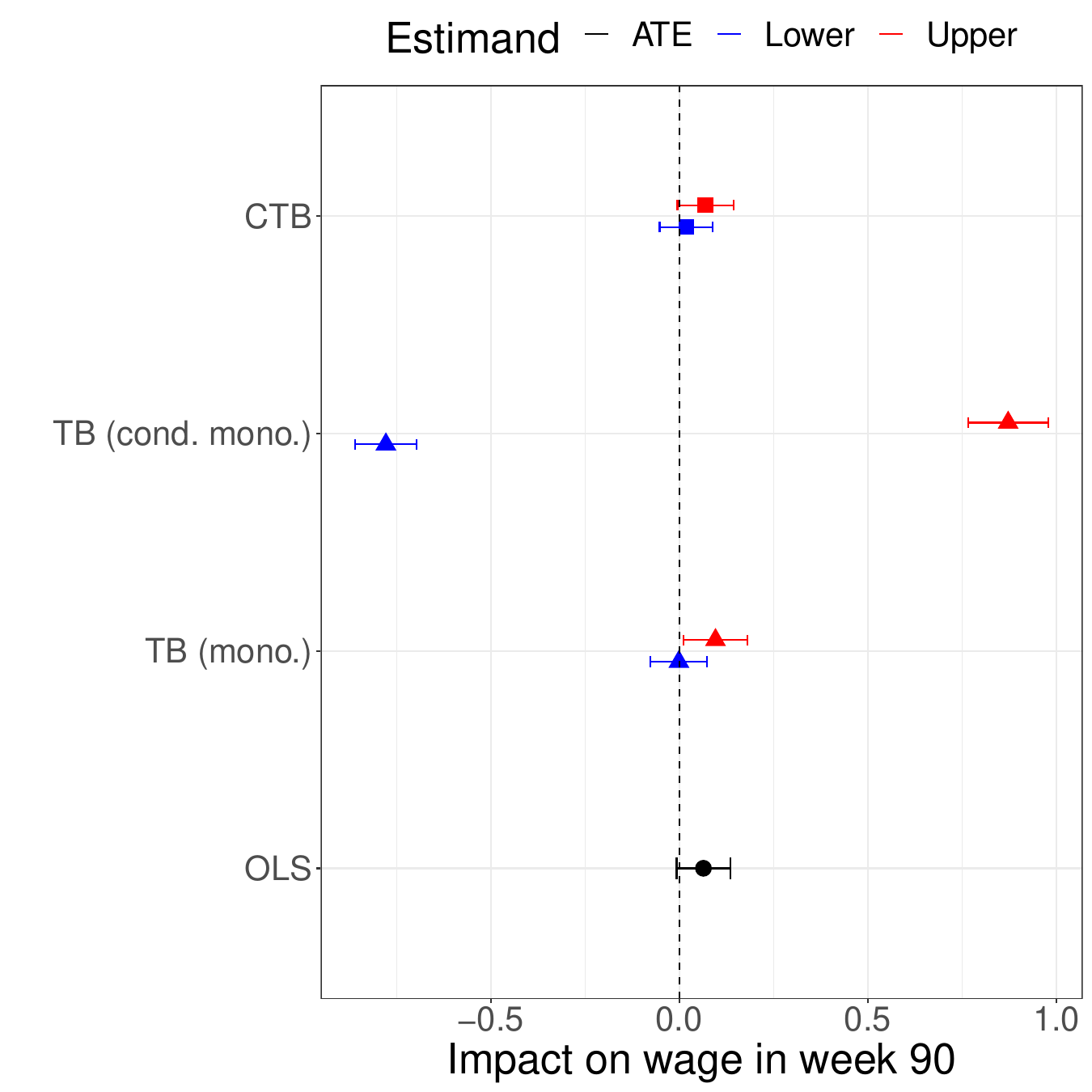}
\includegraphics[width=.49\linewidth, height=.7\linewidth]{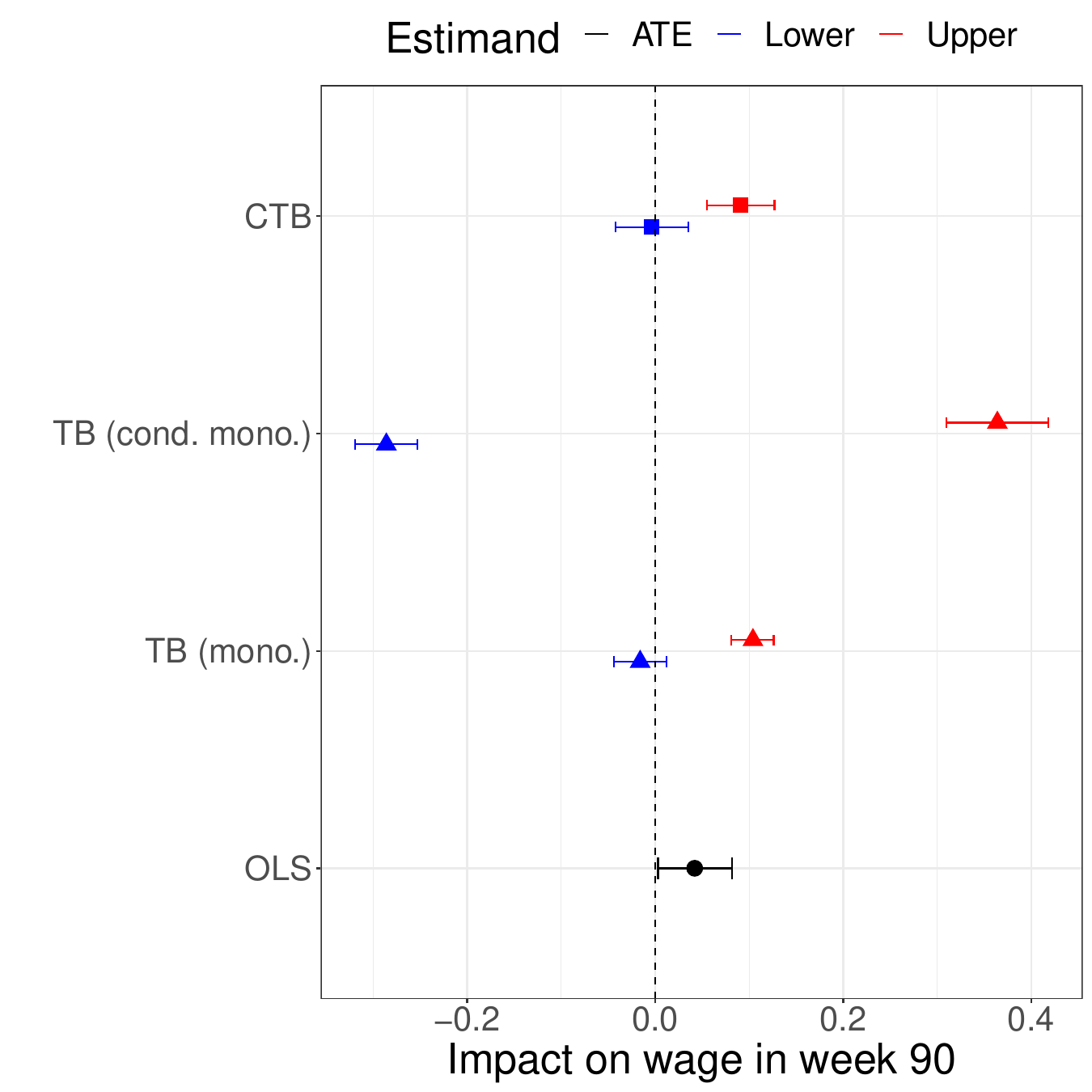}
 \end{center}
 \footnotesize\textbf{Note:} In the left plot, the outcome variable is weekly wage in week 90, while in the right one it is weekly wage in week 208. From top to bottom, the plots show the estimated covariates-tightened trimming bounds (CTB, in squares), basic trimming bounds under monotonic selection (TB (mono.), in triangles) and under conditionally monotonic selection (TB (cond. mono.), in triangles), and the ATE estimate using OLS on the non-missing sample (OLS, in circles). Lower bound estimates are in blue and upper bounds estimates are in red. The ATE estimates are in black. The segments represent the 95\% confidence intervals for the estimates.
\end{figure}

The replication results are presented in Figure \ref{fig_job_agg}. Again, we find that the assumption of monotonic selection is violated and has to be replaced by conditionally monotonic selection. It turns out that $\hat{\mathcal{X}}^{-}$ includes about one third of all the observations. Under conditionally monotonic selection, the basic trimming bounds become much wider and cover zero. However, under the same assumption, our bounds are still informative. They are even narrower than the Lee bounds under monotonic selection. Both the lower bound and the upper bound are positive for wage level in week 90, although the former is not statistically significant at the level of 5\%.\footnote{\citet{semenova2020better} reports a significantly positive estimate for the lower bound. But her analysis relies on covariates that are unavailable in the original dataset.} Overall, the replication confirms the necessity of allowing for conditionally monotonic selection and the superiority of the proposed method to the basic trimming bounds.

\newpage 
\subsection{Replication results of Kalla and Broockman (2022)}\label{appx:B4}
To demonstrate how to apply our method when the outcome variable is binary, we replicate the results in \citet{kalla2022outside}, a field experiment that aims to test the impacts of televised issue advertisements on American voters. We focus on one of the treatments in the experiment (``Prosperous Future'' and ``MariCruz/Kandy''), which consists of two immigration ads, one about how a white woman changed her view on the immigration system and the other showing that undocumented workers also pay taxes. The outcome we are interested in is whether subjects recall the content of the ads at the end of the experiment. In the baseline survey, there were 20,937 subjects in the control group and the treatment group of interest. But the number dropped to 10,650 in the endline survey. We calculate the basic trimming bounds and the aggregated covariate-tightened trimming bounds using the moment conditions presented in Section 4 of the main text. The standard error estimates of the basic trimming bounds are obtained from bootstrap.\footnote{It is straightforward to show that the nuisance parameters, $\nu(\mathbf{x}) = (q_0(\mathbf{x}), q_1(\mathbf{x}), \xi(\mathbf{x}))$, converge to a joint normal distribution, using the same arguments in \cite{lee2009training}. Hence, the target parameters are also asymptotically normal, which justifies using bootstrap for inference.}

\begin{figure}[htp]
\caption{Covariate-tightened Trimming Bounds in Kalla and Broockman (2022)}
\label{ap_fig_app_SK_cond}
 \begin{center}
\includegraphics[width=.6\linewidth, height=.44\linewidth]{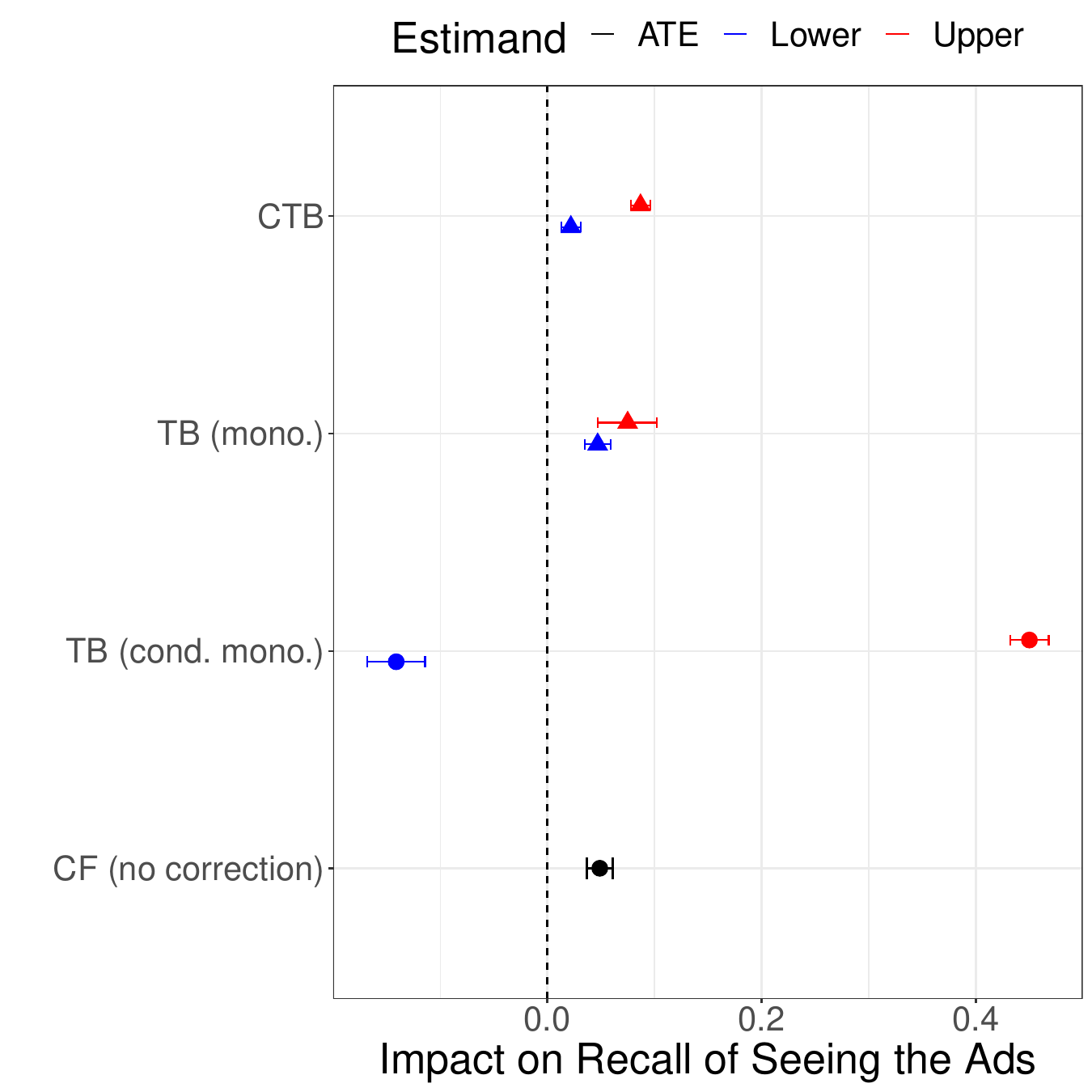}
 \end{center}
 \footnotesize\textbf{Note:} From top to bottom, we have the estimated covariates-tightened trimming bounds (CTB, in squares), basic trimming bounds under monotonicity (TB (mono.), in triangles) and under conditionally monotonic selection (TB (cond. mono.), in triangles), and the ATE estimate using OLS on the non-missing sample (OLS, in circles). Lower bound estimates are in blue and upper bounds estimates are in red. The ATE estimate is in black. The segments represent the 95\% confidence intervals for the estimates.
\end{figure}

\clearpage

\bibliography{cttb}